\documentclass[journal]{IEEEtran}
\usepackage{amsfonts}
\usepackage{mathrsfs}
\usepackage{amsmath}
\usepackage{graphicx}
\usepackage{amssymb}
\usepackage{epstopdf}
\usepackage{enumerate}
\usepackage{longtable}
\usepackage{cite}
\usepackage{stmaryrd}
\usepackage{hyperref}
\newtheorem{corollary}{Corollary}[section]

\newtheorem{lemma}{Lemma}[section]
\newtheorem{theorem}{Theorem}[section]

\newtheorem{proposition}{Proposition}[section]

\newtheorem{construction}{Construction}[section]

\newcommand{\A}{{\mathcal A}}
\newcommand{\B}{{\mathcal B}}
\newcommand{\C}{{\mathcal C}}
\newcommand{\D}{{\mathcal D}}

\newcommand{\G}{{\mathcal G}}

\newcommand{\V}{{\mathcal{V}}}

\newcommand{\vu}{{\sf u}}
\newcommand{\vv}{{\sf v}}

\newcommand{\vzero}{{\sf 0}}

\newcommand{\supp}{{\rm supp}}

\newcommand{\bbZ}{{\mathbb{Z}}}

\title{Quaternary Constant-Composition Codes with Weight Four and Distances Five or Six}

\author{Mingzhi Zhu~and~Gennian Ge
\thanks{The research of G. Ge was supported by  the National Outstanding Youth  Science Foundation of China
under Grant No.~10825103, National Natural Science Foundation of
China under Grant No.~61171198, Specialized Research Fund for the
Doctoral Program of Higher Education,  and Zhejiang Provincial
Natural Science Foundation of China under Grant No.~D7080064.}
\thanks{M. Zhu and G. Ge (Corresponding author. Email: {\tt gnge@zju.edu.cn}) are with the Department of Mathematics, Zhejiang University,
    Hangzhou 310027, Zhejiang, People's Republic of China.}
}

\begin{document}

\maketitle

\begin{abstract}

\boldmath The sizes of optimal constant-composition codes of
weight three have been determined by Chee, Ge and Ling with  four
cases in doubt. Group divisible codes played an important role in
their constructions. In this paper, we study the problem of
constructing optimal quaternary constant-composition codes with
Hamming weight four and minimum distances five or six through  group
divisible codes and Room square approaches. The problem is solved leaving only five lengths
undetermined. Previously, the results on the sizes of such
quaternary constant-composition codes   were scarce.

\end{abstract}


\section{Introduction}

\PARstart{c}onstant-composition codes (CCCs) are a special type of
constant-weight codes  (CWCs) which  are important in coding theory. The
class of constant-composition codes includes the important
permutation codes and have attracted recent interest due to their
numerous applications, such as in determining the zero error
decision feedback capacity of discrete memoryless channels
\cite{Telatar:1990}, multiple-access communications
\cite{D'yachkov:1984}, spherical codes for modulation
\cite{Ericson:1995}, DNA codes
\cite{King:2003,Milenkovic:2006,Chee:2008}, powerline communications
\cite{Chu:2004,Colbourn:2004}, frequency hopping \cite{Chu:2006}, frequency permutation arrays \cite{Huczynska:2006}, and
coding for bandwidth-limited channels \cite{Costello:2007}.

Systematic study began in late 1990's \cite{Bogdanova:1998,Svanstrom:1999,Bogdanova_Ocetarova:1998}. Today, various methods have been applied to
the problem of determining the maximum size of a
constant-composition code, such as computer search methods \cite{Bogdanova:2003}, packing designs
\cite{Ding:2005(2),Chu:2006,Ding:2006,Yan:2008,YinTang:2008,Wen:2009,Yan:2009,Huczynska:2010},
tournament designs \cite{YinYan:2008}, polynomials and nonlinear
functions
\cite{Ding:2005,Ding:2005(3),Chu:2006,Ding:2008,YDing:2008},
difference triangle sets \cite{Chee:2010}, PBD-closure methods
\cite{Yeow:2007,Yeow:2008} and some other methods \cite{Svanstrom:2000,Luo:2003}.

In the paper of Svanstr{\"o}m et al. \cite{Svanstrom:2002}, some
methods for providing upper and lower bounds on the maximum size
$A_3(n,d,\overline{w})$ of a ternary code with length $n$, minimum
Hamming distance $d$, and constant composition $\overline{w}$ were
presented. The sizes of optimal ternary constant-composition codes
with weight three have been determined completely by Chee, Ge and
Ling in \cite{Yeow:2008}. The sizes of optimal ternary
constant-composition codes with  weight four and distance five
have been determined completely by Gao and Ge in \cite{Gao:2010}.

The sizes of optimal quaternary  constant-composition codes with
weight three have been determined almost completely by Chee, Ge
and Ling in \cite{Yeow:2008} with  four lengths in doubt. Recently,
 the problem of determining the sizes for optimal quaternary constant-composition codes with  weight four and
distance seven has been settled by Chee, Dau, Ling and Ling in
\cite{Chee:2010}.  In this paper, we will concentrate our
attention on quaternary CCCs  with   weight four and  distances
five or six. The problem is solved leaving only five lengths
undetermined for the case of distance five. Previously, the
results on the sizes of such quaternary constant-composition codes
were scarce.

%
%
%

\section{Preliminaries}
\vskip 10pt

\subsection{Definitions and Notations}
The set of integers $\{i,i+1,\ldots,j\}$ is denoted by $[i,j]$. The
ring $\bbZ/q\bbZ$ is denoted by $\bbZ_q$. The notation $\Lbag \cdot
\Rbag$ is used for multisets.

All sets considered in this paper are finite if not obviously
infinite. If $X$ and $R$ are finite sets, $R^X$ denotes the set of
vectors of length $|X|$, where each component of a vector $\vu\in
R^X$ has value in $R$ and is indexed by an element of $X$, that is,
$\vu=(\vu_x)_{x\in X}$, and $\vu_x\in R$ for each $x\in X$. A
{\it $q$-ary} code of {\it length} $n$ is a set $\C \subseteq \bbZ_q^X$ for some
$X$ with size $n$. The elements of $\C$ are called {\it codewords}. The
{\it Hamming norm} or the {\it Hamming weight} of a vector $\vu\in\bbZ_q^X$ is
defined as $\|\vu\|=| \{x\in X: \vu_x\not=0\}|$. The distance
induced by this norm is called the {\it Hamming distance}, denoted $d_H$,
so that $d_H(\vu,\vv)=\| \vu-\vv \|$, for $\vu,\vv\in\bbZ_q^X$. The
{\it composition} of a vector $\vu\in \bbZ_q^X$ is the tuple
$\overline{w}=[w_1,\ldots,w_{q-1}]$, where $w_j=|\{x\in X:
\vu_x=j\}|$. For any two vectors $\vu,\vv\in\bbZ_q^X$, define their
{\it support} as $\supp(\vu,\vv)=\{x\in X:\vu_x\not=\vv_x\}$. We write
$\supp(\vu)$ instead of $\supp(\vu,\vzero)$ and also call
$\supp(\vu)$ the support of $\vu$.

A code $\C$ is said to have minimum distance $d$ if $d_H(\vu,\vv)\geq d$ for
all distinct $\vu,\vv\in\C$. If $\|\vu\|=w$ for every codeword $\vu\in\C$,
then $\C$ is said to be of (constant) {\it weight $w$}. A $q$-ary code
$\C$ has {\it constant composition $\overline{w}$} if every codeword in
$\C$ has composition $\overline{w}$. A $q$-ary code of length $n$,
distance $d$, and constant composition $\overline{w}$ is referred to
as an $(n,d,\overline{w})_q$-code. The maximum size of an
$(n,d,\overline{w})_q$-code is denoted as $A_q(n,d,\overline{w})$ and
the $(n,d,\overline{w})_q$-codes achieving this size are called
{\it optimal}. Note that the following operations do not affect distance
and weight properties of an $(n,d,\overline{w})_q$-code:
\begin{enumerate}[(i)]
\item reordering the components of $\overline{w}$, and
\item deleting zero components of $\overline{w}$.
\end{enumerate}
Consequently, throughout this paper, we restrict our attention to
those compositions $\overline{w}=[w_1,\ldots,w_{q-1}]$, where
$w_1\geq \cdots\geq w_{q-1}\geq 1$.

Suppose $\vu\in\bbZ_q^X$ is a codeword of an
$(n,d,\overline{w})_q$-code, where
$\overline{w}=[w_1,\ldots,w_{q-1}]$. Let $w=\sum_{i=1}^{q-1} w_i$.
We can represent $\vu$ equivalently as a $w$-tuple $\langle
a_1, a_2, \ldots, a_w\rangle\in X^w$, where
\begin{align*}
\vu_{a_1}=\cdots=\vu_{a_{w_1}} &=1, \\
\vu_{a_{w_1+1}}=\cdots=\vu_{a_{w_1+w_2}} &=2, \\
\vdots \\
\vu_{a_{\sum_{i=1}^{q-2} w_i+1}} =\cdots=\vu_w &= q-1.
\end{align*}
Throughout this paper, we shall often represent codewords of
constant-composition codes in this form. This has the advantage of
being more succinct and more flexible in manipulation.

\subsection{General Bounds}
\vskip 10pt

\begin{lemma}[Chee et al. \cite{Yeow:2008}]\label{bound1}
\begin{equation*}
\begin{split}
& A_q(n,d,[w_1,\ldots,w_{q-1}])=\\
& \begin{cases} \binom{n}{\sum_{i=1}^{q-1}w_i}
\binom{\sum_{i=1}^{q-1}w_i}{w_1,\ldots,w_{q-1}}, & \text{if $d \leq 2$}\\
\left\lfloor \frac{n}{\sum_{i=1}^{q-1}w_i} \right\rfloor, & \text{if
$d=2
\sum_{i=1}^{q-1}w_i$}\\
1, & \text{if $d \geq 2\sum_{i=1}^{q-1}w_i+1$}.
\end{cases}
\end{split}
\end{equation*}
\end{lemma}
\vskip 10pt

The following Johnson-type bound has been proven for
constant-composition codes.

\vskip 10pt
\begin{lemma}[Svanstr{\"o}m et al. \cite{Svanstrom:2002}]\label{bound2}
\begin{equation*}
A_q(n,d,[w_1,\ldots,w_{q-1}]) \leq \frac{n}{w_1}A_q(n-1,d,[w_1-1,
\ldots ,w_{q-1}]).
\end{equation*}
\end{lemma}
\vskip 10pt

Moreover, we have the following results.

\vskip 10pt
\begin{lemma}[Chee et al. \cite{Yeow:2007}] \label{bound3}
\begin{equation*}
\begin{split}
& A_q(n,d,[w_1,\ldots,w_{q-1}]) \leq\\
& \begin{cases} \left\lfloor \frac{n}{w_1} \left\lfloor
\frac{n-1}{\sum_{i=1}^{q-1}w_i-1} \right\rfloor \right\rfloor &
\text{if
$d=2\sum_{i=1}^{q-1}w_i-3$}\\
\left\lfloor \frac{n}{w_1} \left\lfloor \frac{n-1}{w_1-1}
\right\rfloor \right\rfloor & \text{if $d=2\sum_{i=1}^{q-1}w_i-2$}.\\
\end{cases}
\end{split}
\end{equation*}
\end{lemma}
\vskip 10pt

\begin{corollary}
\begin{equation*}
\begin{split}
A_4(n,5,[2,1,1]) \leq n \left\lfloor \frac{n-1}{2}\right\rfloor\\
\end{split}
\end{equation*}
\begin{equation*}
A_4(n,6,[2,1,1]) \leq \left\lfloor \frac{n}{2} \left\lfloor
\frac{n-1}{3} \right\rfloor \right\rfloor.
\end{equation*}
\end{corollary}
\begin{proof}
The first equation follows by Lemma \ref{bound2} and Lemma
\ref{bound1} where $w_1=1, w_2=2$ and $w_3=1$.

The second equation follows by Lemma \ref{bound3} and Lemma
\ref{bound1} where $w_1=2, w_2=1$ and $w_3=1$.
\end{proof}
\vskip 10pt

In the following, we denote $U(n,5,[2,1,1])=n \left\lfloor
\frac{n-1}{2}\right\rfloor$ and $U(n,6,[2,1,1])=\left\lfloor
\frac{n}{2} \left\lfloor \frac{n-1}{3} \right\rfloor \right\rfloor$.

\subsection{Designs}
\vskip 10pt

Our recursive construction is based on some combinatorial structures
in design theory. The most important tools are pairwise balanced
designs (PBDs) and group divisible designs (GDDs).

Let $K$ be a subset of positive integers and let $\lambda$ be a
positive integer. A {\it pairwise balanced design} (PBD($v,K,\lambda$) or
($K,\lambda$)-PBD) of order $v$ with block sizes from $K$ is a pair
($\V,\B$), where $\V$ is a finite set (the point set) of cardinality
$v$ and $\B$ is a family of subsets (blocks) of $\V$ that satisfy
(1) if $B\in \B$, then $|B|\in K$ and (2) every pair of distinct
elements of $\V$ occurs in exactly $\lambda$ blocks of $\B$. The
integer $\lambda$ is the index of the PBD. The notations PBD($v,K$)
and $K$-PBD of order $v$ are often used when $\lambda=1$. If an
element $k\in K$ is ``starred'' (written $k^\star$), it means that
the PBD has exactly one block with size $k$.

\vskip 10pt
\begin{lemma}[\cite{HBPBD3}]
\label{PBD4-6} For any integer $v\geq 10$, a $(v,\{4,5,6\},1)$-PBD
exists with exceptions $v\in \{7,8,9,10,$ $11,12,14,15,18,19,23\}$.
\end{lemma}
\vskip 10pt

\begin{lemma}[Ling, Zhu, Colbourn, Mullin \cite{Ling:1997}]
\label{PBD5-9} For any integer $v\geq 10$, a
$(v,\{5,6,7,8,9\},1)$-PBD exists with exceptions $v\in
[10,20]\cup[22,24]\cup[27,29]\cup[32,34]$.
\end{lemma}
\vskip 10pt

\begin{lemma}[Rees, Stinson \cite{Rees-Stinson:1989}] \label{PBD4*} There exists a
$(v,\{4,w^\star\},1)$-PBD with $v>w$ if and only if $v\geq 3w+1$,
and:
\begin{enumerate}
\item[(i)] $v\equiv 1$ or $4\pmod{12}$ and $w\equiv 1$ or $4\pmod{12}$;  or
\item[(ii)] $v\equiv 7$ or $10\pmod{12}$ and $w\equiv 7$ or $10\pmod{12}$.
\end{enumerate}
\end{lemma}
\vskip 10pt

Let $K$ and $G$ be sets of positive integers and let $\lambda$ be a
positive integer. A {\it group divisible design} of index $\lambda$ and
order $v$ (($K,\lambda$)-GDD) is a triple ($\V,\G,\B$), where $\V$
is a finite set of cardinality $v$, $\G$ is a partition of $\V$ into
parts (groups) whose sizes lie in $G$, and $\B$ is a family of
subsets (blocks) of $\V$ that satisfy (1) if $B\in \B$ then $|B|\in
K$, (2) every pair of distinct elements of $\V$ occurs in exactly
$\lambda$ blocks or one group, but not both, and (3) $|\G|>1$. If
$v=a_1g_1+a_2g_2+ \dots +a_sg_s$, and if there are $a_i$ groups with
size $g_i$, $i=1, 2, \ldots, s$, then the $(K,\lambda)$-GDD is of
type $g_1^{a_1}g_2^{a_2}\dots g_s^{a_s}$. This is exponential
notation for the group type. If $K=\{k\}$, then the
($K,\lambda$)-GDD is a ($k,\lambda$)-GDD. If $\lambda=1$, the GDD is
a $K$-GDD. Furthermore, a ($\{k\},1$)-GDD is a $k$-GDD. A {\it parallel
class} or {\it resolution class} is a collection of blocks that partition
the point-set of the design. A GDD is {\it resolvable} if the blocks of
the design can be partitioned into parallel classes. A resolvable
GDD is denoted by RGDD.

\vskip 10pt
\begin{lemma} [\cite{HBGDD}]
\label{4RGDD38} There exists a $4$-RGDD of type $g^u$ for each $(g,u)\in
\{(3,8),(4,7)\}$.
\end{lemma}



\vskip 10pt
\begin{lemma} [\cite{HBGDD}]
\label{4GDD} There exists a $4$-GDD of type $g^u m^1$ for each
$(g,u,m)\in \{(3,5,0),$ $(4,6,7),$ $(12,4,18),$ $(12,5,18),$
$(15,4,21)\}$.
\end{lemma}
\vskip 10pt

A $\{k\}\text{-GDD}$ of type $n^k$ is also called a {\em transversal
design} and denoted by ${\rm TD}(k,n)$.

\vskip 10pt
\begin{lemma}[\cite{Abel:2007}]
\label{TD} Let $n$ be a positive integer. Then:
\begin{enumerate}
\item[(i)] a TD$(5,n)$ exists if $n\not\in \{2,3,6,10\}$;
\item[(ii)] a TD$(6,n)$ exists if $n\not\in \{2,3,4,6,10,14,18,22$\};
\item[(iii)] a TD$(7,n)$ exists if $n\not\in \{2,3,4,5,6,10,14,15,18,20,$ $22,26,30,34,38,46,60\}$.
\item[(iv)] a TD$(8,9)$ exists.
\end{enumerate}
\end{lemma}
\vskip 10pt

%

\subsection{Group Divisible Codes}

Given $\vu \in \bbZ_q^X$ and $Y \subseteq X$, the
\emph{restriction of $\vu$ to $Y$}, written $\vu\mid_Y$, is the
vector $\vv \in \bbZ_q^X$ such that
$$\vv_x=
\begin{cases}
\vu_x, & \text{if $x \in Y$} \\
0, & \text{if $x \in X \backslash Y$}. \\
\end{cases}$$

A {\it group divisible code} (GDC) of distance $d$ is a triple
($X,\G,\C$), where $\G=\{G_1,\ldots,G_t\}$ is a partition of $X$
with cardinality $|X|=n$ and $\C\subseteq \bbZ_q^X$ is a $q$-ary
code of length $n$, such that $d_H(u,v) \ge d$ for each distinct
$u,v\in \C$, and $\|u|_{G_i}\|\le 1$ for each $u \in \C$, $1\le i \le
t$. Elements of $\G$ are called groups. We denote a GDC($X,\G,\C$)
of distance $d$ as $w$-GDC($d$) if $\C$ is of constant weight $w$.
If we want to emphasize the composition of the codewords, we denote
the GDC as $\overline{w}$-GDC($d$) when every $u\in \C$ has
composition $\overline{w}$. The type of a GDC($X,\G,\C$) is the
multiset $\Lbag |G|: G\in \G \Rbag$. As in the case of GDDs, the
exponential notation is used to describe the type of a GDC. The size
of a GDC($X,\G,\C$) is $|\C|$. Note that an
$(n,d,\overline{w})_q$-code with size $s$ is equivalent to a
$\overline{w}$-GDC($d$) of type $1^n$ with size $s$.

Constant-composition codes of larger orders can often be obtained from
GDCs via the following two constructions.

\vskip 10pt
\begin{construction}[(Filling in Groups) \cite{Yeow:2008}]
\label{FillGroups} Let $d\leq 2(w-1)$. Suppose there exists a
$w$-${\rm GDC}(d)$ $(X,\G,\C)$ of type $g_1^{t_1}\cdots g_s^{t_s}$
with size $a$. Suppose further that for each $i$, $1\leq i\leq s$,
there exists a $(g_i,d,w)_q$-code $\C_i$ with size $b_i$, then there
exists a $(\sum_{i=1}^s t_ig_i,d,w)_q$-code $\C'$ with size
$a+\sum_{i=1}^s t_ib_i$. In particular, if $\C$ and $\C_i$, $1\leq
i\leq s$, are of constant composition $\overline{w}$, then $\C'$ is
also of constant composition $\overline{w}$.
\end{construction}
\vskip 10pt

\begin{construction}[(Adjoining $y$ Points) \cite{Yeow:2008}]
\label{AdjoinPoints} Let $y\in\bbZ_{\geq 0}$. Suppose there exists a
(master) $w$-${\rm GDC}(d)$ of type $g_1^{t_1}\cdots g_s^{t_s}$ with size $a$, and suppose the following (ingredients) also exist:
\begin{enumerate}[(i)]
\item a $(g_1+y,d,w)_q$-code with size $b$,
\item a $w$-${\rm GDC}(d)$ of type $1^{g_i} y^1$ with size $c_i$ for each $2\leq i\leq s$,
\item a $w$-${\rm GDC}(d)$ of type $1^{g_1} y^1$ with size $c_1$ if $t_1\geq 2$.
\end{enumerate}
Then, there exists a $(y+\sum_{i=1}^s t_ig_i,d,w)_q$-code with size
\begin{equation*}
a+b+(t_1-1)c_1+\sum_{i=2}^{s} t_ic_i.
\end{equation*}
Furthermore, if the master and ingredient codes are of constant
composition, then so is the resulting code.
\end{construction}
\vskip 10pt

The  following two  constructions are useful for generating  GDCs of larger orders from smaller ones.

\begin{construction}[(Fundamental Construction) \cite{Yeow:2008}]
\label{FundCtr} Let $d\le 2(w-1)$, $\D=(X,\G,\A)$ be a (master) GDD,
and $\omega: X\rightarrow \bbZ_{\ge0}$ be a weight function. Suppose
that for each $A\in \A$, there exists an (ingredient) $w$-GDC($d$)
of type $\Lbag \omega(a): a\in A \Rbag$. Then there exists a
$w$-GDC($d$) $\D^{*}$ of type $\Lbag \sum_{x\in G}{\omega(x)}:G\in
\G \Rbag$. Furthermore, if the ingredient GDCs are of constant
composition $\overline{w}$, then $\D^{*}$ is also of constant
composition $\overline{w}$.
\end{construction}

\vskip 10pt
\begin{construction}
\label{Inflation} Suppose there exists a $w$-GDC($d$) of type
$g_1^{t_1}\dots g_s^{t_s}$ with size $a$. Suppose further that there
exists a TD($w,m$), then there exists a $w$-GDC($d$) of type
$(mg_1)^{t_1}\dots (mg_s)^{t_s}$ with size $am^2$. If the original
GDC is of constant composition $\overline{w}$, then so is the
derived GDC.
\end{construction}
\vskip 10pt

\section{Determining the Value of $A_4(n,5,[2,1,1])$}
\vskip 10pt

\subsection{Room Square Construction}
Let $S$ be a set of $n+1$ elements (symbols). A {\it Room square} of side
$n$ (on symbol set $S$), RS($n$), is an $n \times n$ array, $F$,
that satisfies the following properties:
\begin{enumerate}
\item every cell of $F$ either is empty or contains an unordered
pair of symbols from $S$,
\item each symbol of $S$ occurs once in each row and column of $F$,
\item every unordered pair of symbols occurs in precisely one cell
of $F$.
\end{enumerate}

\vskip 10pt
\begin{lemma} [\cite{skewRS}]
\label{RS} A Room square of side $n$ exists if and only if $n$ is
odd and $n\neq 3$ or $5$.
\end{lemma}
\vskip 10pt

   From  each filled cell ($r,c$) of an RS($n$),  $R$, one can obtain an underlying  $4$-subset $\{i,j,r,c\}$, where  $\{i,j\}$ occurs in column $c$ and row $r$ of $R$. A Room
square  of side $n$ is called {\it super-simple} (denoted  by SSRS($n$)),  if for any two filled cells  ($r_1,c_1$) and ($r_2,c_2$) containing the  symbols $\{i_1,j_1\}$ and
$\{i_2,j_2\}$ respectively,  the underlying $4$-subsets $\{i_1,j_1,r_1,c_1\}$ and $\{i_2,j_2,r_2,c_2\}$  share at most two common elements.

\vskip 10pt
\begin{theorem}
\label{SSRStoCodes} For an odd integer $n$, suppose there exists an
SSRS($n$), then there exists an  optimal $(n,5,[2,1,1])_4$-code with size
$n(n-1)/2=U(n,5,[2,1,1])$.
\end{theorem}
\begin{proof}
For each filled cell $(r,c)$ of the given SSRS($n$), $R$, we form a codeword $\langle
i,j,r,c\rangle$ of type $[2,1,1]$, where  $\{i,j\}$ occurs in column $c$ and row $r$ of $R$.  We have in total of  $n(n-1)/2$ such codewords. Now, we  prove that these $n(n-1)/2$ codewords form an optimal code of length $n$.

Because $R$ is super-simple, so any two codewords intersect in at
most two coordinates. If the distance between any two codewords, $\langle
i_1,j_1,r_1,c_1\rangle$ and $\langle i_2,j_2,r_2,c_2\rangle$, is
less than $5$, then one of the following five properties  must be satisfied:
\begin{enumerate}
\item $i_1=i_2$ and $j_1=j_2$,
\item $i_1=j_2$ and $i_2=j_1$,
\item $i_l=j_m$ and $r_1=r_2$ where $l,m\in \{1,2\}$,
\item $i_l=j_m$ and $c_1=c_2$ where $l,m\in \{1,2\}$, or
\item $r_1=r_2$ and $c_1=c_2$.
\end{enumerate}
But any of these five properties  has conflicts with the properties of a Room square.
This means  the distance between any two codewords is greater than or
equal to $5$. So these codewords form an optimal code of length $n$
with size $n(n-1)/2=U(n,5,[2,1,1])$.
\end{proof}
\vskip 10pt

If $\{S_1,\ldots,S_n\}$ is a partition of a set $S$, an
{\it $\{S_1,\ldots,S_n\}$-Room frame} is an $|S|\times |S|$ array, $F$,
indexed by $S$, satisfying:
\begin{enumerate}
\item  every cell of $F$ either is empty or contains an unordered
pari of symbols of $S$,
\item  the subarrays $S_i\times S_i$ are empty, for each $1\le i\le
n$ (these subarrays are holes),
\item  each symbol $x\notin S_i$ occurs once in row (or column) s
for each $s\in S_i$, and
\item  the pairs occurring in $F$ are those $\{s,t\}$, where $(s,t)\in (S\times S)\setminus \cup_{i=1}^{n}{(S_i\times S_i)}
$.
\end{enumerate}

A Room square of side $n$ is equivalent to a Room frame of type
$1^n$.

We can define a super-simple Room frame in the same way as a
super-simple Room square.



\vskip 10pt
\begin{theorem}
\label{SSRFtoCodes}
Suppose there exists a super-simple Room frame
of type $2^t$, then there exists an  optimal $(2t,5,[2,1,1])_4$-code with size
$2t(t-1)=U(2t,5,[2,1,1])$.
\end{theorem}
\begin{proof}
For each filled cell $(r,c)$ of the super-simple Room frame of type $2^t$, $R$, we form a codeword $\langle
i,j,r,c\rangle$ of type $[2,1,1]$, where  $\{i,j\}$ occurs in column $c$ and row $r$ of $R$. There are
$2\times 2t(t-1)/2=2t(t-1)$ filled cells in $R$, which equals $U(2t,5,[2,1,1])$.  Hence, we have in total of  $U(2t,5,[2,1,1])$ such codewords. The rest of the
proof is similar to that of Theorem \ref{SSRStoCodes}.
\end{proof}
\vskip 10pt

A {\it starter} in the  abelian group $G$ of odd order (written additively),
where $|G|=g$ is a set of unordered pairs $S=\{\{s_i,t_i\}:1 \le i
\le (g-1)/2 \}$ that satisfies:
\begin{enumerate}
\item $\{s_i:1 \le i \le (g-1)/2\} \cup \{t_i:1 \le i \le
(g-1)/2\}=G\setminus \{0\}$, and
\item $\{\pm (s_i-t_i): 1 \le i \le (g-1)/2\}=G\setminus  \{0\}$.
\end{enumerate}

A {\it strong starter} is a starter $S=\{\{s_i,t_i\}\}$ in the abelian
group $G$ with the additional property that $s_i+t_i=s_j+t_j$
implies $i=j$, and for each $i$, $s_i+t_i \ne 0$. Let
$S=\{\{s_i,t_i\}: 1 \le i \le (g-1)/2\}$ and $T=\{\{u_i,v_i\}: 1 \le
i \le (g-1)/2\}$ be two starters in $G$. Without loss of generality,
assume that $s_i-t_i=u_i-v_i$, for all $i$. Then $S$ and $T$ are called
{\it orthogonal} starters if $u_i-s_i=u_j-s_j$ implies $i=j$, and if $u_i
\ne s_i$ for all $i$. An {\it adder} for the starter $S=\{\{s_i,t_i\}:1
\le i \le (g-1)/2\}$ is an ordered set
$A_S=\{a_1, a_2, \ldots, a_{(g-1)/2}\}$ of $(g-1)/2$ distinct nonzero
elements from $G$ such that the set $T=\{\{s_i+a_i,t_i+a_i\}:1 \le i
\le (g-1)/2\}$ is also a starter in the group $G$. These two
starters $S$ and $T$ are orthogonal starters.

\vskip 10pt
\begin{theorem}[\cite{HBRS}]
If there exist two orthogonal starters in a group of order $n$, then
there exists a Room square of side $n$. If the group is $\bbZ_n$,
then the resulting Room square is cyclic.
\end{theorem}
\vskip 10pt

Let $G$ be an additive abelian group of order $g$, and let $H$ be a
subgroup of order $h$ of $G$, where $g-h$ even. An $h$-{\it frame starter}
of order $g/h$ in $G\setminus H$ is a set of unordered pairs
$S=\{\{s_i,t_i\},1\le i \le (g-h)/2\}$ such that
\begin{enumerate}
\item $\{s_i : 1\le i \le (g-h)/2\} \cup \{t_i : 1\le i \le (g-h)/2\}=G\setminus
H$, and
\item $\{\pm (s_i-t_i) : 1\le i \le (g-h)/2\}=G\setminus H$.
\end{enumerate}
A frame starter $A=\{\{s_i,t_i\}\}$ is strong if $s_i+t_i=s_j+t_j$
implies $i=j$, and $s_i+t_i\notin H$ for all $i$. Let
$A=\{\{s_i,t_i\}\}$ and $B=\{\{u_i,v_i\}\}$ be two frame starters.
We may assume that $t_i-s_i=v_i-u_i$, for each $1\le i \le (g-t)/2$. We
say that $A$ and $B$ are orthogonal frame starters if
$u_i-s_i=u_j-s_j$ implies $i=j$, and $u_i-s_i\notin H$ for all $i$.

\vskip 10pt
\begin{lemma}[\cite{Dinitz:1980}]
\label{strongframestarter} If $A=\{\{s_i,t_i\}\}$ is a strong
frame starter then $A$ and $-A=\{\{-s_i,-t_i\}\}$ are orthogonal
frame starters.
\end{lemma}
\vskip 10pt

\begin{lemma}[\cite{Dinitz:1980}]
\label{framestarter} If there exists a pair of orthogonal $t$-frame
starters in $G\setminus H$ with $|G|=g$ and $|H|=t$, then there
exists a Room frame of type $t^u$, where $u=g/t$.
\end{lemma}
\vskip 10pt

So if we have a strong starter in a group of order $n$ which can
generate an SSRS(n), then we get an optimal $(n,5,[2,1,1])_4$-code with size $U(n,5,[2,1,1])$. Similarly, if we  have  a strong frame starter
in $G\setminus H$ with $|G|=g$ and $|H|=2$ which can generate a
super-simple Room frame of type $2^{g/2}$, then we get an optimal
$(g,5,[2,1,1])_4$-code with size $U(g,5,[2,1,1])$.

\subsection{Some Small $[2,1,1]$-GDC$(5)$s and Optimal Codes with Distance $5$}
\vskip 10pt

In the sequel, we construct some small $[2,1,1]$-GDC$(5)$s and optimal
codes via computer search. The constructions are based on the familiar difference method,
where a finite group (mostly abelian group $\bbZ_u$) will be
utilized to generate all the codewords of a code or a GDC.
Thus, instead of listing all the codewords, we list a set of base
codewords and generate the others by an additive group and
perhaps some further automorphisms.  Mostly, the set of base
codewords  are divided into two parts, $P$ and $R$, where
each codeword of $P$ will be multiplied by $m^i$ for each $0 \le i \le
s-1$ to generate $s$ codewords, and $R$ is the set of the remaining
base codewords. The desired codes are generated by developing the
base codewords $+M$ modulo $n$. Then, we just need to list $n$, $m$, $s$, $M$, $P$ and $R$
for each code. Sometimes, $R$ may be empty, which will be omitted.

\vskip 10pt
\begin{proposition}
\label{GDC(5)gt}There exists a $[2,1,1]$-GDC$(5)$ of type $g^t$ with size $\frac{g^2t(t-1)}{2}$ for the following parameters:
\begin{enumerate}
\item $g=2$, $t \in \{8, 9, 11\}$,
\item $g=3$, $t \in \{7, 9\}$,
\item $g=4$, $t \in \{5, 6, 7, 8, 9, 11\}$,
\item $g=6$, $t \in \{5\}$.
\end{enumerate}
\end{proposition}
\begin{proof}
For each given pair $\{g, t\}$, let $X_{\{g, t\}}=\bbZ_{gt}$,
$\G_{\{g, t\}}=\{\{i,t+i, \ldots, (g-1)t+i\}:i \in \bbZ_{t}\}$ and  $\C_{\{g, t\}}$ be the set of cyclic (or quasi-cyclic) shifts of the vectors generated by the following vectors respectively. Then $(X_{\{g, t\}},\G_{\{g, t\}},\C_{\{g, t\}})$ is a
$[2,1,1]$-GDC$(5)$ of type $g^t$ with size $\frac{g^2t(t-1)}{2}$. In order to save space, we list only one example here. Other cases can be found in Propositions 6.1--6.4. For $g=2$ and $t=8$, we have $n=16$, $m=11$, $s=2$, $M=2$ and
\begin{equation*}
\footnotesize
\begin{array}{cccccc}
P:& \langle 0, 13, 15, 6\rangle\\
R:& \langle 1, 11, 0, 13\rangle & \langle 0, 1, 4, 14\rangle & \langle 0, 11, 9, 15\rangle & \langle 0, 3, 7, 10\rangle\\
&\langle 0, 6, 2, 1\rangle &\langle 0, 2, 13, 7\rangle & \langle 1, 5, 2, 11\rangle & \langle 0, 7, 14, 12\rangle\\
&\langle 0, 5, 1, 3\rangle & \langle 0, 9, 3, 4\rangle & \langle 0, 12, 6, 9\rangle & \langle 1, 3, 12, 2\rangle.\\
\end{array}
\end{equation*}
\end{proof}
\vskip 10pt

\begin{proposition}
\label{GDC(5)4t21}There exists a $[2,1,1]$-GDC$(5)$ of type $4^t2^1$
with size $8t^2$ for each $t \in \{5,6\}$.
\end{proposition}
\begin{proof}
Detailed constructions can be found in Propositions 6.5.
\end{proof}
\vskip 10pt

\begin{proposition}
\label{CCC(5)}$A_4(n,5,[2,1,1])$  =  $U(n,5,[2,1,1])$ for each $n \in
\{12,14,15,17,20,28,30,36,38,44,46,52,54,62,$
$68,70,76,78,86,92,94,110,126,134\}$.
\end{proposition}
\begin{proof}
Detailed constructions can be found in Propositions 6.6.
\end{proof}
\vskip 10pt

\begin{proposition}
\label{CCC(5s)}$A_4(n,5,[2,1,1])$  =  $U(n,5,[2,1,1])$ for each $n \in
\{19,21,23,24,25,26,27,29,31,32,33,34,35,$
$37,39,40,41,42,43,45,47,48,49,50,51,53,55,56,57,58,$
$59,61,63,64,65,66,67,69,71,73,74,75,77,79,83,85,87,$
$88,89,93,95,99,103,104,106,107,109,111,123,125,127,$
$131,133,138,139\}$.
\end{proposition}
\begin{proof}
All these optimal codes are constructed by strong starters or strong
frame starters with n odd or even respectively. By Theorems \ref{SSRStoCodes} and \ref{SSRFtoCodes}, there
exists an optimal code of such length $n$. The starters are given in a similar way as the codewords in the above propositions.
In order to save space, we list only one example here. Other cases can be found in Proposition 6.7. For  $n=19$, we have $m=4$, $s=9$, $M=1$ and $P: \{1, 3\}$.
\end{proof}
\vskip 10pt

\subsection{The Case of Length $n \equiv 0 \pmod{4}$}
\vskip 10pt

\begin{lemma}
\label{CCC(5)4 8} $A_4(4,5,[2,1,1])=1$ and $A_4(8,5,[2,1,1]) \ge
18$.
\end{lemma}
\begin{proof}
For $n=4$, the one required codeword is $\langle 0, 1, 2, 3\rangle$. For $n=8$, the $18$  required codewords can be found in Proposition 6.8.
\end{proof}
\vskip 10pt

\begin{theorem}
\label{CCC4t} There exists an optimal $(4t,5,[2,1,1])_4$-code with
size $4t(2t-1)$ for all $t \ge 3$.
\end{theorem}
\begin{proof}
For each $3\le t\le 14$ or $t\in \{16,17,19,22,23,$ $26\}$, there exists an
optimal $(4t,5,[2,1,1])_4$-code by Propositions \ref{GDC(5)gt},
\ref{CCC(5)} and \ref{CCC(5s)}.

For each $u\in \{5,6,7,9,11\}$, there exists a $[2,1,1]$-GDC(5) of type
$4^u$. Apply Construction \ref{Inflation} with a TD($4,3$) (which
exists by Lemma \ref{TD}) to get a $[2,1,1]$-GDC(5) of type $12^u$. Filling in the groups of this GDC with an
optimal $(12,5,[2,1,1])_4$-code coming from Proposition
\ref{CCC(5)}, the result is an optimal $(4t,5,[2,1,1])_4$-code for each
$t\in \{15,18,21,27,33\}$.

For each $u\in \{7,8\}$, there exists a $[2,1,1]$-GDC(5) of type $4^u$,
and apply Construction \ref{Inflation} with a TD($4,4$) to get a
$[2,1,1]$-GDC(5) of type $16^u$ for each $u\in \{7,8\}$. Filling in the
groups of this GDC with an optimal $(16,5,[2,1,1])_4$-code (see
 Proposition \ref{GDC(5)gt}), the result is an optimal
$(4t,5,[2,1,1])_4$-code for each $t\in \{28,32\}$.

For $t=31$, take a $4$-RGDD of type $3^8$ (see Lemma~\ref{4RGDD38}). There are $7$ parallel
classes in this $4$-RGDD. Add one ideal point to each of these  $7$ parallel
classes to complete them. The result is a $5$-GDD of type $3^8 7^1$.  Apply Construction \ref{FundCtr} with
weight $4$ to each point of this $5$-GDD and fill in the groups with optimal
codes of lengths $12$ and $28$ to obtain an optimal
($124,5,[2,1,1]$)-code. Here, the input design, a $[2,1,1]$-GDC(5) of
type $4^5$, exists by Proposition \ref{GDC(5)gt}.

For all $t\geq 34$ or $t\in \{20,24,25,29,30\}$, take a
$(t+1,\{5,6,7,8,9\},1)$-PBD from Lemma \ref{PBD5-9} and remove one
point to obtain a $\{5,6,7,8,9\}$-GDD of type $4^i 5^j 6^k 7^l 8^m$
with $4i+5j+6k+7l+8m=t$. Apply Construction \ref{FundCtr} with
weight $4$ using  $[2,1,1]$-GDC(5)s of types $4^t$ for  $t\in
\{5,6,7,8,9\}$  (Proposition \ref{GDC(5)gt}) as input ingredients to obtain a
$[2,1,1]$-GDC(5) of type $16^i 20^j 24^k 28^l 32^m$ and length $4t$. Filling in the groups of this GDC with optimal
$(n,5,[2,1,1])_4$-codes for   $n\in \{16,20,24,28,32\}$ by Propositions
\ref{GDC(5)gt}, \ref{CCC(5)} and \ref{CCC(5s)}, the result is an
optimal $(4t,5,[2,1,1])_4$-code for all $t\geq 34$ or $t\in
\{20,24,25,29,30\}$.
\end{proof}

\subsection{The Case of Length $n \equiv 1 \pmod{4}$}
\vskip 10pt

\begin{lemma}
\label{CCC(5)5913} $A_4(5,5,[2,1,1])=2$, $A_4(9,5,[2,1,1])\ge 27$
and $A_4(13,5,[2,1,1])\ge 72$.
\end{lemma}
\begin{proof}
All the required codewords can be found in Proposition 6.9.
\end{proof}
\vskip 10pt

\begin{theorem}
\label{CCC4t+1}  There exists an optimal $(4t+1,5,[2,1,1])_4$-code
with size $2t(4t+1)$ for all $t \ge 4$.
\end{theorem}
\begin{proof}
For each $4\le t\le 19$ or $t\in \{21,22,23,27,$ $31,33\}$, there exists an
optimal $(4t+1,5,[2,1,1])_4$-code by Propositions \ref{CCC(5)} and
\ref{CCC(5s)}.

For $t=26$, there exists a $[2,1,1]$-GDC(5) of type $3^7$. Apply
Construction \ref{Inflation} with a TD($4,5$) to get a
$[2,1,1]$-GDC(5) of type $15^7$.  Fill in the groups with an
optimal $(15,5,[2,1,1])_4$-code to obtain an optimal
$(4t+1,5,[2,1,1])_4$-code.

For each $u\in \{7,8\}$, there exists a $[2,1,1]$-GDC(5) of type $4^u$.
Apply Construction \ref{Inflation} with a TD($4,4$) to get a
$[2,1,1]$-GDC(5) of type $16^u$. Adjoin one ideal
point to this GDC, and fill in the groups together with the extra point with an optimal
$(17,5,[2,1,1])_4$-code to obtain an optimal
$(4t+1,5,[2,1,1])_4$-code for each $t\in \{28,32\}$.

For all $t\geq 34$ or $t\in \{20,24,25,29,30\}$, take a
$(t+1,\{5,6,7,8,9\},1)$-PBD from Lemma \ref{PBD5-9} and remove one
point to obtain a $\{5,6,7,8,9\}$-GDD of type $4^i 5^j 6^k 7^l 8^m$
with $4i+5j+6k+7l+8m=t$. Apply construction \ref{FundCtr} with
weight 4 using $[2,1,1]$-GDC(5)s of types $4^t$  for  $t\in
\{5,6,7,8,9\}$ (Lemma \ref{GDC(5)gt})  as input ingredients to obtain a $[2,1,1]$-GDC(5)
of type $16^i 20^j 24^k 28^l 32^m$ and length $4t$. Adjoin one
ideal point to this GDC, and fill in the groups  together with the extra point with optimal
$(n,5,[2,1,1])_4$-codes for  $n\in \{17,21,25,29,33\}$ (Propositions
\ref{CCC(5)} and \ref{CCC(5s)}) to obtain an optimal
$(4t+1,5,[2,1,1])_4$-code for all $t\geq 34$ or $t\in
\{20,24,25,29,30\}$.
\end{proof}

\subsection{The Case of Length $n \equiv 2 \pmod{4}$}
\vskip 10pt

\begin{lemma}
\label{CCC(5)610} $A_4(6,5,[2,1,1])=6$, $A_4(10,5,[2,1,1]) \ge 36$.
\end{lemma}
\begin{proof}
All the required codewords can be found in Proposition 6.10.
\end{proof}
\vskip 10pt

\begin{theorem}
\label{CCC4t+2(1)} There exists an optimal $(4t+2,5,[2,1,1])_4$-code
with size $2t(4t+2)$ for each $3\le t \le 37$.
\end{theorem}

\begin{proof}
For each $3\le t\le 34$ and $t\notin \{20,22,24,25,$ $28,29,30,32\}$, there exists an
optimal $(4t+2,5,[2,1,1])_4$-code with  size $2t(4t+2)$ by Propositions
\ref{GDC(5)gt}, \ref{CCC(5)} and \ref{CCC(5s)}.

For each $t\in\{20,24,28,32,36\}$, there exists a $[2,1,1]$-GDC(5) of type $4^u$ with $u\in\{5,6,7,8,9\}$. Apply
Construction \ref{Inflation} with a TD($4,4$) to get a
$[2,1,1]$-GDC(5) of type $16^u$ with $u\in\{5,6,7,8,9\}$. Adjoin two additional points and fill in the groups together with the two extra points with a $[2,1,1]$-GDC(5) of type $2^9$ to obtain an optimal $[2,1,1]$-GDC(5) of type $2^{8u+1}$ and lengths $82,98,114,130$ or $146$.

For each $t\in \{22, 37\}$, there exists a $[2,1,1]$-GDC(5) of type $6^5$. Let $t_1=(4t+2)/30$. Apply
Construction \ref{Inflation} with a TD($4,t_1$) to get a
$[2,1,1]$-GDC(5) of type $(6t_1)^5$. Fill in the groups with an
optimal $(6t_1,5,[2,1,1])_4$-code to obtain an optimal code of length
$4t+2$.

For each $t\in\{25,30,35\}$, there exists a $[2,1,1]$-GDC(5) of type $4^u$ with $u\in\{5,6,7\}$. Apply
Construction \ref{Inflation} with a TD($4,5$) to get a
$[2,1,1]$-GDC(5) of type $20^u$ with $u\in\{5,6,7\}$. Adjoin two additional points and fill in the groups  together with the two extra points with a $[2,1,1]$-GDC(5) of type $2^{11}$ to obtain an  optimal  $[2,1,1]$-GDC(5) of type $2^{10u+1}$ and  lengths $102,122$ or $142$.

For $t=29$, take a TD($6,5$) from Lemma \ref{TD}. Apply
Construction \ref{FundCtr} with weight $4$ to the points in the
first five groups and $4$ points in the last group,  and  weight $2$ to
the other $1$ points in the last group. All the remaining points are given weight $0$.
Note that there exist $[2,1,1]$-GDC(5)s of types $4^5$, $4^6$ and
$4^5 2^1$ by Propositions \ref{GDC(5)gt} and \ref{GDC(5)4t21}. The result is a $[2,1,1]$-GDC(5) of type $20^5
 18^1$. Filling in the groups with optimal
codes of lengths $18$ or $20$, the results are optimal codes
of length $118=4\times 29+2$.
\end{proof}
\vskip 10pt

\begin{theorem}
\label{CCC4t+2(2)} There exists an optimal $(4t+2,5,[2,1,1])_4$-code
with size $2t(4t+2)$ for all $t \ge 38$.
\end{theorem}

\begin{proof}
Take a TD$(7,r)$ for  $r \ge 7$ and $r\not\in \{10,14,15,18,$ $20, 22,26,30,34,38,46,60\}$
  from Lemma \ref{TD}. Apply Construction
\ref{FundCtr} with weight $4$ to the points in the first five
groups, $x$ points in the sixth group and $y$ points in the last
group, and  weight $2$ to the other $z$ points in the last group. All the remaining points are given weight $0$. Here,   we require that
$x\ge 4$, $4y+2z\ge 14$ and $z$ odd.  Note that there exist $[2,1,1]$-GDC(5)s
of types $4^5$, $4^6$, $4^7$, $4^5 2^1$, and $4^6 2^1$ by Propositions
\ref{GDC(5)gt} and \ref{GDC(5)4t21}. The result is a
$[2,1,1]$-GDC(5) of type $(4r)^5 (4x)^1 (4y+2z)^1$. Fill in the groups with optimal codes of
lengths $4u$ with $u\ge 4$ (which exist by Theorem \ref{CCC4t}) or
$4v+2$ with $3\le v \le 37$ (which exist by Theorem
\ref{CCC4t+2(1)}). The result is an optimal
$(20r+4x+4y+2z,5,[2,1,1])_4$-code, where $20r+4x+4y+2z$ can take any
value $n$ with $n \equiv 2 \pmod{4}$ and $n \ge 4\times 38+2=154$.
\end{proof}

\subsection{The Case of Length $n \equiv 3 \pmod{4}$}
\vskip 10pt

\begin{lemma}
\label{CCC(5)711} $A_4(7,5,[2,1,1])=10$, $A_4(11,5,[2,1,1])\ge 48$.
\end{lemma}
\begin{proof}
All the required codewords can be found in Proposition 6.11.
\end{proof}
\vskip 10pt

\begin{theorem}
\label{CCC4t+3(1)} There exists an optimal $(4t+3,5,[2,1,1])_4$-code
with size $(2t+1)(4t+3)$ for each $3\le t \le 37$.
\end{theorem}

\begin{proof}
For each $3\le t\le 34$ and $t\notin \{22,28,29,33\}$, there exists an
optimal $(4t+3,5,[2,1,1])_4$-code with size $(2t+1)(4t+3)$ by
Propositions \ref{CCC(5)} and \ref{CCC(5s)}.

For $t=22$, there exists a $[2,1,1]$-GDC(5) of type $6^5$. Apply
Construction \ref{Inflation} with a TD($4,3$) to get a
$[2,1,1]$-GDC(5) of type $18^5$. Adjoin one ideal point and fill in
the groups together with the extra point with an optimal $(19,5,[2,1,1])_4$-code to obtain an
optimal code of length $91=4\times 22+3$.

For each $t\in\{28,29\}$, take a TD($6,5$) from Lemma \ref{TD}. Apply
Construction \ref{FundCtr} with weight $4$ to the points in the
first five groups and  $x$ points in the last group, and weight $2$ to
the other $y$ points in the last group.  All the remaining points are given weight $0$.
Note that there exist $[2,1,1]$-GDC(5)s of types $4^5$, $4^6$ and
$4^5 2^1$ by Propositions \ref{GDC(5)gt} and \ref{GDC(5)4t21}. The result is a $[2,1,1]$-GDC(5) of type $20^5
 (4x+2y)^1$. Here, $4x+2y$ can take $14$ or $18$ when $x=3, y=1$ or
 $x=4, y=1$.  Adjoining one ideal point and filling in the groups together with the extra point with optimal codes of lengths $15$, $19$ or $21$, the results are optimal codes
of lengths $115=4\times 28+3$ or $119=4\times 29+3$.

For $t=33$, there exists a $[2,1,1]$-GDC(5) of type $3^9$ by
Proposition \ref{GDC(5)gt}. Apply Construction \ref{Inflation}
with a TD($4,5$) to get a $[2,1,1]$-GDC(5) of type $15^9$. Fill
in the groups with an optimal $(15,5,[2,1,1])_4$-code to obtain an
optimal code of length $135=4\times 33+3$.

For $t=35$, take a TD($9,8$) from Lemma~\ref{TD}. Apply Construction \ref{FundCtr}
with weight $2$ to the points in the first eight groups and $7$
points in the last group. The other points are given weight 0. Note
that there exist $[2,1,1]$-GDC(5)s of types $2^8$, $2^9$  by Proposition
\ref{GDC(5)gt}. The result is a $[2,1,1]$-GDC(5) of type $16^8
14^1$. Adjoining one ideal point and filling in the group together with the extra point with
optimal codes of lengths $15$ or $17$, the result is an optimal code of
length $143=4\times 35+3$.

For $t=36$, there exists a $[2,1,1]$-GDC(5) of type $3^7$ by
Proposition \ref{GDC(5)gt}. Apply Construction \ref{Inflation}
with a TD($4,7$) to get a $[2,1,1]$-GDC(5) of type $21^7$. Fill
in the groups with an optimal $(21,5,[2,1,1])_4$-code to obtain an
optimal code of length $147=4\times 36+3$.

For $t=37$, there exists a $[2,1,1]$-GDC(5) of type $6^5$. Apply
Construction \ref{Inflation} with a TD($4,5$) to get a
$[2,1,1]$-GDC(5) of type $30^5$. Adjoin one ideal point and fill in
the groups together with the extra point with an optimal $(31,6,[2,1,1])_4$-code to obtain an
optimal code of length $151=4\times 37+3$.
\end{proof}
\vskip 10pt

\begin{theorem}
\label{CCC4t+3(2)} There exists an optimal $(4t+3,5,[2,1,1])_4$-code
with size $2t(4t+3)$ for all $t \ge 38$.
\end{theorem}

\begin{proof}
Take a TD$(7,r)$ for  $r \ge 7$ and $r\not\in \{10,14,15,18,$ $20, 22,26,30,34,38,46,60\}$
 from Lemma \ref{TD}. Apply Construction
\ref{FundCtr} with weight $4$ to the points in the first five
groups, $x$ points in the sixth group and $y$ points in the last
group, and weight $2$ to the other $z$ points in the last group.  All the remaining points are given weight $0$. We require that
$x\ge 4$ and $4y+2z\ge 14$. Note that there exist $[2,1,1]$-GDC(5)s
of types $4^5$, $4^6$, $4^7$, $4^5 2^1$, and $4^6 2^1$ by Propositions
\ref{GDC(5)gt} and \ref{GDC(5)4t21}. The result is a
$[2,1,1]$-GDC(5) of type $(4r)^5 (4x)^1 (4y+2z)^1$. Adjoin one ideal point and fill in the
groups together with the extra point with optimal codes of lengths $4u+1$ with $u\ge 4$ (which exist
by Theorem \ref{CCC4t}) or $4v+3$ with $3\le v \le 37$ (which
exist by Theorem \ref{CCC4t+3(1)}). The result is an optimal
$(20r+4x+4y+2z+1,5,[2,1,1])_4$-code, where $20r+4x+4y+2z+1$ can take
any value of $4t+3$ with $t$ greater than $38$.
\end{proof}
\vskip 10pt

\section{Determining the Value of $A_4(n,6,[2,1,1])$}
\vskip 10pt

\subsection{Some Small $[2,1,1]$-GDC$($6$)$ and Optimal Codes with Distance $6$}
\vskip 10pt

First, we construct some small $[2,1,1]$-GDC$(6)$s and optimal
codes via computer search. In the codes with infinite points, the subscripts on the elements
$x_0\in \{x\}\times \bbZ_u$ for $x\in \{a,b,c,d,e\}$ are developed modulo the unique subgroup
in the abelian group $\bbZ_{n}$ of order $u$.
\vskip 10pt

\begin{proposition}
\label{GDC(6)gt}There exists a $[2,1,1]$-GDC$(6)$ of type $g^t$ with size $\frac{g^2t(t-1)}{6}$ for the following parameters:
\begin{enumerate}
\item $g=2$, $t \in \{10, 13, 16, 19, 22, 25, 28, 34\}$,
\item $g=3$, $t \in \{5, 7\}$,
\item $g=4$, $t \in \{4, 7\}$,
\item $g=6$, $t \in \{4, 5, 6, 7\}$,
\item $g\in \{7, 10, 13, 22\}$, $t \in \{4\}$.
\end{enumerate}
\end{proposition}
\begin{proof}
Detailed constructions can be found in Propositions 7.1--7.6, 7.9 and 7.10.
\end{proof}
\vskip 10pt

\begin{proposition}
\label{GDC(6)12tu1}There exists a $[2,1,1]$-GDC$(6)$ of type
$12^tu^1$ with size $12t(2t+\frac{u}{3}-2)$ for the following parameters:
\begin{enumerate}
\item $u=9$, $t \in \{4,5,\ldots,15,17,18,19, 23\}$,
\item $u=15$, $t \in \{7, 8, \ldots, 15\}$.
\end{enumerate}
\end{proposition}
\begin{proof}
Detailed constructions can be found in Propositions 7.7 and 7.8.
\end{proof}
\vskip 10pt

\begin{proposition}
\label{GDC(6)14_2}There exists a $[2,1,1]$-GDC$(6)$ of type
$1^{12}2^1$ with size $28$.
\end{proposition}
\begin{proof}
Detailed construction can be found in Proposition 7.11.
\end{proof}
\vskip 10pt

\begin{proposition}
\label{GDC(6)1t111}There exists a $[2,1,1]$-GDC$(6)$ of type
$1^t11^1$ with size $6u^2+20u$, where $6u=t$ for each $t \in \{30,36,54,$
$66,78\}$.
\end{proposition}
\begin{proof}
Detailed constructions can be found in Proposition 7.12.
\end{proof}
\vskip 10pt

\begin{proposition}
\label{CCC(6)}$A_4(n,6,[2,1,1])$  =  $U(n,6,[2,1,1])$ for each $n
\in \{6, 8, 10, 11, 13, 14, 16, 17, 19, 22, 23, 25, 28, 31, 34,$ $
35, 37, 43, 55, 67, 79, 103\}$.
\end{proposition}
\begin{proof}
Detailed constructions can be found in Proposition 7.13.
\end{proof}
\vskip 10pt

\subsection{The Case of Length $n \equiv 0, 1 \pmod{6}$}
\vskip 10pt

\begin{lemma}
\label{CCC(6)7} $A_4(7,6,[2,1,1])=4$.
\end{lemma}
\begin{proof}
All the required codewords can be found in Proposition 7.14.
\end{proof}
\vskip 10pt

\begin{lemma}
\label{CCC6t} For any positive integer $t$, if
$A_4(6t+1,6,[2,1,1])=U(6t+1,6,[2,1,1])$ then
$A_4(6t,6,[2,1,1])=U(6t,6,[2,1,1])$.
\end{lemma}
\begin{proof}
In an optimal $(6t+1,6,[2,1,1])_4$-code, every coordinate has exactly $2t+2t=4t$ non-zero
elements. Fix a  coordinate $x$ and remove all the $4t$ codewords containing   non-zero
elements in this  coordinate $x$. Shorten all the remaining codewords by deleting the element $0$ in   coordinate $x$ from them.  The resultant
codewords form an optimal $(6t,6,[2,1,1])_4$-code with size
$U(6t+1,6,[2,1,1])-4t=6t^2-3t=U(6t,6,[2,1,1])$.
\end{proof}
\vskip 10pt

\begin{lemma}
\label{GDC(6)12t} There exists a $[2,1,1]$-GDC(6) of type $12^t$ for
all $t\ge 4$.
\end{lemma}
\begin{proof}
When $t \equiv 0$ or $1\pmod{4}$ and $t\geq 4$, there exists a
$(3t+1,\{4\},1)$-PBD   by Lemma \ref{PBD4*}. Deleting one
point from the point set gives a $\{4\}$-GDD of type $3^t$. When
$t\equiv 2$ or $3\pmod{4}$ and  $t\geq 7$, there exists a
$(3t+1,\{4,7^\star\},1)$-PBD  by Lemma \ref{PBD4*}.  Remove one
point from this PBD which is not in the unique block with size
$7$ to  obtain  a
$\{4,7^\star\}$-GDD of type $3^t$. Hence, we always have a
$\{4,7\}$-GDD of type $3^t$ for all $t\geq 4$ and $t\neq 6$.

Apply Construction \ref{FundCtr} with weight $4$ to obtain a
$[2,1,1]$-GDC($6$) of type $12^t$ for all $t\geq 4$ and $t\neq 6$.
Here, the input $[2,1,1]$-GDC($6$)s of types $4^4$ and $4^7$ exist by
Proposition \ref{GDC(6)gt}.

For $t=6$, take a $\{5\}$-GDD of type $4^6$ (see \cite{GeLing:2005})
and apply Construction \ref{FundCtr} with weight $3$ to obtain a
$[2,1,1]$-GDC($6$) of type $12^6$. Here, the input
$[2,1,1]$-GDC($6$) of type $3^5$ exists by Proposition
\ref{GDC(6)gt}.
\end{proof}

\vskip 10pt
\begin{theorem}
\label{CCC6t+1(1)} There exists an optimal
$(12t+1,6,[2,1,1])_4$-code with size $24t^2+2t$ for all $t\ge 1$.
\end{theorem}
\begin{proof}
For all $t\ge 4$, adjoin one ideal point to a $[2,1,1]$-GDC($6$) of
type $12^t$ (Lemma \ref{GDC(6)12t}) and fill in the groups together with the extra point with an
optimal $(13,6,[2,1,1])_4$-code (which exists by Proposition
\ref{CCC(6)}) to obtain an optimal $(12t+1,6,[2,1,1])_4$-code with size
$24t^2+2t$.

For each $t\in \{1,2,3\}$, there exists an optimal
$(12t+1,6,[2,1,1])_4$-code with size $24t^2+2t$ by Proposition
\ref{CCC(6)}.
\end{proof}
\vskip 10pt

\begin{theorem}
\label{CCC6t+1(2)} There exists an optimal
$(12t+7,6,[2,1,1])_4$-code with size $(12t+7)(2t+1)$ for each $1 \le
t\le 16$.
\end{theorem}
\begin{proof}
For each $t\in \{1,2,3,4,5,6,8\}$, there exists an optimal
$(12t+7,6,[2,1,1])_4$-code with size $(12t+7)(2t+1)$ by Proposition
\ref{CCC(6)}.

For the other $9$ values of $t$, we first construct $9$  GDCs  of types $18^5$, $24^4 18^1$,  $24^4 30^1$,  $30^4 18^1$, $30^5$,  $30^5 12^1$,  $30^5 24^1$, $24^7 18^1$ and $24^7 30^1$ as follows: For $t=7$, take a $\{4\}$-GDD of type $3^5$ (which exists by Lemma
\ref{4GDD}). Apply Construction \ref{FundCtr} with weight $6$ to
obtain a $[2,1,1]$-GDC($6$) of type $18^5$. For $t=9$, take a TD($5,4$) from Lemma \ref{TD}. Remove one point
from a group to obtain a $\{4,5\}$-GDD of type $4^4 3^1$. Apply
Construction \ref{FundCtr} with weight $6$ to obtain a
$[2,1,1]$-GDC($6$) of type $24^4 18^1$.
For $t=10$, take a TD($5,5$) from Lemma \ref{TD}. Remove $4$ points
from a block to obtain a $\{4,5\}$-GDD of type $4^4 5^1$. Apply
Construction \ref{FundCtr} with weight $6$ to obtain a
$[2,1,1]$-GDC($6$) of type $24^4 30^1$.
For $t=11$, take a TD($5,5$) from Lemma \ref{TD}. Remove $2$ points
from a group to obtain a $\{4,5\}$-GDD of type $5^4 3^1$. Apply
Construction \ref{FundCtr} with weight $6$ to obtain a
$[2,1,1]$-GDC($6$) of type $30^4 18^1$.
For $t=12$, take a TD($5,5$) from Lemma \ref{TD}. Apply Construction
\ref{FundCtr} with weight $6$ to obtain a $[2,1,1]$-GDC($6$) of type
$30^5$.
For each $t\in \{13,14\}$, take a TD($6,5$) from Lemma \ref{TD}. Remove
$3$ or $1$ points from a group to obtain a $\{4,5\}$-GDD of types
$5^5 2^1$ or $5^5 4^1$. Apply Construction
\ref{FundCtr} with weight $6$ to obtain a $[2,1,1]$-GDC($6$) of types
$30^5 12^1$ or $30^5 24^1$.
For each $t\in \{15,16\}$, take a $4$-RGDD of type $4^7$ (see Lemma~\ref{4RGDD38}). There are $8$ parallel
classes in this $4$-RGDD. Add one ideal point to each of the $u$ parallel
classes for $u\in \{3, 5\}$ to complete them. The result is a $\{4,5\}$-GDD of type
$4^7 u^1$.  Apply Construction \ref{FundCtr} with weight $6$ to obtain a $[2,1,1]$-GDC($6$) of types $24^7 18^1$ or $24^7 30^1$.
Here, the input $[2,1,1]$-GDC($6$)s of types $6^4$ and $6^5$ exist by Proposition \ref{GDC(6)gt}.
Adjoining one ideal point to each of the above GDCs and filling in the groups together with the extra point with an optimal code of length $n_1 \in \{13, 19, 25, 31\}$, the result is an optimal code of length $12t+7$ with $t\in \{7, 9, 10, 11, 12, 13, 14, 15, 16\}$ as desired.
\end{proof}
\vskip 10pt

\begin{theorem}
\label{CCC6t+1(3)} There exists an optimal
$(12t+7,6,[2,1,1])_4$-code with size $(12t+7)(2t+1)$ for all $t\ge
17$.
\end{theorem}
\begin{proof}
Take a TD$(6,2t)$ from Lemma \ref{TD}. Apply Construction
\ref{FundCtr} with weight $6$ to the points in the first $4$ groups,
$2x$ points in the fifth group, and $y$ points in the last group.
The other points are given weight 0. Note that there exist
$[2,1,1]$-GDC($6$)s of types $6^4$, $6^5$, $6^6$ by Proposition
\ref{GDC(6)gt}. The result is a $[2,1,1]$-GDC($6$) of type $(12t)^4
(12x)^1 (6y)^1$. We require that $y\ge 3$. Adjoin one ideal point
and fill in the groups together with the extra point with optimal codes of lengths $12t+1$ with
$t\ge 1$ (which exist by Theorem \ref{CCC6t+1(1)}) or
$12u+7$ with $1\le u \le 16$ (which exist by Theorem
\ref{CCC6t+1(2)}). The result is an optimal
$(48t+12x+6y+1,6,[2,1,1])_4$-code, where $48t+12x+6y+1$ can take any
value greater than $12\times 17+7=211$ except for the case of
$295=12\times 24+7$.

For $t=24$, take a TD$(7,8)$. Apply Construction  \ref{FundCtr}
with weight $6$ to the points in the first $5$ groups, $6$ points in
the sixth group, and $3$ points in the last group. The other points
are given weight 0. Note that there exist $[2,1,1]$-GDC($6$)s of
types $6^5$, $6^6$, $6^7$ by Proposition \ref{GDC(6)gt}. The result
is a $[2,1,1]$-GDC($6$) of type $48^5 36^1 18^1$. Adjoin one
ideal point and fill in the groups together with the extra point with optimal codes of lengths
$37$, $49$ (both exist by Theorem \ref{CCC6t+1(1)}) or $19$ (which exists by Theorem \ref{CCC6t+1(2)}). The result is an optimal $(295,6,[2,1,1])_4$-code.
\end{proof}
\vskip 10pt

\begin{theorem} There exists an optimal $(6t,6,[2,1,1])_4$-code with size $6t^2-3t$ for all $t\ge
1$.
\end{theorem}
\begin{proof}
The result follows by combining Lemma \ref{CCC6t}, Theorems
\ref{CCC6t+1(1)}--\ref{CCC6t+1(3)}
and the fact that there exists an optimal $(6,6,[2,1,1])_4$-code with size $3$ (see Proposition \ref{CCC(6)}).
\end{proof}
\vskip 10pt

\subsection{The Case of Length $n \equiv 2 \pmod{6}$}
\vskip 10pt

\begin{lemma}
\label{CCC6t+2(1)} Every $[2,1,1]$-GDC($6$) of type $2^{3t+1}$ is an
optimal $(6t+2,6,[2,1,1])_4$-code with size $6t^2+2t$.
\end{lemma}
\begin{proof}
The size of a $[2,1,1]$-GDC($6$) of type $2^{3t+1}$ is $6t^2+2t$
which meets the upper bound of an optimal $(6t+2,6,[2,1,1])_4$-code.
\end{proof}
\vskip 10pt

\begin{theorem}
\label{CCC6t+2(2)} There exists an optimal $(6t+2,6,[2,1,1])_4$-code
with size $6t^2+2t$ for each $t\in \{1,2,\ldots,11\} \cup
\{14,17,18,22\}$.
\end{theorem}
\begin{proof}
For each $t\in \{1,2\}$, there exists an optimal
$(6t+2,6,[2,1,1])_4$-code with size $6t^2+2t$ by Proposition
\ref{CCC(6)}.

For each $t\in \{3,4,5,6,7,8,9,11\}$, there exists a $[2,1,1]$-GDC($6$)
of type $2^{3t+1}$ by Proposition \ref{GDC(6)gt}.

For each $t\in \{10,14,18,22\}$, take a $[2,1,1]$-GDC($6$) of type
$12^{t/2}$ (which exists by Lemma \ref{GDC(6)12t}). Adjoin two ideal
points and fill in the groups together with the two extra points with a $[2,1,1]$-GDC($6$) of type
$1^{12} 2^1$ (which exists by Proposition \ref{GDC(6)14_2}). The
result is an optimal $(6t+2,6,[2,1,1])_4$-code with size $6t^2+2t$.

For $t=13$, take an optimal $(10,6,[2,1,1])_4$-code with size $10$
which  exists by Proposition \ref{CCC(6)}. This code can also be
regarded as a $[2,1,1]$-GDC($6$) of type $1^{10}$. Apply Construction
\ref{Inflation} with a TD($4,8$) to get a $[2,1,1]$-GDC($6$) of type
$8^{10}$. Fill in the groups with an optimal
$(8,6,[2,1,1])_4$-code. The result is an optimal code of length
$80=6\times 13+2$.

For $t=17$, take a $[2,1,1]$-GDC($6$) of type $2^{13}$ (which exists
by Proposition \ref{GDC(6)gt}). Apply Construction \ref{Inflation}
with a TD($4,4$) to get a $[2,1,1]$-GDC($6$) of type $8^{13}$.
Fill in the groups with an optimal $(8,6,[2,1,1])_4$-code. The
result is an optimal code of length $104=6\times 17+2$.
\end{proof}
\vskip 10pt

\begin{theorem}
\label{CCC6t+2(3)} There exists an optimal $(6t+2,6,[2,1,1])_4$-code
with size $6t^2+2t$ for all $t\ge 23$ and $t\in\{12,15,16,$
$19,20,21\}$.
\end{theorem}
\begin{proof}
For each $t\geq 23$ or $t\in \{12,15,16,19,20,21\}$, take a
$(t+1,\{4,5,6\},1)$-PBD from Lemma \ref{PBD4-6} and remove one point
to obtain a $\{4,5,6\}$-GDD of type $3^i 4^j 5^k$ with $3i+4j+5k=t$.
Apply construction \ref{FundCtr} with weight $6$ and input
$[2,1,1]$-GDC($6$)s of types $6^t$ for $t\in \{4,5,6\}$ (Lemma
\ref{GDC(6)gt}) to obtain a $[2,1,1]$-GDC($6$) of type $18^i 24^j
30^k$ and length $6t$. Adjoining two ideal points and filling
in the groups together with the two extra points with $[2,1,1]$-GDC($6$)s of types $2^u$ for
$u\in \{10,13,16\}$ (see Proposition \ref{GDC(6)gt}), the result is a
$[2,1,1]$-GDC($6$) of type $2^{3t+1}$ for all $t\geq 23$ or $t\in
\{12,15,16,19,20,21\}$. By Lemma \ref{CCC6t+2(1)}, the resultant $[2,1,1]$-GDC($6$) is an
optimal $(6t+2,6,[2,1,1])_4$-code with size $6t^2+2t$.
\end{proof}
\vskip 10pt

\subsection{The Case of Length $n \equiv 3, 4 \pmod{6}$}
\vskip 10pt

\begin{lemma}
\label{CCC(6)4} $A_4(4,6,[2,1,1])=1$.
\end{lemma}
\begin{proof}
The one required codeword is $\langle 0, 1, 2, 3\rangle$.
\end{proof}

\vskip 10pt

\begin{lemma}
\label{CCC6t3} For any positive integer $t$, if
$A_4(6t+4,6,[2,1,1])=U(6t+4,6,[2,1,1])$ then
$A_4(6t+3,6,[2,1,1])=U(6t+3,6,[2,1,1])$.
\end{lemma}
\begin{proof}
In an optimal $(6t+4,6,[2,1,1])_4$-code, every coordinate has exactly $2t+1+2t+1=4t+2$ non-zero
elements. Fix a  coordinate $x$ and remove all the $4t+2$ codewords containing   non-zero
elements in this  coordinate $x$. Shorten all the remaining codewords by deleting the element $0$ in   coordinate $x$ from them.  The resultant
codewords form an optimal $(6t+3,6,[2,1,1])_4$-code with size
$U(6t+4,6,[2,1,1])-4t-2=6t^2+3t=U(6t+3,6,[2,1,1])$.
\end{proof}
\vskip 10pt

\begin{lemma}
\label{GDC12t91(1)} There exists a $[2,1,1]$-GDC($6$) of type $12^t
9^1$ for all $t\ge 4$.
\end{lemma}
\begin{proof}
For each $t\in \{4,5,\ldots,15\} \cup \{17,18,19,23\}$, a
$[2,1,1]$-GDC($6$) of type $12^t 9^1$ exists by Proposition
\ref{GDC(6)12tu1}.

For $t\in \{16, 21, 22, 26, 27, 28, 31, 32, 33\}$, we first construct $9$  GDCs  of types $48^4$, $48^4 60^1$, $48^4 72^1$, $48^5 72^1$, $60^4 84^1$, $84^4$, $48^6 84^1$, $84^4 48^1$ and  $84^4 60^1$ as follows: For $t=16$, take a TD($4,8$) from Lemma \ref{TD}. Apply Construction
\ref{FundCtr} with weight $6$ to obtain a $[2,1,1]$-GDC($6$) of type
$48^4$. For $t=21$, take a TD($5,5$) from Lemma \ref{TD}. Remove $4$ points
from a block to obtain a $\{4,5\}$-GDD of type $4^4 5^1$. Apply
Construction \ref{FundCtr} with weight $12$ to obtain a
$[2,1,1]$-GDC($6$) of type $48^4 60^1$.
For each $t\in\{22,26\}$, take a $4$-GDD of type $12^u 18^1$ for
$u\in\{4,5\}$ from Lemma \ref{4GDD}. Apply Construction \ref{FundCtr}
with weight $4$ to obtain a $[2,1,1]$-GDC($6$) of type $48^u 72^1$.
For $t=27$, take a $4$-GDD of type $15^4 21^1$ from Lemma \ref{4GDD}.
Apply Construction \ref{FundCtr} with weight $4$ to obtain a
$[2,1,1]$-GDC($6$) of type $60^4 84^1$.
For $t=28$, take a TD($4,14$) from  Lemma \ref{TD}. Apply Construction
\ref{FundCtr} with weight $6$ to obtain a $[2,1,1]$-GDC($6$) of type
$84^4$.
For $t=31$, take a $4$-GDD of type $4^6 7^1$ from  Lemma \ref{4GDD}.
Apply Construction \ref{FundCtr} with weight $12$ to obtain a
$[2,1,1]$-GDC($6$) of type $48^6 84^1$.
For each $t\in \{32,33\}$, take a TD($5,14$) from  Lemma \ref{TD}. Delete
$6$ or $4$ points from a group to get a $\{4,5\}$-GDD of types $14^4
8^1$ or $14^4 10^1$. Apply Construction \ref{FundCtr}
with weight $6$ to obtain a $[2,1,1]$-GDC($6$) of types $84^4 48^1$
or $84^4 60^1$.
Here, the input $[2,1,1]$-GDC($6$)s of types $4^4, 6^4, 12^4$ and $12^5$ exist by
Proposition \ref{GDC(6)gt}. Now, adjoin $9$ ideal points to each of the above GDCs and fill in
the groups together with the nine extra points with a $[2,1,1]$-GDC($6$) of type $12^v 9^1$ with $v \in \{4, 5, 6, 7\}$ (which
all exist by Proposition \ref{GDC(6)12tu1}) to obtain a
$[2,1,1]$-GDC($6$) of type $12^t 9^1$ as desired.

For all $t\geq 34$ or $t\in \{20,24,25,29,30\}$, take a
$(t+1,\{5,6,7,8,9\},1)$-PBD from Lemma \ref{PBD5-9} and remove one
point to obtain a $\{5,6,7,8,9\}$-GDD of type $4^i 5^j 6^k 7^l 8^m$
with $4i+5j+6k+7l+8m=t$. Apply construction \ref{FundCtr} with
weight $12$ and input $[2,1,1]$-GDC($6$)s of types $12^t$ for $t\in
\{5,6,7,8,9\}$ (Lemma \ref{GDC(6)12t}) to obtain a
$[2,1,1]$-GDC($6$) of type $48^i 60^j 72^k 84^l 96^m$ and length
$12t$. Adjoin $9$ ideal points and fill in the groups together with the nine extra points with
$[2,1,1]$-GDC($6$)s of types $12^u 9^1$ for $u\in \{4,5,6,7,8\}$
(which exist by Proposition \ref{GDC(6)12tu1}) to obtain a
$[2,1,1]$-GDC($6$) of type $12^t 9^1$.
\end{proof}
\vskip 10pt

\begin{theorem}
\label{CCC12t+10(1)} There exists an optimal
$(12t+10,6,[2,1,1])_4$-code with size $(4t+3)(6t+5)$ for all $t \ge
0$.
\end{theorem}
\begin{proof}
For each $t\in\{0,1,2\}$, there exists an optimal
$(12t+10,6,[2,1,1])_4$-code with size $(4t+3)(6t+5)$ by Proposition
\ref{CCC(6)}.

For $t=3$, take a $[2,1,1]$-GDC($6$) of type $3^5$ (which exists by
Proposition \ref{GDC(6)gt}). Apply Construction \ref{Inflation} with
a TD($4,3$) to get a $[2,1,1]$-GDC($6$) of type $9^5$.  Adjoin one
ideal point and fill in the groups together with the extra point with an optimal
$(10,6,[2,1,1])_4$-code. The result is an optimal code of length
$46=12\times 3+10$.

For all $t\ge 4$, there exists a $[2,1,1]$-GDC($6$) of type $12^t 9^1$.
Adjoin one ideal point to this GDC and fill in the groups together with the extra point with
optimal codes of lengths $10$ or $13$ to obtain an optimal
$(12t+10,6,[2,1,1])_4$-code with size $(4t+3)(6t+5)$.
\end{proof}
\vskip 10pt

\begin{theorem}
\label{CCC12t+16(1)} There exists an optimal
$(12t+16,6,[2,1,1])_4$-code with size $(4t+5)(6t+8)$ for each $0 \le t
\le 15$.
\end{theorem}
\begin{proof}
For each $t\in \{0,1\}$, there exists an optimal
$(12t+16,6,[2,1,1])_4$-code with size $(4t+5)(6t+8)$ by Proposition
\ref{CCC(6)}.

For $t=2$, take a $[2,1,1]$-GDC($6$) of type $10^4$ (which exists by
Proposition \ref{GDC(6)gt}). Fill in the groups with an optimal
$(10,6,[2,1,1])_4$-code. The result is an optimal code of length
$40=12\times 2+16$.

For $t=3$, take a $[2,1,1]$-GDC($6$) of type $13^4$ (which exists by
Proposition \ref{GDC(6)gt}). Fill in the groups with an optimal
$(13,6,[2,1,1])_4$-code. The result is an optimal code of length
$52=12\times 3+16$.

For $t=4$, take a $[2,1,1]$-GDC($6$) of type $3^7$ (which exists by
Proposition \ref{GDC(6)gt}). Apply Construction \ref{Inflation} with
a TD($4,3$) to get a $[2,1,1]$-GDC($6$) of type $9^7$. Adjoin one
ideal point and fill in the groups together with the extra point with an optimal
$(10,6,[2,1,1])_4$-code. The result is an optimal code of length
$64=12\times 4+16$.

For $t=5$, take a $[2,1,1]$-GDC($6$) of type $3^5$ (which exists by
Proposition \ref{GDC(6)gt}). Apply Construction \ref{Inflation} with
a TD($4,5$) to get a $[2,1,1]$-GDC($6$) of type $15^5$. Adjoin one
ideal point and fill in the groups together with the extra point with an optimal
$(16,6,[2,1,1])_4$-code. The result is an optimal code of length
$76=12\times 5+16$.

For $t=6$, take a $[2,1,1]$-GDC($6$) of type $22^4$ (which exists by
Proposition \ref{GDC(6)gt}). Fill in the groups with an optimal
$(22,6,[2,1,1])_4$-code. The result is an optimal code of length
$88=12\times 6+16$.

For each $t\in \{7,8,\ldots,15\}$, there exists a $[2,1,1]$-GDC(6) of
type $12^t 15^1$ by Proposition \ref{GDC(6)12tu1}. Adjoining one
ideal point and filling in the groups together with the extra point with optimal codes of lengths
$13$ or $16$, the result is an optimal $(12t+16,6,[2,1,1])_4$-code.
\end{proof}
\vskip 10pt

\begin{theorem}
\label{CCC12t+16(2)} There exists an optimal
$(12t+16,6,[2,1,1])_4$-code with size $(4t+5)(6t+8)$ for all $t \ge
16$.
\end{theorem}
\begin{proof}
For all $t\ge 16$, take a TD($6,2r$)  for $r\ge 4$ and $r\not\in \{5,7,9,11\}$  from Lemma \ref{TD}. Apply
Construction  \ref{FundCtr} with weight $6$ to the points in the
first $4$ groups, $2x$ points in the fifth group, and $y$ points in
the last group. We require that $x=0$ or $x\ge4$ and $1\le y \le 31$, $y$ odd. The other points are given weight 0. Note that there
exist $[2,1,1]$-GDC($6$)s of types $6^4$, $6^5$ and $6^6$ by Proposition
\ref{GDC(6)gt}. The result is a $[2,1,1]$-GDC($6$) of type $(12r)^4
(12x)^1 (6y)^1$.
Adjoining $9$ ideal points and filling in the groups together with the nine extra points with
$[2,1,1]$-GDC(6)s of types $12^u 9^1$ for all $u\ge 4$, the result is a
$[2,1,1]$-GDC(6) of type $12^{4r+x} (6y+9)^1$. Adjoin one more
ideal point and fill in the groups together with the extra point with optimal codes of lengths $13$
or $6y+10$ for $1\le y \le 31$ and $y$ odd (which exist
by Theorem \ref{CCC12t+16(1)}). The result is an optimal
$(48r+12x+6y+10,6,[2,1,1])_4$-code, where $48r+12x+6y+10$ can take
any value of $12t+16$ with $t$ greater than $16$.
\end{proof}
\vskip 10pt

\begin{theorem} There exists an optimal $(6t+3,6,[2,1,1])_4$-code with size $6t^2+3t$ for all $t\ge
1$.
\end{theorem}
\begin{proof}
The result follows by combining Lemma \ref{CCC6t3} and Theorems
\ref{CCC12t+10(1)}--\ref{CCC12t+16(2)}.
\end{proof}
\vskip 10pt

\subsection{The Case of Length $n \equiv 5 \pmod{6}$}
\vskip 10pt

\begin{lemma}
\label{CCC(6)5} $A_4(5,6,[2,1,1])=1$.
\end{lemma}
\begin{proof}
The one required codeword is $\langle 0, 1, 2, 3\rangle$.
\end{proof}

\vskip 10pt

\begin{theorem}
\label{CCC12t+11(1)} There exists an optimal
$(12t+11,6,[2,1,1])_4$-code with size $24t^2+40t+16$ for all $t \ge
0$.
\end{theorem}
\begin{proof}
For each $t\in \{0,1,2\}$, there exists an optimal
$(12t+11,6,[2,1,1])_4$-code with size $24t^2+40t+16$ by Proposition
\ref{CCC(6)}.

For $t=3$, there exists a $[2,1,1]$-GDC(6) of type $1^{36} 11^1$ by
Proposition \ref{GDC(6)1t111}. Fill in the groups with an optimal
$(11,6,[2,1,1])_4$-code to obtain an optimal
$(47,6,[2,1,1])_4$-code.

For all $t\ge 4$, there exists a $[2,1,1]$-GDC(6) of type $12^t 9^1$ by
Lemma \ref{GDC12t91(1)}. Adjoining $2$ ideal points and
filling in the groups together with the two extra points with a $[2,1,1]$-GDC(6) of type $1^{12} 2^1$
(which exists by Proposition \ref{GDC(6)14_2}) and an optimal
$(11,6,[2,1,1])_4$-code, the result is an optimal
$(12t+11,6,[2,1,1])_4$-code with size $24t^2+40t+16$.
\end{proof}
\vskip 10pt

\begin{theorem}
\label{CCC12t+17(1)} There exists an optimal
$(12t+17,6,[2,1,1])_4$-code with size $24t^2+64t+42$ for all $t \ge
0$.
\end{theorem}
\begin{proof}
For $t=0$, there exists an optimal $(17,6,[2,1,1])_4$-code with size
$42$ by Proposition \ref{CCC(6)}.

For $t=1$, take a $[2,1,1]$-GDC($6$) of type $7^4$ (which exists by
Proposition \ref{GDC(6)gt}). Adjoin one ideal point and fill in the
groups together with the extra point with an optimal $(8,6,[2,1,1])_4$-code. The result is an
optimal code of length $29=12\times 1+17$.

For each $t\in \{2,4,5,6\}$, there exists a $[2,1,1]$-GDC(6) of type
$1^{12t+6} 11^1$ by Proposition \ref{GDC(6)1t111}. Fill in the
groups with an optimal $(11,6,[2,1,1])_4$-code to obtain an optimal
$(12t+17,6,[2,1,1])_4$-code.

For $t=3$, take a $[2,1,1]$-GDC($6$) of type $13^4$ (which exists by
Proposition \ref{GDC(6)gt}). Adjoin one ideal point and fill in the
groups together with the extra point with an optimal $(14,6,[2,1,1])_4$-code. The result is an
optimal code of length $53=12\times 3+17$.

For each $t\in \{7,8,\ldots,15\}$, there exists a $[2,1,1]$-GDC(6) of
type $12^t 15^1$ by Proposition \ref{GDC(6)12tu1}. Adjoining $2$
ideal points and filling in the groups together with the two extra points with a $[2,1,1]$-GDC(6) of
type $1^{12} 2^1$ and an optimal code of length $17$, the result is
an optimal $(12t+17,6,[2,1,1])_4$-code.

For all $t\ge 16$, take a TD($6,2r$) for $r\ge 4$ and $r\not\in \{5,7,9,11\}$ from Lemma \ref{TD}. Apply
Construction  \ref{FundCtr} with weight $6$ to the points in the
first $4$ groups, $2x$ points in the fifth group, and $y$ points in
the last group. We require that $x=0$ or $x\ge 4$ and $1\le y
\le 31$, $y$ odd. The other points are given weight 0. Note that there
exist $[2,1,1]$-GDC($6$)s of types $6^4$, $6^5$ and $6^6$ by Proposition
\ref{GDC(6)gt}. The result is a $[2,1,1]$-GDC($6$) of type $(12r)^4
(12x)^1 (6y)^1$.
Adjoining $9$ ideal points and filling in the groups together with the nine extra points with
$[2,1,1]$-GDC(6)s of types $12^u 9^1$ for all $u\ge 4$, the result is a
$[2,1,1]$-GDC(6) of type $12^{4r+x} (6y+9)^1$. Adjoin $2$ more
ideal points and fill in the groups together with the two extra points with a $[2,1,1]$-GDC(6) of type
$1^{12} 2^1$ and an optimal code of length $6y+11$ with $1\le y
\le 31$ and $y$ odd (which exists by Theorem \ref{CCC12t+16(1)}).
The result is an optimal $(48r+12x+6y+11,6,[2,1,1])_4$-code, where
$48r+12x+6y+11$ can take any value of $12t+17$ with $t$
greater than $16$.
\end{proof}
\vskip 10pt

\section{Conclusion}
\vskip 10pt

In this paper, we determine almost completely the spectrum with
sizes for optimal quaternary constant-composition codes with
weight four and minimum distances five or six. We summarize our
main results of this paper as follows:

\vskip 10pt
\begin{theorem} For any integer $n \ge 4$
\begin{equation*}
A_4(n,5,[2,1,1])=\begin{cases}
1,&\text{if $n=4$} \\
2,&\text{if $n=5$} \\
6,&\text{if $n=6$} \\
10,&\text{if $n=7$} \\
n \left\lfloor \frac{n-1}{2}\right\rfloor,&\text{if  $n\ge 12$ and
$n \ne 13$.}\end{cases}
\end{equation*}
\begin{equation*}
A_4(8,5,[2,1,1])\ge 18,
\end{equation*}
\begin{equation*}
A_4(9,5,[2,1,1])\ge 27,
\end{equation*}
\begin{equation*}
A_4(10,5,[2,1,1])\ge 36,
\end{equation*}
\begin{equation*}
A_4(11,5,[2,1,1])\ge 48,
\end{equation*}
\begin{equation*}
A_4(13,5,[2,1,1])\ge 72.
\end{equation*}

\end{theorem}

\vskip 10pt
\begin{theorem} For any integer $n \ge 4$
\begin{equation*}
A_4(n,6,[2,1,1])=\begin{cases}
1,&\text{if $n=4,5$} \\
4,&\text{if $n=7$} \\
\left\lfloor \frac{n}{2} \left\lfloor \frac{n-1}{3} \right\rfloor
\right\rfloor, &\text{if  $n\ge 6$ and $n \ne 7$.}\end{cases}
\end{equation*}
\end{theorem}
\vskip 5pt

\section*{Acknowledgments}

The authors thank Prof. Olgica Milenkovic, the Associate Editor, Prof.
Helmut B{\"o}lcskei, the Editor-in-Chief of IEEE
Transactions on Information Theory, and the two anonymous
referees for their constructive comments and suggestions that
greatly improved the readability of this article.

\vskip 10pt
\bibliographystyle{IEEEtran}

\newpage

The following information is for referees only, not for publication.\\

\section{Base codewords for CCCs and GDCs with distance $5$ and type $[2,1,1]$}
\vskip 10pt

\begin{proposition}
There exists a $[2,1,1]$-GDC$(5)$ of type $2^t$ with size $2t(t-1)$ for each $t \in \{9,11\}$, which is also an optimal
$(2t,5,[2,1,1])_4$-code.
\end{proposition}
\begin{proof}For each $t \in \{9,11\}$, let $X_t=\bbZ_{2t}$,
$\G_t=\{\{i,t+i\}:i \in \bbZ_t\}$ and  $\C_t$ be the set of cyclic (or quasi-cyclic) shifts of the vectors generated by the following vectors respectively. Then $(X_t,\G_t,\C_t)$ is a
$[2,1,1]$-GDC$(5)$ of type $2^t$ with size $2t(t-1)$, where
\begin{enumerate}

\item $t=9$, $n=18$, $m=5$, $s=2$, $M=1$
\begin{equation*}
\begin{array}{cccc}
P:& \langle 0, 2, 1, 15\rangle\\
R:& \langle 0, 3, 15, 17\rangle & \langle 0, 11, 14, 1\rangle & \langle 0, 13, 6, 7\rangle\\
&\langle 0, 12, 4, 10\rangle & \langle 0, 1, 8, 5\rangle & \langle 0, 4, 2, 6\rangle\\
\end{array}
\end{equation*}

\item $t=11$, $n=22$, $m=3$, $s=5$, $M=2$
\begin{equation*}
\begin{array}{cccc}
P:& \langle 0, 1, 2, 3\rangle & \langle 0, 2, 7, 4\rangle & \langle 0, 7, 20, 19\rangle\\
&\langle 1, 3, 5, 8\rangle.\\
\end{array}
\end{equation*}
\end{enumerate}
\end{proof}
\vskip 10pt

\begin{proposition}
There exists a $[2,1,1]$-GDC$(5)$ of type $3^t$ with size $\frac{9t(t-1)}{2}$ for each $t \in \{7,9\}$.
\end{proposition}
\begin{proof}For each $t \in \{7,9\}$, let $X_t=\bbZ_{3t}$,
 $\G_t=\{\{i,t+i,2t+i\}:i \in \bbZ_t\}$ and  $\C_t$ be the set of cyclic shifts of the vectors generated by the following vectors respectively. Then $(X_t,\G_t,\C_t)$
is a $[2,1,1]$-GDC$(5)$ of type $3^t$ with size $\frac{9t(t-1)}{2}$,
where
\begin{enumerate}
\item $t=7$, $n=21$, $m=5$, $s=2$, $M=1$
\begin{equation*}
\begin{array}{cccc}
P:& \langle 0, 13, 16, 18\rangle & \langle 0, 12, 4, 8\rangle\\
R:& \langle 0, 10, 8, 2\rangle & \langle 0, 17, 1, 12\rangle & \langle 0, 15, 12, 3\rangle\\
&\langle 0, 1, 10, 11\rangle & \langle 0, 5, 11, 20\rangle\\
\end{array}
\end{equation*}

\item $t=9$, $n=27$, $m=2$, $s=3$, $M=1$
\begin{equation*}
\begin{array}{cccc}
P:& \langle 0, 1, 23, 6\rangle & \langle 0, 8, 12, 7\rangle\\
R:& \langle 0, 14, 15, 17\rangle & \langle 0, 21, 14, 2\rangle & \langle 0, 17, 3, 11\rangle\\
&\langle 0, 20, 26, 15\rangle & \langle 0, 15, 25, 19\rangle & \langle 0, 3, 5, 16\rangle.\\
\end{array}
\end{equation*}
\end{enumerate}
\end{proof}
\vskip 10pt

\begin{proposition}
There exists a $[2,1,1]$-GDC$(5)$ of type $4^t$ with size $8t(t-1)$ for each $t \in \{5,6,7,8,9,11\}$.
\end{proposition}
\begin{proof}For each $t \in \{5,6,7,8,9,11\}$, let $X_t=\bbZ_{4t}$,
 $\G_t=\{\{i,t+i,2t+i,3t+i\}:i \in \bbZ_t\}$  and  $\C_t$ be the set of cyclic shifts of the vectors generated by the following vectors respectively.  Then
$(X_t,\G_t,\C_t)$ is a $[2,1,1]$-GDC$(5)$ of type $4^t$ with size
$8t(t-1)$, where
\begin{enumerate}
\item $t=5$, $n=20$, $m=3$, $s=2$, $M=1$
\begin{equation*}
\begin{array}{cccc}
P:& \langle 0, 4, 11, 13\rangle&\langle 0, 6, 4, 2\rangle\\
R:& \langle 0, 9, 6, 3\rangle & \langle 0, 1, 9, 12\rangle & \langle 0, 3, 2, 1\rangle\\
&\langle 0, 7, 3, 4\rangle\\
\end{array}
\end{equation*}

\item $t=6$, $n=24$, $m=29$, $s=2$, $M=1$
\begin{equation*}
\begin{array}{cccc}
P:& \langle 0, 11, 9, 10\rangle & \langle 0, 10, 23, 21\rangle\\
R:& \langle 0, 15, 16, 8\rangle & \langle 0, 1, 4, 15\rangle & \langle 0, 20, 7, 16\rangle\\
&\langle 0, 3, 8, 4\rangle & \langle 0, 5, 20, 3\rangle & \langle 0, 8, 10, 13\rangle\\
\end{array}
\end{equation*}

\item $t=7$, $n=28$, $m=3$, $s=3$, $M=1$
\begin{equation*}
\begin{array}{cccc}
P:& \langle 0, 8, 5, 4\rangle & \langle 0, 13, 23, 15\rangle & \langle 0, 27, 3, 25\rangle\\
R:& \langle 0, 2, 22, 11\rangle & \langle 0, 6, 24, 5\rangle & \langle 0, 10, 26, 13\rangle\\
\end{array}
\end{equation*}

\item $t=8$, $n=32$, $m=5$, $s=3$, $M=1$
\begin{equation*}
\begin{array}{cccc}
P:& \langle 0, 3, 14, 21\rangle & \langle 0, 1, 27, 6\rangle\\
R:& \langle 0, 19, 12, 7\rangle & \langle 0, 4, 1, 19\rangle & \langle 0, 22, 5, 17\rangle\\
&\langle 0, 14, 31, 28\rangle & \langle 0, 20, 9, 11\rangle & \langle 0, 2, 20, 1\rangle\\
&\langle 0, 6, 28, 10\rangle & \langle 0, 9, 13, 12\rangle\\
\end{array}
\end{equation*}

\item $t=9$, $n=36$, $m=11$, $s=2$, $M=1$
\begin{equation*}
\begin{array}{cccc}
P:& \langle 0, 5, 21, 19\rangle & \langle 0, 32, 30, 7\rangle & \langle 0, 13, 33, 2\rangle\\
&\langle 0, 14, 24, 4\rangle\\
R:& \langle 0, 3, 31, 6\rangle & \langle 0, 34, 35, 30\rangle & \langle 0, 21, 26, 20\rangle\\
&\langle 0, 24, 13, 16\rangle & \langle 0, 6, 17, 21\rangle & \langle 0, 7, 29, 31\rangle\\
&\langle 0, 11, 19, 12\rangle & \langle 0, 16, 23, 33\rangle\\
\end{array}
\end{equation*}

\item $t=11$, $n=44$, $m=3$, $s=5$, $M=1$
\begin{equation*}
\begin{array}{cccc}
P:& \langle 0, 1, 2, 4\rangle & \langle 0, 8, 16, 25\rangle & \langle 0, 18, 41, 14\rangle\\
&\langle 0, 21, 19, 39\rangle.\\
\end{array}
\end{equation*}
\end{enumerate}
\end{proof}
\vskip 10pt

\begin{proposition}
There exists a $[2,1,1]$-GDC$(5)$ of type $6^5$ with size $360$.
\end{proposition}
\begin{proof}Let $X=\bbZ_{30}$,  $\G=\{\{i,5+i,10+i,15+i,20+i,25+i\}:i \in \bbZ_5\}$
 and  $\C$ be the set of cyclic  shifts of the vectors generated by the following vectors. Then $(X,\G,\C)$
is a $[2,1,1]$-GDC$(5)$ of type $6^5$ with size $360$, where $n=30$,
$m=13$, $s=2$, $M=1$ and
\begin{equation*}
\begin{array}{cccc}
P:& \langle 0, 9, 7, 3\rangle & \langle 0, 2, 29, 28\rangle\\
R:& \langle 0, 18, 26, 29\rangle & \langle 0, 13, 22, 19\rangle & \langle 0, 11, 14, 2\rangle\\
&\langle 0, 14, 16, 27\rangle & \langle 0, 7, 13, 14\rangle & \langle0, 22, 11, 23\rangle\\
&\langle 0, 29, 23, 17\rangle & \langle 0, 6, 18, 22\rangle.\\
\end{array}
\end{equation*}
\end{proof}
\vskip 10pt

\begin{proposition}
There exists a $[2,1,1]$-GDC$(5)$ of type $4^t2^1$
with size $8t^2$ for each $t \in \{5,6\}$.
\end{proposition}
\begin{proof}
For each $t \in \{5,6\}$, let $X_t=\bbZ_{4t} \cup \{a,b\}$,
$\G_t=\{\{i,t+i,2t+i,3t+i\}:i \in \bbZ_t\} \cup \{\{a,b\}\}$  and  $\C_t$ be the set of quasi-cyclic shifts of the vectors generated by the following vectors respectively.  Here, the elements $a,b$ keep fixed under the action of the automorphism group. Then
$(X_t,\G_t,\C_t)$ is a $[2,1,1]$-GDC$(5)$ of type $4^t2^1$ with size
$8t^2$, where
\begin{enumerate}

\item $t=5$, $m=3$, $s=2$, $M=2$
\begin{equation*}
\begin{array}{cccc}
P:& \langle 1, 19, 13, 7\rangle & \langle 1, 5, 8, 14\rangle & \langle 0, 4, 1, 7\rangle\\
&\langle 0, 2, 14, 8\rangle\\
R:& \langle 0, 17, a, 11\rangle & \langle 1, 12, 5, 4\rangle & \langle 0, 19, 7, 16\rangle\\
&\langle 1, 18, 7, a\rangle & \langle 1, 14, 0, b\rangle & \langle1, 8, b, 2\rangle\\
&\langle b, 0, 8, 19\rangle & \langle 0, 1, 18, 17\rangle & \langle b, 1, 19, 12\rangle\\
&\langle a, 0, 19, 2\rangle & \langle a, 1, 12, 13\rangle & \langle0, 11, 4, 13\rangle\\
\end{array}
\end{equation*}

\item $t=6$, $m=5$, $s=2$, $M=2$
\begin{equation*}
\begin{array}{cccc}
P:& \langle 0, 23, 10, 14\rangle & \langle 0, 13, 8, 11\rangle & \langle 0, 22, 7, 15\rangle\\
&\langle 0, 5, 15, 10\rangle\\
R:& \langle 0, 7, 4, 23\rangle & \langle a, 0, 20, 21\rangle & \langle 0, 20, 1, 4\rangle\\
&\langle 1, 23, 14, 3\rangle & \langle 1, 14, 9, 10\rangle & \langle 1, 17, 21, 12\rangle\\
&\langle b, 0, 23, 19\rangle & \langle 1, 22, 15, b\rangle & \langle 0, 8, 22, 9\rangle\\
&\langle 1, 10, 23, a\rangle & \langle 1, 11, 4, 21\rangle & \langle a, 1, 17, 14\rangle\\
&\langle b, 1, 10, 0\rangle & \langle 0, 21, b, 5\rangle & \langle 0, 9, a, 16\rangle\\
&\langle 1, 5, 6, 22\rangle.\\
\end{array}
\end{equation*}
\end{enumerate}
\end{proof}
\vskip 10pt

\begin{proposition}
$A_4(n,5,[2,1,1])$  =  $U(n,5,[2,1,1])$ for each $n \in
\{12,14,15,17,20,28,30,36,38,44,46,52,54,62,$
$68,70,76,78,86,92,94,110,126,134\}$.
\end{proposition}
\begin{proof}
For each given $n$, we take point set $X=\bbZ_n$. The codes are the sets of cyclic shifts of the vectors generated by the following vectors respectively.

\begin{enumerate}
\item $n=12$, $m=1$, $s=1$, $M=1$
\begin{equation*}
\begin{array}{cccc}
P:& \langle0, 1, 3, 9\rangle & \langle0, 2, 7, 5\rangle & \langle0, 3, 11, 7\rangle\\
&\langle0, 4, 10, 2\rangle & \langle0, 5, 9, 11\rangle\\
\end{array}
\end{equation*}

\item $n=14$, $m=9$, $s=3$, $M=1$
\begin{equation*}
\begin{array}{cccc}
P:& \langle 0, 1, 2, 10\rangle\\
R:& \langle 0, 4, 10, 7\rangle & \langle 0, 6, 13, 5\rangle & \langle0, 2, 5, 4\rangle\\
\end{array}
\end{equation*}

\item $n=15$, $m=1$, $s=1$, $M=1$
\begin{equation*}
\begin{array}{cccc}
P:& \langle0, 1, 2, 4\rangle & \langle0, 2, 12, 9\rangle & \langle0, 3, 9, 13\rangle\\
&\langle0, 4, 7, 6\rangle & \langle0, 5, 13, 1\rangle & \langle0, 6, 5, 14\rangle\\
&\langle0, 7, 11, 12\rangle\\
\end{array}
\end{equation*}


\item $n=17$, $m=2$, $s=2$, $M=1$
\begin{equation*}
\begin{array}{cccc}
P:& \langle 0, 1, 7, 6\rangle\\
R:& \langle 0, 6, 2, 9\rangle & \langle 0, 8, 11, 15\rangle & \langle0, 3, 1, 2\rangle\\
&\langle 0, 10, 9, 4\rangle & \langle 0, 4, 8, 1\rangle & \langle 0,5, 10, 13\rangle\\
\end{array}
\end{equation*}

\item $n=20$, $m=3$, $s=2$, $M=1$
\begin{equation*}
\begin{array}{cccc}
P:& \langle 0, 18, 9, 16\rangle & \langle 0, 4, 12, 17\rangle &\langle 0, 3, 5, 4\rangle\\
R:& \langle 0, 7, 10, 2\rangle & \langle 0, 1, 18, 7\rangle & \langle0, 5, 19, 10\rangle\\
\end{array}
\end{equation*}


\item $n=28$, $m=5$, $s=2$, $M=1$
\begin{equation*}
\begin{array}{cccc}
P:& \langle 0, 1, 2, 4\rangle & \langle 0, 2, 11, 8\rangle & \langle 0, 3, 18, 16\rangle\\
R:& \langle 0, 21, 24, 19\rangle & \langle 0, 16, 20, 21\rangle & \langle 0, 4, 26,18\rangle\\
&\langle 0, 6, 14, 23\rangle & \langle 0, 20, 13, 27\rangle & \langle 0, 9, 16, 10\rangle\\
&\langle 0, 11, 23, 22\rangle\\
\end{array}
\end{equation*}

\item $n=30$, $m=23$, $s=3$, $M=1$
\begin{equation*}
\begin{array}{cccc}
P:& \langle 0, 1, 2, 4\rangle & \langle 0, 2, 11, 14\rangle\\
R:& \langle 0, 25, 12, 24\rangle & \langle 0, 26, 6, 7\rangle &\langle 0, 9, 5, 19\rangle\\
&\langle 0, 12, 7, 13\rangle & \langle 0, 24, 22, 17\rangle & \langle0, 13, 3, 21\rangle\\
&\langle 0, 27, 15, 25\rangle & \langle 0, 10, 24, 15\rangle\\
\end{array}
\end{equation*}

\item $n=36$, $m=31$, $s=2$, $M=1$
\begin{equation*}
\begin{array}{cccc}
P:& \langle 0, 1, 2, 4\rangle & \langle 0, 2, 5, 9\rangle & \langle 0, 3, 13, 14\rangle\\
&\langle 0, 4, 27, 23\rangle & \langle 0, 11, 15, 8\rangle\\
R:& \langle 0, 22, 18, 10\rangle & \langle 0, 9, 34, 31\rangle & \langle 0, 23, 35,12\rangle\\
&\langle 0, 29, 17, 28\rangle & \langle 0, 8, 28, 6\rangle & \langle 0, 6, 14, 26\rangle\\
&\langle 0, 12, 6, 30\rangle\\
\end{array}
\end{equation*}

\item $n=38$, $m=13$, $s=4$, $M=1$
\begin{equation*}
\begin{array}{cccc}
P:& \langle 0, 32, 26, 21\rangle & \langle 0, 30, 5, 14\rangle & \langle 0, 7, 37, 13\rangle\\
R:& \langle 0, 14, 15, 11\rangle & \langle 0, 21, 35, 29\rangle & \langle 0, 29, 2,33\rangle\\
&\langle 0, 37, 28, 18\rangle & \langle 0, 25, 20, 37\rangle & \langle 0, 3, 22,28\rangle\\
\end{array}
\end{equation*}

\item $n=44$, $m=5$, $s=4$, $M=1$
\begin{equation*}
\begin{array}{cccc}
P:& \langle 0, 1, 2, 4\rangle & \langle 0, 2, 14, 7\rangle & \langle0, 3, 24, 21\rangle\\
R:& \langle 0, 24, 33, 32\rangle & \langle 0, 36, 39, 28\rangle & \langle 0, 27, 23, 40\rangle\\
&\langle 0, 16, 35, 14\rangle & \langle 0, 33, 7, 11\rangle & \langle 0, 18, 38,1\rangle\\
&\langle 0, 4, 15, 34\rangle & \langle 0, 9, 22, 33\rangle & \langle 0, 12, 43, 38\rangle\\
\end{array}
\end{equation*}

\item $n=46$, $m=7$, $s=4$, $M=1$
\begin{equation*}
\begin{array}{cccc}
P:& \langle 0, 1, 2, 4\rangle & \langle 0, 2, 12, 20\rangle & \langle0, 5, 32, 29\rangle\\
R:& \langle 0, 30, 41, 31\rangle & \langle 0, 17, 34, 16\rangle & \langle 0, 8, 16,43\rangle\\
&\langle 0, 26, 13, 33\rangle & \langle 0, 19, 44, 13\rangle & \langle 0, 22, 45, 15\rangle\\
&\langle 0, 34, 31, 25\rangle & \langle 0, 9, 18, 36\rangle & \langle 0, 36, 29,22\rangle\\
&\langle 0, 18, 37, 23\rangle\\
\end{array}
\end{equation*}

\item $n=52$, $m=7$, $s=5$, $M=1$
\begin{equation*}
\begin{array}{cccc}
P:& \langle 0, 1, 2, 4\rangle & \langle 0, 2, 5, 24\rangle & \langle 0, 4, 20, 5\rangle\\
R:& \langle 0, 44, 30, 2\rangle & \langle 0, 13, 19, 26\rangle & \langle 0, 25, 12,44\rangle\\
&\langle 0, 17, 32, 11\rangle & \langle 0, 37, 11, 15\rangle & \langle 0, 22, 47, 17\rangle\\
&\langle 0, 5, 34, 39\rangle & \langle 0, 41, 22, 18\rangle & \langle 0, 19, 17,33\rangle\\
&\langle 0, 23, 13, 48\rangle\\
\end{array}
\end{equation*}

\item $n=54$, $m=29$, $s=4$, $M=1$
\begin{equation*}
\begin{array}{cccc}
P:& \langle 0, 1, 2, 4\rangle & \langle 0, 2, 5, 12\rangle & \langle 0, 5, 11, 1\rangle\\
&\langle 0, 10, 32, 53\rangle\\
R:& \langle 0, 45, 43, 6\rangle & \langle 0, 12, 53, 30\rangle & \langle 0, 32, 23,34\rangle\\
&\langle 0, 51, 18, 11\rangle & \langle 0, 36, 28, 27\rangle & \langle 0, 24, 50, 41\rangle\\
&\langle 0, 48, 19, 7\rangle & \langle 0, 11, 38, 9\rangle & \langle 0, 33, 42, 28\rangle\\
&\langle 0, 15, 30, 36\rangle\\
\end{array}
\end{equation*}

\item $n=62$, $m=7$, $s=7$, $M=1$
\begin{equation*}
\begin{array}{cccc}
P:& \langle 0, 1, 2, 4\rangle & \langle 0, 2, 17, 42\rangle & \langle 0, 3, 59, 51\rangle\\
R:& \langle 0, 22, 46, 39\rangle & \langle 0, 32, 22, 61\rangle & \langle 0, 46, 44, 59\rangle\\
&\langle 0, 42, 11, 5\rangle & \langle 0, 38, 34, 31\rangle & \langle 0, 50, 42,45\rangle\\
&\langle 0, 53, 9, 41\rangle & \langle 0, 18, 48, 53\rangle & \langle 0, 6, 12, 33\rangle\\
\end{array}
\end{equation*}

\item $n=68$, $m=7$, $s=7$, $M=1$
\begin{equation*}
\begin{array}{cccc}
P:& \langle 0, 1, 2, 4\rangle & \langle 0, 5, 64, 19\rangle & \langle 0, 2, 29, 66\rangle\\
R:& \langle 0, 58, 62, 51\rangle & \langle 0, 40, 10, 25\rangle & \langle 0, 4, 58,17\rangle\\
&\langle 0, 60, 28, 41\rangle & \langle 0, 24, 50, 1\rangle & \langle 0, 36, 12,2\rangle\\
&\langle 0, 51, 34, 20\rangle & \langle 0, 16, 33, 15\rangle & \langle 0, 12, 60, 43\rangle\\
&\langle 0, 20, 66, 7\rangle & \langle 0, 25, 20, 48\rangle & \langle 0, 29, 45,39\rangle\\
\end{array}
\end{equation*}

\item $n=70$, $m=13$, $s=4$, $M=1$
\begin{equation*}
\begin{array}{cccc}
P:& \langle 0, 1, 2, 4\rangle & \langle 0, 2, 6, 9\rangle & \langle 0, 3, 12, 16\rangle\\
&\langle 0, 4, 3, 22\rangle & \langle 0, 6, 37, 64\rangle\\
R:& \langle 0, 56, 18, 61\rangle & \langle 0, 61, 36, 28\rangle & \langle 0, 51, 5,23\rangle\\
&\langle 0, 50, 10, 25\rangle & \langle 0, 60, 50, 31\rangle & \langle 0, 47, 25, 57\rangle\\
&\langle 0, 45, 66, 40\rangle & \langle 0, 40, 35, 59\rangle & \langle 0, 63, 55, 53\rangle\\
&\langle 0, 49, 42, 35\rangle & \langle 0, 65, 15, 50\rangle & \langle 0, 28, 56, 43\rangle\\
&\langle 0, 15, 64, 14\rangle & \langle 0, 33, 40, 30\rangle\\
\end{array}
\end{equation*}

\item $n=76$, $m=25$, $s=7$, $M=1$
\begin{equation*}
\begin{array}{cccc}
P:& \langle 0, 1, 2, 4\rangle & \langle 0, 2, 6, 7\rangle & \langle 0,4, 7, 10\rangle\\
&\langle 0, 7, 40, 48\rangle\\
R:& \langle 0, 9, 18, 38\rangle & \langle 0, 73, 70, 64\rangle & \langle 0, 63, 56,52\rangle\\
&\langle 0, 58, 36, 45\rangle & \langle 0, 6, 64, 61\rangle & \langle 0, 21, 53,57\rangle\\
&\langle 0, 57, 19, 8\rangle & \langle 0, 64, 55, 46\rangle & \langle 0, 36, 63,14\rangle\\
\end{array}
\end{equation*}

\item $n=78$, $m=11$, $s=5$, $M=1$
\begin{equation*}
\begin{array}{cccc}
P:& \langle 0, 17, 76, 34\rangle & \langle 0, 9, 15, 56\rangle & \langle 0, 35, 42,2\rangle\\
&\langle 0, 14, 63, 18\rangle & \langle 0, 18, 28, 73\rangle\\
R:& \langle 0, 32, 5, 75\rangle & \langle 0, 13, 29, 9\rangle & \langle 0, 40, 26, 66\rangle\\
&\langle 0, 41, 60, 30\rangle & \langle 0, 30, 44, 3\rangle & \langle 0, 58, 2, 15\rangle\\
&\langle 0, 50, 27, 24\rangle & \langle 0, 62, 37, 5\rangle & \langle 0, 11, 31,39\rangle\\
&\langle 0, 26, 65, 13\rangle & \langle 0, 15, 38, 48\rangle & \langle 0, 44, 52, 50\rangle\\
&\langle 0, 4, 17, 40\rangle\\
\end{array}
\end{equation*}

\item $n=86$, $m=9$, $s=7$, $M=1$
\begin{equation*}
\begin{array}{cccc}
P:& \langle 0, 3, 15, 25\rangle & \langle 0, 1, 2, 37\rangle & \langle 0, 6, 16, 11\rangle\\
&\langle 0, 2, 5, 74\rangle\\
R:& \langle 0, 72, 70, 68\rangle & \langle 0, 28, 51, 77\rangle & \langle 0, 57, 85, 27\rangle\\
&\langle 0, 63, 29, 6\rangle & \langle 0, 52, 35, 51\rangle & \langle 0, 16, 46,59\rangle\\
&\langle 0, 17, 77, 12\rangle & \langle 0, 42, 24, 10\rangle & \langle 0, 67, 38, 76\rangle\\
&\langle 0, 40, 83, 4\rangle & \langle 0, 48, 67, 3\rangle & \langle 0, 65, 59, 83\rangle\\
&\langle 0, 7, 72, 15\rangle & \langle 0, 35, 79, 20\rangle\\
\end{array}
\end{equation*}

\item $n=92$, $m=7$, $s=11$, $M=1$
\begin{equation*}
\begin{array}{cccc}
P:& \langle 0, 1, 2, 4\rangle & \langle 0, 16, 8, 13\rangle & \langle 0, 2, 47, 90\rangle\\
R:& \langle 0, 21, 90, 51\rangle & \langle 0, 47, 38, 9\rangle & \langle 0, 35, 46,69\rangle\\
&\langle 0, 75, 62, 6\rangle & \langle 0, 33, 74, 14\rangle & \langle 0, 23, 66,1\rangle\\
&\langle 0, 3, 25, 49\rangle & \langle 0, 27, 85, 2\rangle & \langle 0, 61, 19, 79\rangle\\
&\langle 0, 39, 29, 81\rangle & \langle 0, 5, 91, 15\rangle & \langle 0, 37, 23,63\rangle\\
\end{array}
\end{equation*}

\item $n=94$, $m=5$, $s=11$, $M=1$
\begin{equation*}
\begin{array}{cccc}
P:& \langle 0, 64, 45, 75\rangle & \langle 0, 21, 14, 9\rangle & \langle 0, 80, 24,64\rangle\\
R:& \langle 0, 31, 66, 67\rangle & \langle 0, 23, 78, 26\rangle & \langle 0, 88, 15, 23\rangle\\
&\langle 0, 89, 67, 81\rangle & \langle 0, 61, 64, 47\rangle & \langle 0, 1, 30,53\rangle\\
&\langle 0, 15, 9, 66\rangle & \langle 0, 19, 51, 54\rangle & \langle 0, 9, 53, 57\rangle\\
&\langle 0, 3, 84, 24\rangle & \langle 0, 25, 77, 74\rangle & \langle 0, 45, 92,77\rangle\\
&\langle 0, 37, 48, 15\rangle\\
\end{array}
\end{equation*}


\item $n=110$, $m=57$, $s=9$, $M=1$
\begin{equation*}
\begin{array}{cccc}
P:& \langle 0, 38, 25, 2\rangle & \langle 0, 17, 1, 85\rangle & \langle 0, 14, 26, 13\rangle\\
&\langle 0, 69, 10, 59\rangle\\
R:& \langle 0, 60, 66, 106\rangle & \langle 0, 100, 7, 82\rangle &\langle 0, 15, 68, 44\rangle\\
&\langle 0, 105, 11, 50\rangle & \langle 0, 37, 55, 15\rangle &\langle 0, 74, 36, 108\rangle\\
&\langle 0, 25, 83, 103\rangle & \langle 0, 20, 33, 97\rangle &\langle 0, 65, 32, 48\rangle\\
&\langle 0, 22, 44, 11\rangle & \langle 0, 70, 49, 54\rangle &\langle 0, 80, 4,27\rangle\\
&\langle 0, 99, 107, 102\rangle & \langle 0, 44, 43, 1\rangle &\langle 0, 19, 88, 41\rangle\\
&\langle 0, 75, 60, 21\rangle & \langle 0, 18, 99, 79\rangle &\langle 0, 33, 64, 66\rangle\\
\end{array}
\end{equation*}

\item $n=126$, $m=11$, $s=6$, $M=1$
\begin{equation*}
\begin{array}{cccc}
P:& \langle 0, 50, 87, 20\rangle & \langle 0, 69, 58, 10\rangle &\langle 0, 113, 52, 2\rangle\\
&\langle 0, 52, 100, 35\rangle &\langle 0, 96, 41, 11\rangle & \langle 0, 106, 78, 44\rangle\\
& \langle 0, 103, 110, 75\rangle & \langle 0, 37, 57, 68\rangle\\
R:& \langle 0, 108, 99, 42\rangle & \langle 0, 7, 90, 108\rangle &\langle 0, 99,89, 78\rangle\\
&\langle 0, 49, 108, 54\rangle & \langle 0, 81, 97, 99\rangle &\langle 0, 84, 21, 72\rangle\\
&\langle 0, 98, 18, 125\rangle &\langle 0, 35, 54, 90\rangle &\langle 0, 70, 72, 115\rangle\\
&\langle 0, 14, 50, 117\rangle & \langle 0, 90, 9, 84\rangle &\langle 0, 117, 22, 21\rangle\\
& \langle 0, 21, 105, 102\rangle & \langle 0, 54, 81, 63\rangle  \\
\end{array}
\end{equation*}

\item $n=134$, $m=117$, $s=10$, $M=1$
\begin{equation*}
\begin{array}{cccc}
P:& \langle 0, 5, 81, 30\rangle & \langle 0, 120, 32, 132\rangle &\langle 0, 79, 24, 74\rangle\\
&\langle 0, 23, 63, 62\rangle & \langle 0, 72, 25, 9\rangle\\
R:& \langle 0, 87, 74, 101\rangle & \langle 0, 112, 84, 75\rangle &\langle 0, 56, 61, 67\rangle\\
&\langle 0, 118, 75, 20\rangle & \langle 0, 66, 104, 93\rangle &\langle 0, 82, 49, 54\rangle\\
&\langle 0, 50, 35, 131\rangle & \langle 0, 13, 19, 78\rangle &\langle 0, 90, 72, 8\rangle \\
&\langle 0, 54, 65, 112\rangle & \langle 0, 60, 83, 79\rangle & \langle0, 46, 10,77\rangle\\
&\langle 0, 114,94, 68\rangle & \langle 0, 106, 80, 23\rangle &\langle 0, 4, 109, 90\rangle\\
&\langle 0, 15, 82, 50\rangle.\\
\end{array}
\end{equation*}
\end{enumerate}
\end{proof}
\vskip 10pt

\begin{proposition}
$A_4(n,5,[2,1,1])$  =  $U(n,5,[2,1,1])$ for each $n \in
\{21,23,24,25,26,27,29,31,32,33,34,35,$
$37,39,40,41,42,43,45,47,48,49,50,51,53,55,56,57,58,$
$59,61,63,64,65,66,67,69,71,73,74,75,77,79,83,85,87,$
$88,89,93,95,99,103,104,106,107,109,111,123,125,127,$
$131,133,138,139\}$.
\end{proposition}
\begin{proof}
All these optimal codes are constructed by strong starters or strong
frame starters with n odd or even respectively.
The starters are given in a similar way as the codewords in the above propositions.

\begin{enumerate}



\item $n=21$, $m=1$, $s=1$, $M=1$
\begin{equation*}
\begin{array}{ccc}
P:& \{3, 19\}, \{2, 14\}, \{6, 13\}, \{8, 10\}, \{15, 16\},\\
&\{11, 17\}, \{5, 9\}, \{1, 4\}, \{7, 18\}, \{12, 20\}\\
\end{array}
\end{equation*}

\item $n=23$, $m=2$, $s=11$, $M=1$
\begin{equation*}
\begin{array}{ccc}
P:& \{1, 5\}\\
\end{array}
\end{equation*}

\item $n=24$, $m=5$, $s=2$, $M=1$
\begin{equation*}
\begin{array}{ccc}
P:& \{ 1, 10\}, \{ 8, 21\}\\
R:& \{ 3, 7\}, \{ 4, 18\}, \{ 6, 11\}, \{ 13, 15\}, \{ 14, 22\},\\
&\{ 17, 23\}, \{ 19, 20\}\\
\end{array}
\end{equation*}

\item $n=25$, $m=2$, $s=2$, $M=1$
\begin{equation*}
\begin{array}{ccc}
P:& \{ 1, 8\}, \{ 6, 9\}, \{ 7, 15\}, \{ 10, 23\}\\
R:& \{ 3, 24\}, \{ 4, 19\}, \{ 11, 13\}, \{ 17, 22\}\\
\end{array}
\end{equation*}

\item $n=26$, $m=3$, $s=3$, $M=1$
\begin{equation*}
\begin{array}{ccc}
P:& \{ 1, 5\}, \{ 2, 10\}, \{ 7, 22\}, \{ 8, 25\}\\
\end{array}
\end{equation*}

\item $n=27$, $m=5$, $s=2$, $M=1$
\begin{equation*}
\begin{array}{ccc}
P:& \{ 13, 17\}, \{ 7, 22\}, \{ 6, 14\}\\
R:& \{ 1, 10\}, \{ 5, 21\}, \{ 9, 12\}, \{ 15, 25\}, \{ 18, 23\},\\
&\{ 19, 20\}, \{ 24, 26\}\\
\end{array}
\end{equation*}

\item $n=29$, $m=7$, $s=7$, $M=1$
\begin{equation*}
\begin{array}{ccc}
P:& \{ 1, 3\}, \{ 4, 10\}\\
\end{array}
\end{equation*}

\item $n=31$, $m=7$, $s=15$, $M=1$
\begin{equation*}
\begin{array}{ccc}
P:& \{ 1, 3\}\\
\end{array}
\end{equation*}

\item $n=32$, $m=3$, $s=2$, $M=1$
\begin{equation*}
\begin{array}{ccc}
P:& \{ 13, 18\}, \{ 19, 30\}, \{ 3, 31\}, \{ 11, 14\}, \{ 8, 27\},\\
&\{ 2, 12\}\\
R:& \{ 5, 23\}, \{ 15, 21\}, \{ 20, 28\}\\
\end{array}
\end{equation*}

\item $n=33$, $m=5$, $s=2$, $M=1$
\begin{equation*}
\begin{array}{ccc}
P:& \{ 1, 15\}, \{ 6, 11\}, \{ 16, 29\}, \{ 17, 24\}, \{ 7, 31\},\\
&\{ 4, 27\}, \{ 25, 28\}\\
R:& \{ 10, 32\}, \{ 12, 18\}\\
\end{array}
\end{equation*}

\item $n=34$, $m=3$, $s=4$, $M=1$
\begin{equation*}
\begin{array}{ccc}
P:& \{ 1, 4\}, \{ 13, 18\}, \{ 16, 30\}, \{ 21, 33\}\\
\end{array}
\end{equation*}

\item $n=35$, $m=2$, $s=3$, $M=1$
\begin{equation*}
\begin{array}{ccc}
P:& \{ 1, 3\}, \{ 7, 20\}, \{ 22, 25\}, \{ 29, 34\}\\
R:& \{ 8, 24\}, \{ 13, 27\}, \{ 16, 17\}, \{ 19, 26\}, \{ 21, 32\}\\
\end{array}
\end{equation*}

\item $n=37$, $m=7$, $s=9$, $M=1$
\begin{equation*}
\begin{array}{ccc}
P:& \{ 1, 3\}, \{ 2, 5\}\\
\end{array}
\end{equation*}

\item $n=39$, $m=4$, $s=3$, $M=1$
\begin{equation*}
\begin{array}{ccc}
P:& \{ 11, 34\}, \{ 2, 31\}, \{ 21, 27\}, \{ 18, 29\}\\
R:& \{ 1, 4\}, \{ 9, 16\}, \{ 10, 12\}, \{ 13, 22\}, \{ 14, 26\},\\
&\{ 17, 25\}, \{ 23, 36\}\\
\end{array}
\end{equation*}

\item $n=40$, $m=3$, $s=2$, $M=1$
\begin{equation*}
\begin{array}{ccc}
P:& \{ 1, 18\}, \{ 4, 7\}, \{ 9, 36\}, \{ 17, 35\}, \{ 22, 24\},\\
&\{ 16, 23\}, \{ 31, 39\}, \{ 2, 30\}, \{ 33, 38\}\\
R:& \{ 5, 15\}\\
\end{array}
\end{equation*}

\item $n=41$, $m=10$, $s=5$, $M=1$
\begin{equation*}
\begin{array}{ccc}
P:& \{ 1, 3\}, \{ 2, 5\}, \{ 4, 14\}, \{ 11, 35\}\\
\end{array}
\end{equation*}

\item $n=42$, $m=23$, $s=2$, $M=1$
\begin{equation*}
\begin{array}{ccc}
P:& \{ 1, 18\}, \{ 5, 38\}, \{ 4, 6\}, \{ 22, 27\}, \{ 7, 25\},\\
&\{ 9, 41\}, \{ 16, 17\}, \{ 20, 28\}\\
R:& \{ 3, 15\}, \{ 10, 24\}, \{ 11, 26\}, \{ 30, 37\}\\
\end{array}
\end{equation*}

\item $n=43$, $m=9$, $s=21$, $M=1$
\begin{equation*}
\begin{array}{ccc}
P:& \{ 1, 3\}\\
\end{array}
\end{equation*}

\item $n=45$, $m=2$, $s=3$, $M=1$
\begin{equation*}
\begin{array}{ccc}
P:& \{ 13, 23\}, \{ 4, 28\}, \{ 18, 31\}, \{ 6, 20\}, \{ 21, 44\}\\
R:& \{ 3, 15\}, \{ 5, 32\}, \{ 9, 25\}, \{ 10, 19\}, \{ 14, 29\},\\
&\{ 30, 38\}, \{ 33, 37\}\\
\end{array}
\end{equation*}

\item $n=47$, $m=2$, $s=23$, $M=1$
\begin{equation*}
\begin{array}{ccc}
P: \{ 1, 5\}\\
\end{array}
\end{equation*}

\item $n=48$, $m=29$, $s=2$, $M=1$
\begin{equation*}
\begin{array}{ccc}
P:& \{ 19, 46\}, \{ 39, 43\}, \{ 13, 32\}, \{ 14, 21\}, \{ 2, 45\},\\
&\{ 15, 29\}, \{ 37, 40\}, \{ 11, 28\}, \{ 5, 35\}\\
R:& \{ 4, 6\}, \{ 12, 20\}, \{ 18, 34\}, \{ 26, 36\}, \{ 30, 42\}\\
\end{array}
\end{equation*}

\item $n=49$, $m=10$, $s=11$, $M=1$
\begin{equation*}
\begin{array}{ccc}
P:& \{ 1, 3\}\\
R:& \{ 5, 29\}, \{ 27, 46\}, \{ 17, 18\}, \{ 42, 45\}, \{ 7, 34\},\\
&\{ 21, 33\}, \{ 35, 41\}, \{ 23, 37\}, \{ 25, 36\}, \{ 15, 43\},\\
&\{ 19, 26\}, \{ 9, 14\}, \{ 28, 38\}\\
\end{array}
\end{equation*}

\item $n=50$, $m=7$, $s=2$, $M=1$
\begin{equation*}
\begin{array}{ccc}
P:& \{ 9, 27\}, \{ 6, 36\}, \{ 4, 49\}, \{ 12, 46\}, \{ 23, 31\},\\
&\{ 38, 45\}, \{ 18, 20\}, \{ 1, 14\}, \{ 19, 41\}, \{ 5, 44\},\\
&\{ 3, 32\}, \{ 30, 47\}\\
\end{array}
\end{equation*}

\item $n=51$, $m=2$, $s=3$, $M=1$
\begin{equation*}
\begin{array}{ccc}
P:& \{ 7, 23\}, \{ 1, 31\}, \{ 27, 47\}, \{ 25, 32\}, \{ 9, 10\},\\
& \{ 42, 48\}\\
R:& \{ 5, 30\}, \{ 8, 44\}, \{ 12, 17\}, \{ 16, 24\}, \{ 19, 29\},\\
&\{ 21, 38\}, \{ 34, 37\}\\
\end{array}
\end{equation*}

\item $n=53$, $m=10$, $s=13$, $M=1$
\begin{equation*}
\begin{array}{ccc}
P:& \{ 1, 3\}, \{ 4, 14\}\\
\end{array}
\end{equation*}

\item $n=55$, $m=2$, $s=7$, $M=1$
\begin{equation*}
\begin{array}{ccc}
P:& \{ 1, 3\}, \{ 5, 18\}\\
R:& \{ 15, 42\}, \{ 49, 54\}, \{ 23, 44\}, \{ 28, 38\}, \{ 11, 51\},\\
&\{ 31, 53\}, \{ 30, 47\}, \{ 14, 33\}, \{ 29, 43\}, \{ 7, 37\},\\
&\{ 19, 39\}, \{ 21, 22\}, \{ 35, 46\}\\
\end{array}
\end{equation*}

\item $n=56$, $m=3$, $s=3$, $M=1$
\begin{equation*}
\begin{array}{ccc}
P:& \{ 30, 38\}, \{ 11, 44\}, \{ 12, 17\}, \{ 39, 40\}, \{ 16, 18\},\\
&\{ 1, 53\}\\
R:& \{ 7, 21\}, \{ 10, 31\}, \{ 13, 35\}, \{ 14, 45\}, \{ 19, 26\},\\
&\{ 22, 49\}, \{ 23, 42\}, \{ 25, 55\}, \{ 27, 37\}\\
\end{array}
\end{equation*}

\item $n=57$, $m=4$, $s=4$, $M=1$
\begin{equation*}
\begin{array}{ccc}
P:& \{ 13, 30\}, \{ 8, 31\}, \{ 11, 53\}, \{ 35, 43\}, \{ 3, 7\}\\
R:& \{ 2, 22\}, \{ 9, 27\}, \{ 15, 45\}, \{ 18, 51\}, \{ 19, 25\},\\
&\{ 23, 42\}, \{ 33, 54\}, \{ 36, 38\}\\
\end{array}
\end{equation*}

\item $n=58$, $m=3$, $s=7$, $M=1$
\begin{equation*}
\begin{array}{ccc}
P:& \{ 1, 12\}, \{ 17, 30\}, \{ 28, 46\}, \{ 41, 57\}\\
\end{array}
\end{equation*}

\item $n=59$, $m=3$, $s=29$, $M=1$
\begin{equation*}
\begin{array}{ccc}
P:& \{1, 6\}\\
\end{array}
\end{equation*}

\item $n=61$, $m=12$, $s=15$, $M=1$
\begin{equation*}
\begin{array}{ccc}
P:& \{ 1, 3\}, \{ 2, 6\}\\
\end{array}
\end{equation*}

\item $n=63$, $m=11$, $s=5$, $M=1$
\begin{equation*}
\begin{array}{ccc}
P:& \{ 7, 24\}, \{ 15, 38\}, \{ 8, 20\}, \{ 32, 61\}\\
R:& \{ 2, 9\}, \{ 16, 42\}, \{ 17, 30\}, \{ 18, 43\}, \{ 19, 54\},\\
&\{ 21, 45\}, \{ 22, 58\}, \{ 27, 48\}, \{ 35, 53\}, \{ 36, 50\},\\
& \{ 46, 55\}\\
\end{array}
\end{equation*}

\item $n=64$, $m=11$, $s=4$, $M=1$
\begin{equation*}
\begin{array}{ccc}
P:& \{ 21, 62\}, \{ 3, 17\}, \{ 37, 57\}, \{ 5, 44\}, \{ 2, 20\}\\
R:& \{ 1, 10\}, \{ 6, 7\}, \{ 8, 15\}, \{ 11, 24\}, \{ 13, 48\},\\
&\{ 16, 53\}, \{ 18, 58\}, \{ 19, 34\}, \{ 30, 41\}, \{ 40, 56\},\\
&\{ 46, 54\}\\
\end{array}
\end{equation*}

\item $n=65$, $m=7$, $s=4$, $M=1$
\begin{equation*}
\begin{array}{ccc}
P:& \{ 30, 33\}, \{ 26, 49\}, \{ 48, 58\}, \{ 54, 63\}, \{ 43, 50\}\\
R:& \{ 1, 7\}, \{ 2, 10\}, \{ 3, 22\}, \{ 5, 6\}, \{ 8, 28\},\\
&\{ 14, 42\}, \{ 17, 44\}, \{ 21, 34\}, \{ 23, 35\}, \{ 24, 60\},\\
&\{ 31, 56\}, \{ 38, 64\}\\
\end{array}
\end{equation*}

\item $n=66$, $m=13$, $s=5$, $M=1$
\begin{equation*}
\begin{array}{ccc}
P:& \{ 3, 23\}, \{ 37, 56\}, \{ 7, 58\}, \{ 24, 65\}\\
R:& \{ 6, 18\}, \{ 9, 27\}, \{ 10, 12\}, \{ 11, 51\}, \{ 14, 20\},\\
&\{ 16, 40\}, \{ 21, 55\}, \{ 22, 50\}, \{ 32, 62\}, \{ 36, 44\},\\
&\{ 42, 64\}, \{ 52, 63\}\\
\end{array}
\end{equation*}

\item $n=67$, $m=4$, $s=33$, $M=1$
\begin{equation*}
\begin{array}{ccc}
P:& \{1, 3\}\\
\end{array}
\end{equation*}

\item $n=69$, $m=2$, $s=5$, $M=1$
\begin{equation*}
\begin{array}{ccc}
P:& \{ 3, 44\}, \{ 10, 54\}, \{ 56, 57\}, \{ 30, 67\}, \{ 2, 50\}\\
R:& \{ 1, 23\}, \{ 5, 28\}, \{ 13, 64\}, \{ 25, 49\}, \{ 26, 29\},\\
&\{ 27, 63\}, \{ 35, 46\}, \{ 47, 59\}, \{ 52, 58\}\\
\end{array}
\end{equation*}

\item $n=71$, $m=2$, $s=35$, $M=1$
\begin{equation*}
\begin{array}{ccc}
P:& \{ 1, 7\}\\
\end{array}
\end{equation*}

\item $n=73$, $m=2$, $s=9$, $M=1$
\begin{equation*}
\begin{array}{ccc}
P:& \{ 1, 3\}, \{ 5, 15\}, \{ 9, 33\}, \{ 13, 35\}\\
\end{array}
\end{equation*}

\item $n=74$, $m=3$, $s=9$, $M=1$
\begin{equation*}
\begin{array}{ccc}
P:& \{ 1, 4\}, \{ 13, 56\}, \{ 18, 70\}, \{ 61, 73\}\\
\end{array}
\end{equation*}

\item $n=75$, $m=17$, $s=7$, $M=1$
\begin{equation*}
\begin{array}{ccc}
P:& \{ 1, 3\}, \{ 7, 16\}, \{ 24, 43\}\\
R:& \{ 25, 65\}, \{ 31, 35\}, \{ 14, 45\}, \{ 60, 70\}, \{ 23, 71\},\\
&\{ 37, 66\}, \{ 9, 69\}, \{ 27, 72\}, \{ 13, 20\}, \{ 2, 15\},\\
&\{ 11, 29\}, \{ 34, 40\}, \{ 10, 53\}, \{ 48, 68\}, \{ 50, 55\},\\
&\{ 5, 30\}\\
\end{array}
\end{equation*}

\item $n=77$, $m=2$, $s=5$, $M=1$
\begin{equation*}
\begin{array}{ccc}
P:& \{ 9, 32\}, \{ 1, 57\}, \{ 28, 47\}, \{ 45, 48\}, \{ 40, 49\}\\
R:& \{ 5, 31\}, \{ 10, 23\}, \{ 11, 43\}, \{ 15, 55\}, \{ 20, 54\},\\
&\{ 22, 30\}, \{ 29, 62\}, \{ 33, 53\}, \{ 39, 61\}, \{ 41, 66\},\\
&\{ 44, 60\}, \{ 46, 73\}, \{ 58, 69\}\\
\end{array}
\end{equation*}

\item $n=79$, $m=2$, $s=39$, $M=1$
\begin{equation*}
\begin{array}{ccc}
P:& \{ 1, 3\}\\
\end{array}
\end{equation*}


\item $n=83$, $m=3$, $s=41$, $M=1$
\begin{equation*}
\begin{array}{ccc}
P:& \{ 1, 5\}\\
\end{array}
\end{equation*}

\item $n=85$, $m=22$, $s=13$, $M=1$
\begin{equation*}
\begin{array}{ccc}
P:& \{ 1, 3\}, \{ 2, 6\}\\
R:& \{ 34, 60\}, \{ 39, 74\}, \{ 37, 68\}, \{ 49, 62\}, \{ 45, 50\},\\
&\{ 52, 75\}, \{ 15, 26\}, \{ 10, 70\}, \{ 17, 80\}, \{ 4, 55\},\\
&\{ 30, 40\}, \{ 8, 25\}, \{ 13, 58\}, \{ 31, 51\}, \{ 5, 20\},\\
&\{ 35, 65\}\\
\end{array}
\end{equation*}

\item $n=87$, $m=2$, $s=11$, $M=1$
\begin{equation*}
\begin{array}{ccc}
P:& \{ 57, 74\}, \{ 44, 71\}, \{ 14, 69\}\\
R:& \{ 7, 22\}, \{ 9, 67\}, \{ 11, 24\}, \{ 18, 48\}, \{ 29, 79\},\\
&\{ 31, 47\}, \{ 36, 62\}, \{ 37, 72\}, \{ 49, 85\}, \{ 58, 83\}\\
\end{array}
\end{equation*}

\item $n=88$, $m=7$, $s=5$, $M=1$
\begin{equation*}
\begin{array}{ccc}
P:& \{ 9, 67\}, \{ 48, 65\}, \{ 23, 30\}, \{ 21, 69\}, \{ 79, 84\},\\
&\{ 38, 58\}\\
R:& \{ 4, 18\}, \{ 5, 16\}, \{ 11, 24\}, \{ 12, 22\}, \{ 19, 52\},\\
&\{ 20, 86\}, \{ 28, 55\}, \{ 32, 50\}, \{ 33, 35\}, \{ 40, 78\},\\
&\{ 45, 66\}, \{ 51, 80\}, \{ 74, 77\}\\
\end{array}
\end{equation*}

\item $n=89$, $m=2$, $s=11$, $M=1$
\begin{equation*}
\begin{array}{ccc}
P:& \{ 1, 3\}, \{ 5, 15\}, \{ 9, 22\}, \{ 19, 65\}\\
\end{array}
\end{equation*}

\item $n=93$, $m=14$, $s=12$, $M=1$
\begin{equation*}
\begin{array}{ccc}
P:& \{ 35, 80\}, \{ 48, 85\}, \{ 23, 66\}\\
R:& \{ 3, 77\}, \{ 5, 8\}, \{ 18, 26\}, \{ 19, 49\}, \{ 30, 81\},\\
&\{ 31, 86\}, \{ 39, 70\}, \{ 42, 55\}, \{ 46, 50\}, \{ 62, 88\}\\
\end{array}
\end{equation*}

\item $n=95$, $m=21$, $s=7$, $M=1$
\begin{equation*}
\begin{array}{ccc}
P:& \{ 20, 51\}, \{ 39, 63\}, \{ 2, 37\}, \{ 81, 94\}, \{ 83, 90\}\\
R:& \{ 3, 9\}, \{ 5, 24\}, \{ 6, 38\}, \{ 8, 73\}, \{ 10, 67\},\\
&\{ 11, 14\}, \{ 13, 29\}, \{ 19, 41\}, \{ 25, 76\}, \{ 31, 57\},\\
&\{ 47, 62\}, \{ 50, 77\}\\
\end{array}
\end{equation*}

%

\item $n=99$, $m=20$, $s=8$, $M=1$
\begin{equation*}
\begin{array}{ccc}
P:& \{ 1, 3\}, \{ 2, 17\}, \{ 21, 38\}, \{ 34, 41\}\\
R:& \{ 72, 88\}, \{ 6, 30\}, \{ 5, 27\}, \{ 7, 83\}, \{ 26, 81\},\\
&\{ 35, 44\}, \{ 77, 95\}, \{ 22, 33\}, \{ 11, 75\}, \{ 15, 36\},\\
&\{ 54, 58\}, \{ 9, 63\}, \{ 18, 45\}, \{ 71, 90\}, \{ 19, 55\},\\
&\{ 25, 91\}, \{ 66, 76\}\\
\end{array}
\end{equation*}

\item $n=103$, $m=2$, $s=51$, $M=1$
\begin{equation*}
\begin{array}{ccc}
P:& \{ 1, 3\}\\
\end{array}
\end{equation*}

\item $n=104$, $m=11$, $s=6$, $M=1$
\begin{equation*}
\begin{array}{ccc}
P:& \{ 27, 59\}, \{ 12, 56\}, \{ 45, 48\}, \{ 31, 50\}, \{ 70, 92\},\\
&\{ 2, 47\}\\
R:& \{ 1, 39\}, \{ 6, 35\}, \{ 10, 23\}, \{ 11, 13\}, \{ 17, 78\},\\
&\{ 19, 93\}, \{ 21, 86\}, \{ 26, 83\}, \{ 41, 51\}, \{ 65, 91\},\\
&\{ 66, 73\}, \{ 75, 102\}, \{ 81, 87\}, \{ 82, 97\}, \{ 85, 103\}\\
\end{array}
\end{equation*}

\item $n=106$, $m=7$, $s=13$, $M=1$
\begin{equation*}
\begin{array}{ccc}
P:& \{ 1, 3\}, \{ 8, 24\}, \{ 82, 103\}, \{ 98, 105\}\\
\end{array}
\end{equation*}

\item $n=107$, $m=3$, $s=53$, $M=1$
\begin{equation*}
\begin{array}{ccc}
P:& \{ 1, 5\}\\
\end{array}
\end{equation*}

\item $n=109$, $m=2$, $s=18$, $M=1$
\begin{equation*}
\begin{array}{ccc}
P:& \{ 1, 3\}, \{ 50, 106\}, \{ 59, 108\}\\
\end{array}
\end{equation*}

\item $n=111$, $m=7$, $s=8$, $M=1$
\begin{equation*}
\begin{array}{ccc}
P:& \{ 2, 95\}, \{ 13, 80\}, \{ 4, 86\}, \{ 36, 106\}, \{ 15, 70\},\\
&\{ 12, 45\}\\
R:& \{ 10, 18\}, \{ 21, 31\}, \{ 32, 77\}, \{ 37, 64\}, \{ 44, 81\},\\
&\{ 54, 74\}, \{ 59, 97\}\\
\end{array}
\end{equation*}



\item $n=123$, $m=5$, $s=9$, $M=1$
\begin{equation*}
\begin{array}{ccc}
P:& \{ 76, 113\}, \{ 56, 116\}, \{ 26, 114\}, \{ 45, 117\}, \{ 2, 6\},\\
&\{ 16, 92\}\\
R:& \{ 5, 121\}, \{ 9, 41\}, \{ 25, 75\}, \{ 43, 82\}, \{ 48, 89\},\\
&\{ 51, 85\}, \{ 67, 79\}\\
\end{array}
\end{equation*}

\item $n=125$, $m=2$, $s=10$, $M=1$
\begin{equation*}
\begin{array}{ccc}
P:& \{ 1, 11\}, \{ 10, 67\}, \{ 28, 31\}, \{ 29, 111\}, \{ 108, 115\}\\
R:& \{ 14, 48\}, \{ 17, 78\}, \{ 24, 100\}, \{ 25, 96\}, \{ 34, 59\},\\
&\{ 39, 71\}, \{ 47, 119\}, \{ 50, 77\}, \{ 63, 101\}, \{ 75, 92\},\\
&\{ 68, 118\}, \{ 94, 113\}\\
\end{array}
\end{equation*}

\item $n=127$, $m=9$, $s=63$, $M=1$
\begin{equation*}
\begin{array}{ccc}
P:& \{ 1, 3\}\\
\end{array}
\end{equation*}


\item $n=131$, $m=3$, $s=65$, $M=1$
\begin{equation*}
\begin{array}{ccc}
P:& \{ 1, 6\}\\
\end{array}
\end{equation*}

\item $n=133$, $m=11$, $s=3$, $M=1$
\begin{equation*}
\begin{array}{ccc}
P:& \{ 1, 3\}, \{ 2, 5\}, \{ 89, 124\}, \{ 6, 10\}, \{ 7, 8\},\\
&\{ 80, 87\}, \{ 24, 69\}, \{ 15, 31\}, \{ 56, 95\}, \{ 17, 30\},\\
&\{ 23, 105\}, \{ 28, 45\}, \{ 29, 47\}, \{ 46, 98\}, \{ 35, 67\},\\
&\{ 58, 128\}, \{ 40, 59\}, \{ 25, 85\}, \{ 16, 36\}, \{ 12, 79\},\\
&\{ 19, 50\}, \{ 68, 93\}\\
\end{array}
\end{equation*}

\item $n=138$, $m=19$, $s=10$, $M=1$
\begin{equation*}
\begin{array}{ccc}
P:& \{ 1, 3\}, \{ 2, 5\}, \{ 4, 12\}, \{ 8, 81\}\\
R:& \{ 37, 78\}, \{ 7, 96\}, \{ 43, 115\}, \{ 77, 131\}, \{ 41, 66\},\\
&\{ 47, 124\}, \{ 6, 89\}, \{ 13, 59\}, \{ 34, 87\}, \{ 73, 92\},\\
&\{ 18, 48\}, \{ 17, 52\}, \{ 23, 127\}, \{ 31, 91\}, \{ 44, 135\},\\
&\{ 70, 121\}, \{ 83, 84\}, \{ 10, 104\}, \{ 30, 72\}, \{ 65, 88\},\\
&\{ 67, 126\}, \{ 35, 53\}, \{ 16, 118\}, \{ 22, 28\}, \{ 46, 94\},\\
&\{ 113, 130\}, \{ 102, 114\}, \{ 109, 133\}\\
\end{array}
\end{equation*}

\item $n=139$, $m=4$, $s=69$, $M=1$
\begin{equation*}
\begin{array}{ccc}
P:& \{ 1, 3\}.\\
\end{array}
\end{equation*}

\end{enumerate}
\end{proof}
\vskip 10pt

\begin{proposition}
$A_4(8,5,[2,1,1]) \ge 18$.
\end{proposition}
\begin{proof}
For $n=8$, the $18$  required codewords are:
\begin{equation*}
\begin{array}{cccc}
\langle 3, 7, 5, 2\rangle & \langle 2, 5, 1, 7\rangle & \langle 1,
3, 6, 7\rangle & \langle 2, 3, 4, 1\rangle\\
\langle 0, 5, 4, 2\rangle & \langle 6, 7, 2, 1\rangle & \langle 0,
7, 6, 5\rangle  & \langle 4, 6, 3, 2\rangle\\
\langle 0, 6, 1, 4\rangle & \langle 5, 6, 7, 3\rangle & \langle 3,
4, 0, 5\rangle  & \langle 0, 3, 2, 6\rangle\\
\langle 4, 5, 6, 1\rangle & \langle 4, 7, 1, 3\rangle & \langle 2,
7, 0, 4\rangle  & \langle 0, 1, 5, 3\rangle\\
\langle 1, 4, 2, 0\rangle & \langle 2, 6, 5, 0\rangle.\\
\end{array}
\end{equation*}
\end{proof}

\begin{proposition}
$A_4(5,5,[2,1,1])=2$, $A_4(9,5,[2,1,1])\ge 27$
and $A_4(13,5,[2,1,1])\ge 72$.
\end{proposition}
\begin{proof}
For $n=5$, the $2$ required codewords are:
\begin{equation*}
\begin{array}{ccc}
\langle 0, 1, 2, 3\rangle & \langle 2, 3, 0, 4\rangle\\
\end{array}
\end{equation*}

For $n=9$, the $27$  required codewords are:
\begin{equation*}
\begin{array}{cccc}
\langle 0, 5, 4, 1\rangle & \langle 1, 8, 0, 3\rangle & \langle 0,
6, 2, 3\rangle & \langle 2, 3, 0, 4\rangle\\
\langle 1, 2, 6, 0\rangle & \langle 0, 8, 6, 4\rangle & \langle 4, 5, 7, 8\rangle & \langle 6, 8, 1, 2\rangle\\
\langle 5, 8, 2, 0\rangle & \langle 1, 3, 2, 7\rangle & \langle 0, 2, 7, 5\rangle & \langle 1, 5, 8, 6\rangle\\
\langle 2, 7, 3, 8\rangle & \langle 0, 3, 5, 8\rangle & \langle 7,
8, 5, 1\rangle & \langle 3, 4, 8, 1\rangle\\
\langle 4, 7, 1, 6\rangle & \langle 2, 5, 1, 3\rangle & \langle 4, 6, 0, 5\rangle & \langle 6, 7, 8, 0\rangle\\
\langle 5, 7, 6, 2\rangle & \langle 0, 4, 3, 7\rangle & \langle 1, 4, 5, 2\rangle & \langle 1, 6, 3, 4\rangle\\
\langle 2, 6, 4, 7\rangle & \langle 3, 7, 4, 5\rangle & \langle 3, 8, 7, 6\rangle\\
\end{array}
\end{equation*}

For $n=13$, the $72$  required  codewords are all quasi-cyclic shifts with length $2$ of the following vectors, where the element $a$ keeps fixed under the action of the automorphism group.
\begin{equation*}
\begin{array}{cccc}
\langle a, 0, 3, 2\rangle & \langle a, 9, 2, 7\rangle & \langle 0,
11, a, 8\rangle & \langle 0, 7, 2, a\rangle\\
\langle 0, 5, 4, 7\rangle & \langle 0, 3, 5, 11\rangle & \langle 0, 4, 11, 1\rangle & \langle 0, 9, 6, 3\rangle\\
\langle 0, 10, 8, 4\rangle & \langle 1, 7, 5, 2\rangle & \langle 1, 3, 4, 6\rangle & \langle 0, 1, 9, 5\rangle.\\
\end{array}
\end{equation*}
\end{proof}
\vskip 10pt

\begin{proposition}
$A_4(6,5,[2,1,1])=6$, $A_4(10,5,[2,1,1]) \ge 36$.
\end{proposition}
\begin{proof}
For $n=6$, the $6$ required codewords are:
\begin{equation*}
\begin{array}{cccc}
\langle 0, 1, 2, 3\rangle & \langle 0, 2, 4, 5\rangle & \langle 1,
3, 5, 4\rangle & \langle 2, 4, 3, 1\rangle\\
\langle 3, 5, 0, 2\rangle & \langle 4, 5, 1, 0\rangle\\
\end{array}
\end{equation*}
For $n=10$, the $36$  required codewords are:
\begin{equation*}
\begin{array}{cccc}
\langle 8, 9, 4, 1\rangle & \langle 4, 8, 7, 6\rangle & \langle 4,
5, 2, 7\rangle & \langle 1, 7, 8, 4\rangle\\
\langle 5, 9, 3, 4\rangle & \langle 0, 8, 1, 7\rangle & \langle 1, 8, 6, 3\rangle & \langle 1, 6, 7, 0\rangle\\
\langle 1, 9, 2, 6\rangle & \langle 3, 7, 0, 1\rangle & \langle 4, 9, 6, 0\rangle & \langle 6, 9, 8, 7\rangle\\
\langle 6, 7, 2, 3\rangle & \langle 3, 9, 7, 2\rangle & \langle 0,
1, 5, 9\rangle & \langle 2, 8, 5, 4\rangle\\
\langle 1, 2, 3, 7\rangle & \langle 2, 7, 4, 0\rangle & \langle 0, 4, 8, 5\rangle & \langle 4, 6, 5, 1\rangle\\
\langle 0, 3, 2, 4\rangle & \langle 3, 4, 9, 8\rangle & \langle 1, 4, 0, 2\rangle & \langle 0, 2, 6, 1\rangle\\
\langle 7, 8, 3, 9\rangle & \langle 0, 6, 3, 8\rangle & \langle 0,
5, 7, 3\rangle & \langle 5, 7, 6, 8\rangle\\
\langle 3, 6, 1, 9\rangle & \langle 2, 3, 8, 6\rangle & \langle 7, 9, 1, 5\rangle & \langle 5, 8, 9, 0\rangle\\
\langle 0, 7, 9, 6\rangle & \langle 1, 3, 4, 5\rangle & \langle 2, 6, 9, 5\rangle & \langle 2, 9, 0, 8\rangle.\\
\end{array}
\end{equation*}
\end{proof}
\vskip 10pt

\begin{proposition}
$A_4(7,5,[2,1,1])=10$, $A_4(11,5,[2,1,1])\ge 48$.
\end{proposition}
\begin{proof}
For $n=7$, the $10$ required codewords are:
\begin{equation*}
\begin{array}{cccc}
\langle 0, 1, 2, 3\rangle & \langle 0, 2, 4, 5\rangle & \langle 0,
3, 5, 6\rangle & \langle 0, 4, 6, 1\rangle\\
\langle 1, 4, 3, 5\rangle & \langle 1, 6, 4, 2\rangle & \langle 2, 3, 1, 4\rangle & \langle 2, 6, 3, 0\rangle\\
\langle 4, 5, 2, 6\rangle & \langle 5, 6, 1, 3\rangle\\
\end{array}
\end{equation*}

For $n=11$, the $48$  required  codewords are:
\begin{equation*}
\begin{array}{cccc}
\langle 4, 6, 10, 9\rangle & \langle 3, 10, 2, 0\rangle & \langle 0,
8, 2, 9\rangle & \langle 3, 8, 10, 5\rangle\\
\langle 1, 8, 5, 2\rangle & \langle 0, 2, 1, 5\rangle & \langle 6, 8, 9, 1\rangle & \langle 3, 4, 1, 2\rangle\\
\langle 2, 8, 7, 4\rangle & \langle 2, 10, 8, 1\rangle & \langle 6, 10, 1, 7\rangle & \langle 5, 8, 6, 7\rangle\\
\langle 0, 10, 9, 4\rangle & \langle 2, 4, 0, 10\rangle & \langle 7,
8, 1, 3\rangle & \langle 3, 9, 5, 4\rangle\\
\langle 8, 9, 4, 10\rangle & \langle 5, 9, 1, 8\rangle & \langle 4, 7, 6, 1\rangle & \langle 1, 4, 9, 5\rangle\\
\langle 5, 7, 0, 4\rangle & \langle 1, 2, 3, 7\rangle & \langle 8, 10, 0, 6\rangle & \langle 6, 9, 0, 2\rangle\\
\langle 7, 10, 5, 8\rangle & \langle 0, 4, 5, 6\rangle & \langle 4,
10, 7, 3\rangle & \langle 4, 8, 3, 0\rangle\\
\langle 1, 3, 8, 9\rangle & \langle 3, 5, 7, 1\rangle & \langle 2, 9, 10, 3\rangle & \langle 0, 6, 7, 8\rangle\\
\langle 5, 10, 4, 2\rangle & \langle 1, 9, 7, 6\rangle & \langle 0, 1, 6, 10\rangle & \langle 0, 3, 4, 7\rangle\\
\langle 9, 10, 6, 5\rangle & \langle 3, 7, 9, 10\rangle & \langle 0,
5, 8, 3\rangle & \langle 1, 6, 4, 3\rangle\\
\langle 2, 5, 9, 6\rangle & \langle 0, 7, 10, 2\rangle & \langle 4, 9, 2, 7\rangle & \langle 7, 9, 8, 0\rangle\\
\langle 5, 6, 3, 10\rangle & \langle 0, 9, 3, 1\rangle & \langle 2, 3, 6, 8\rangle & \langle 6, 7, 2, 5\rangle.\\
\end{array}
\end{equation*}
\end{proof}
\vskip 10pt

\section{Base codewords for CCCs and GDCs with distance $6$ and type $[2,1,1]$}
\vskip 10pt
\begin{table*}
\centering
\begin{tabular}{cccccccc}
\hline $\langle 6, 10, 8, 16\rangle$ & $\langle 9, 15, 12,
16\rangle$ & $\langle 2, 12,
 1, 9\rangle$ & $\langle 6, 14, 9, 13\rangle$ & $\langle 11, 13, 10, 6\rangle$ &
 $\langle 0, 4, 7, 14\rangle$ & $\langle 2, 15, 14, 6\rangle$ & $\langle 0, 15,
11, 3\rangle$\\
$\langle 1, 6, 12, 11\rangle$ & $\langle 5, 12, 0, 4\rangle$ &
$\langle 3, 11, 8 , 9\rangle$ & $\langle 5, 10, 9, 11\rangle$ &
$\langle 0, 1, 16, 2\rangle$ & $\langle 1, 3, 15, 5\rangle$ &
$\langle 7, 8, 10, 2\rangle$ & $\langle 3, 13, 12, 0
\rangle$\\
$\langle 6, 7, 15, 0\rangle$ & $\langle 10, 15, 1, 4\rangle$ &
$\langle 11, 12, 16, 14\rangle$ & $\langle 3, 16, 7, 6\rangle$ &
$\langle 8, 14, 4, 11\rangle$ & $\langle 5, 13, 1, 14\rangle$ &
$\langle 10, 12, 7, 13\rangle$ & $\langle 7, 11,
 5, 1\rangle$\\
$\langle 0, 8, 13, 5\rangle$ & $\langle 1, 4, 8, 13\rangle$ &
$\langle 0, 9, 6, 10\rangle$ & $\langle 4, 16, 10, 12\rangle$ &
$\langle 4, 9, 5, 3\rangle$ & $\langle 5, 15, 8, 7\rangle$ &
$\langle 8, 16, 9, 1\rangle$ & $\langle 5, 6, 3, 2\rangle$\\
$\langle 4, 11, 15, 2\rangle$ & $\langle 1, 9, 14, 7\rangle$ &
$\langle 2, 16, 11, 13\rangle$ & $\langle 7, 13, 16, 4\rangle$ &
$\langle 7, 14, 3, 12\rangle$ & $\langle 8, 12, 6, 3\rangle$ &
$\langle 14, 16, 5, 15\rangle$ & $\langle 9, 13,
2, 8\rangle$\\
$\langle 10, 14, 2, 0\rangle$ & $\langle 2, 3, 4, 10\rangle$\\
\hline
\end{tabular}
\caption{Codewords of an Optimal $(17,6,[2,1,1])_4$-code with Size
$42$.} \label{17,6,[2,1,1]}
\end{table*}

\begin{table*}
\centering
\begin{tabular}{cccccccc}
\hline $\langle 3, 10, 2, 15\rangle$ & $\langle 14, 22, 3, 4\rangle$
& $\langle 7, 8, 21, 9\rangle$ & $\langle 20, 21, 5, 7\rangle$ &
$\langle 5, 14, 15, 1\rangle$ & $ \langle 15, 18, 21, 10\rangle$ &
$\langle 5, 13, 18, 20\rangle$ & $\langle 0, 9,
 14, 15\rangle$\\
$\langle 8, 13, 1, 15\rangle$ & $\langle 6, 21, 10, 14\rangle$ &
$\langle 4, 9, 18, 8\rangle$ & $\langle 17, 21, 19, 8\rangle$ &
$\langle 8, 11, 20, 6\rangle$ &
 $\langle 11, 22, 21, 0\rangle$ & $\langle 2, 16, 9, 3\rangle$ & $\langle 5, 12,
 0, 4\rangle$\\
$\langle 3, 18, 20, 22\rangle$ & $\langle 3, 8, 12, 14\rangle$ &
$\langle 14, 16 , 0, 11\rangle$ & $\langle 15, 19, 13, 7\rangle$ &
$\langle 20, 22, 8, 13\rangle $ & $\langle 5, 7, 17, 2\rangle$ &
$\langle 0, 20, 3, 12\rangle$ & $\langle 4, 21, 22, 15\rangle$\\
$\langle 12, 22, 16, 7\rangle$ & $\langle 14, 19, 6, 18\rangle$ &
$\langle 5, 19 , 11, 21\rangle$ & $\langle 10, 20, 9, 11\rangle$ &
$\langle 7, 16, 18, 15\rangle$ & $\langle 1, 21, 18, 2\rangle$ &
$\langle 2, 12, 6, 1\rangle$ & $\langle 6,
20, 15, 19\rangle$\\
$\langle 12, 14, 21, 20\rangle$ & $\langle 1, 11, 16, 9\rangle$ &
$\langle 9, 19 , 10, 2\rangle$ & $\langle 4, 20, 14, 2\rangle$ &
$\langle 19, 22, 1, 17\rangle$
 & $\langle 10, 12, 17, 19\rangle$ & $\langle 8, 18, 2, 16\rangle$ & $\langle 0,
 10, 7, 21\rangle$\\
$\langle 0, 6, 11, 2\rangle$ & $\langle 2, 14, 8, 7\rangle$ &
$\langle 6, 9, 13,
 22\rangle$ & $\langle 3, 19, 4, 0\rangle$ & $\langle 1, 7, 22, 20\rangle$ & $\langle 3, 9, 7, 1\rangle$ & $\langle 0, 4, 16, 13\rangle$ & $\langle
4, 7, 6, 10\rangle$\\
$\langle 12, 18, 9, 13\rangle$ & $\langle 2, 13, 19, 4\rangle$ &
$\langle 1, 15,
 3, 6\rangle$ & $\langle 8, 10, 4, 5\rangle$ & $\langle 11, 13, 10, 12\rangle$ &
 $\langle 0, 8, 17, 22\rangle$ & $\langle 17, 18, 7, 11\rangle$ & $\langle 3, 21
, 13, 11\rangle$\\
$\langle 2, 11, 5, 18\rangle$ & $\langle 4, 11, 17, 3\rangle$ &
$\langle 2, 17, 22, 10\rangle$ & $\langle 1, 4, 5, 19\rangle$ &
$\langle 15, 22, 5, 9\rangle$ & $\langle 5, 6, 8, 3\rangle$ &
$\langle 16, 19, 8, 12\rangle$ & $\langle 15, 17,
4, 14\rangle$\\
$\langle 1, 10, 14, 13\rangle$ & $\langle 16, 20, 17, 1\rangle$ &
$\langle 9, 17 , 20, 5\rangle$ & $\langle 12, 15, 11, 8\rangle$ &
$\langle 10, 22, 18, 6\rangle $ & $\langle 7, 11, 14, 19\rangle$ &
$\langle 5, 16, 10, 22\rangle$ & $\langle 1
, 17, 12, 0\rangle$\\
$\langle 7, 13, 0, 3\rangle$ & $\langle 13, 16, 21, 6\rangle$ &
$\langle 0, 18, 19, 5\rangle$ & $\langle 3, 17, 6, 16\rangle$ &
$\langle 13, 14, 17, 9\rangle$ &
 $\langle 2, 15, 20, 0\rangle$ & $\langle 9, 21, 12, 16\rangle$ & $\langle 6, 18
, 1, 4\rangle$\\
\hline
\end{tabular}
\caption{Codewords of an Optimal $(23,6,[2,1,1])_4$-code with Size
$80$.} \label{23,6,[2,1,1]}
\end{table*}

\begin{table*}
\centering
\begin{tabular}{cccccccc}
\hline
 $\langle 1, 3, 20, 30\rangle$ & $\langle 16, 34, 33,
17\rangle$ & $\langle 4, 31 , 16, 6\rangle$ & $\langle 13, 19, 4,
20\rangle$ & $\langle 10, 21, 2, 25\rangle $ & $\langle 19, 25, 16,
14\rangle$ & $\langle 6, 18, 0, 32\rangle$ & $\langle 22, 24, 28, 34\rangle$\\
$\langle 24, 33, 8, 3\rangle$ & $\langle 9, 10, 6, 30\rangle$ &
$\langle 8, 10, 15, 24\rangle$ & $\langle 4, 23, 17, 33\rangle$ &
$\langle 8, 31, 20, 25\rangle$
 & $\langle 25, 27, 30, 34\rangle$ & $\langle 12, 17, 24, 13\rangle$ & $\langle
13, 29, 25, 28\rangle$\\
$\langle 4, 10, 7, 14\rangle$ & $\langle 32, 33, 18, 4\rangle$ &
$\langle 14, 27 , 4, 23\rangle$ & $\langle 20, 33, 30, 16\rangle$ &
$\langle 23, 29, 16, 7\rangle$ & $\langle 7, 16, 27, 25\rangle$ &
$\langle 1, 31, 11, 24\rangle$ & $\langle
2, 22, 7, 0\rangle$\\
$\langle 13, 23, 30, 1\rangle$ & $\langle 4, 12, 15, 32\rangle$ &
$\langle 20, 27, 24, 26\rangle$ & $\langle 3, 9, 12, 31\rangle$ &
$\langle 2, 24, 17, 27\rangle$ & $\langle 17, 32, 6, 3\rangle$ &
$\langle 7, 30, 3, 13\rangle$ & $\langle 9,
 29, 21, 34\rangle$\\
$\langle 10, 20, 18, 19\rangle$ & $\langle 1, 22, 19, 27\rangle$ &
$\langle 27, 29, 0, 2\rangle$ & $\langle 14, 22, 13, 17\rangle$ &
$\langle 17, 21, 30, 19\rangle$ & $\langle 2, 9, 16, 5\rangle$ &
$\langle 17, 29, 4, 15\rangle$ & $\langle
13, 18, 10, 22\rangle$\\
$\langle 0, 31, 26, 30\rangle$ & $\langle 3, 25, 7, 32\rangle$ &
$\langle 2, 20,
 14, 1\rangle$ & $\langle 5, 8, 26, 29\rangle$ & $\langle 25, 33, 10, 0\rangle$
& $\langle 3, 8, 23, 16\rangle$ & $\langle 11, 16, 29, 20\rangle$ &
$\langle 16,
 32, 12, 10\rangle$\\
$\langle 17, 33, 27, 22\rangle$ & $\langle 20, 25, 13, 6\rangle$ &
$\langle 4, 11, 2, 9\rangle$ & $\langle 5, 23, 32, 9\rangle$ &
$\langle 16, 18, 19, 24\rangle $ & $\langle 15, 24, 31, 29\rangle$ &
$\langle 13, 31, 17, 21\rangle$ & $\langle
 5, 30, 4, 11\rangle$\\
$\langle 21, 22, 26, 15\rangle$ & $\langle 19, 28, 17, 32\rangle$ &
$\langle 31,
 33, 19, 23\rangle$ & $\langle 0, 32, 19, 21\rangle$ & $\langle 20, 21, 7, 29\rangle$ & $\langle 11, 17, 31, 0\rangle$ & $\langle 9, 25, 17,
8\rangle$ & $\langle 27, 32, 22, 5\rangle$\\
$\langle 9, 24, 1, 20\rangle$ & $\langle 6, 11, 22, 10\rangle$ &
$\langle 18, 30 , 15, 21\rangle$ & $\langle 3, 6, 24, 5\rangle$ &
$\langle 8, 22, 33, 30\rangle$
 & $\langle 11, 19, 21, 5\rangle$ & $\langle 4, 21, 1, 0\rangle$ & $\langle 16,
22, 9, 4\rangle$\\
$\langle 1, 8, 6, 14\rangle$ & $\langle 6, 28, 14, 9\rangle$ &
$\langle 21, 34, 18, 6\rangle$ & $\langle 26, 30, 9, 17\rangle$ &
$\langle 6, 29, 19, 26\rangle$ & $\langle 20, 23, 11, 34\rangle$ &
$\langle 2, 25, 29, 18\rangle$ & $\langle 4,
 24, 30, 18\rangle$\\
$\langle 7, 32, 2, 20\rangle$ & $\langle 13, 27, 9, 12\rangle$ &
$\langle 5, 7, 6, 28\rangle$ & $\langle 12, 22, 5, 18\rangle$ &
$\langle 4, 20, 5, 8\rangle$ & $\langle 21, 27, 28, 16\rangle$ &
$\langle 0, 5, 33, 20\rangle$ & $\langle 14, 29, 20, 10\rangle$\\
$\langle 15, 20, 9, 28\rangle$ & $\langle 19, 30, 1, 2\rangle$ &
$\langle 2, 3, 21, 33\rangle$ & $\langle 2, 31, 28, 15\rangle$ &
$\langle 5, 31, 27, 14\rangle$
 & $\langle 3, 27, 10, 15\rangle$ & $\langle 14, 16, 5, 30\rangle$ & $\langle 7,
 19, 22, 33\rangle$\\
$\langle 28, 33, 12, 21\rangle$ & $\langle 32, 34, 23, 15\rangle$ &
$\langle 8, 21, 9, 32\rangle$ & $\langle 2, 12, 26, 19\rangle$ &
$\langle 4, 34, 22, 25\rangle$ & $\langle 7, 14, 1, 26\rangle$ &
$\langle 29, 31, 22, 12\rangle$ & $\langle
 26, 34, 7, 27\rangle$\\
$\langle 5, 24, 10, 12\rangle$ & $\langle 1, 26, 21, 18\rangle$ &
$\langle 29, 32, 1, 8\rangle$ & $\langle 15, 23, 3, 6\rangle$ &
$\langle 11, 34, 24, 8\rangle$
 & $\langle 19, 34, 9, 3\rangle$ & $\langle 12, 30, 16, 8\rangle$ & $\langle 16,
 28, 13, 2\rangle$\\
$\langle 10, 31, 32, 34\rangle$ & $\langle 5, 15, 1, 22\rangle$ &
$\langle 5, 18 , 34, 2\rangle$ & $\langle 9, 14, 15, 19\rangle$ &
$\langle 8, 27, 11, 19\rangle $ & $\langle 8, 17, 2, 28\rangle$ &
$\langle 6, 12, 31, 7\rangle$ & $\langle 6,
30, 25, 23\rangle$\\
$\langle 6, 13, 33, 15\rangle$ & $\langle 28, 30, 0, 22\rangle$ &
$\langle 7, 15 , 17, 10\rangle$ & $\langle 5, 21, 13, 3\rangle$ &
$\langle 19, 23, 27, 18\rangle$ & $\langle 29, 30, 24, 33\rangle$ &
$\langle 17, 18, 20, 7\rangle$ & $\langle
 14, 18, 3, 8\rangle$\\
$\langle 18, 28, 23, 29\rangle$ & $\langle 20, 22, 31, 32\rangle$ &
$\langle 0, 7, 29, 24\rangle$ & $\langle 0, 10, 11, 28\rangle$ &
$\langle 30, 34, 10, 20\rangle$ & $\langle 4, 26, 29, 13\rangle$ &
$\langle 26, 28, 8, 20\rangle$ & $\langle 12, 34, 29, 1\rangle$\\
$\langle 15, 25, 21, 4\rangle$ & $\langle 1, 10, 5, 13\rangle$ &
$\langle 9, 18,
 4, 26\rangle$ & $\langle 18, 27, 31, 33\rangle$ & $\langle 13, 24, 7, 16\rangle
$ & $\langle 15, 33, 34, 14\rangle$ & $\langle 3, 29, 18, 11\rangle$
& $\langle
1, 17, 34, 9\rangle$\\
$\langle 7, 11, 12, 23\rangle$ & $\langle 23, 26, 14, 22\rangle$ &
$\langle 11, 28, 3, 27\rangle$ & $\langle 12, 21, 14, 11\rangle$ &
$\langle 26, 33, 5, 6\rangle$ & $\langle 2, 23, 10, 8\rangle$ &
$\langle 1, 33, 2, 29\rangle$ & $\langle 0, 14, 2, 6\rangle$\\
$\langle 21, 23, 24, 31\rangle$ & $\langle 11, 18, 1, 25\rangle$ &
$\langle 10, 17, 23, 26\rangle$ & $\langle 0, 1, 25, 12\rangle$ &
$\langle 14, 34, 31, 28\rangle$ & $\langle 0, 8, 18, 13\rangle$ &
$\langle 3, 13, 26, 14\rangle$ & $\langle
 25, 26, 12, 24\rangle$\\
$\langle 10, 12, 33, 27\rangle$ & $\langle 19, 26, 10, 31\rangle$ &
$\langle 25,
 28, 5, 31\rangle$ & $\langle 2, 6, 4, 34\rangle$ & $\langle 6, 27, 1, 17\rangle
$ & $\langle 30, 32, 14, 31\rangle$ & $\langle 0, 3, 17, 34\rangle$
& $\langle 7
, 31, 9, 18\rangle$\\
$\langle 5, 17, 25, 16\rangle$ & $\langle 10, 22, 3, 29\rangle$ &
$\langle 0, 15 , 16, 27\rangle$ & $\langle 1, 28, 4, 7\rangle$ &
$\langle 1, 16, 15, 23\rangle$
 & $\langle 12, 20, 0, 3\rangle$ & $\langle 9, 33, 7, 11\rangle$ & $\langle 15,
19, 8, 12\rangle$\\
$\langle 3, 4, 28, 19\rangle$ & $\langle 24, 32, 26, 11\rangle$ &
$\langle 19, 24, 6, 0\rangle$ & $\langle 7, 8, 34, 4\rangle$ &
$\langle 9, 32, 28, 13\rangle$ & $\langle 15, 26, 32, 2\rangle$ &
$\langle 22, 25, 23, 11\rangle$ & $\langle 14, 24, 25, 21\rangle$\\
$\langle 2, 13, 11, 32\rangle$ & $\langle 0, 9, 22, 23\rangle$ &
$\langle 13, 34 , 0, 5\rangle$ & $\langle 6, 16, 8, 21\rangle$ &
$\langle 11, 15, 13, 30\rangle$
 & $\langle 11, 14, 32, 33\rangle$ & $\langle 12, 23, 28, 25\rangle$ & $\langle
16, 26, 3, 0\rangle$\\\hline
\end{tabular}
\caption{Codewords of an Optimal $(35,6,[2,1,1])_4$-code with Size
$192$.} \label{35,6,[2,1,1]}
\end{table*}

\begin{proposition}
There exists a $[2,1,1]$-GDC$(6)$ of type $2^{3t+1}$
with size $6t^2+2t$ for each $t \in \{3,4,5,6,7,8,9,11\}$.
\end{proposition}
\begin{proof}
For each $t \in \{3,4,5,6,7,8,9,11\}$, let $X_t=\bbZ_{2(3t+1)}$,
$\G_t=\{\{i,3t+1+i\}:i \in \bbZ_{3t+1}\}$ and $\C_t$ be the set of cyclic shifts of the vectors generated by the following vectors respectively. Then $(X_t,\G_t,\C_t)$ is a
$[2,1,1]$-GDC$(6)$ of type $2^{3t+1}$ with size $6t^2+2t$, where
\begin{enumerate}
\item $t=3$, $n=20$, $m=1$, $s=1$, $M=1$
\begin{equation*}
\begin{array}{ccc}
P: \langle0, 1, 4, 7\rangle & \langle0, 2, 14, 15\rangle & \langle0, 9, 5, 17\rangle\\
\end{array}
\end{equation*}

\item $t=4$, $n=26$, $m=3$, $s=3$, $M=1$
\begin{equation*}
\begin{array}{ccc}
P: \langle 0, 1, 5, 22\rangle\\
R: \langle 0, 2, 8, 20\rangle\\
\end{array}
\end{equation*}

\item  $t=5$, $n=32$, $m=7$, $s=2$, $M=1$
\begin{equation*}
\begin{array}{ccc}
P: \langle 0, 1, 3, 29\rangle & \langle 0, 6, 18, 19\rangle\\
R: \langle 0, 8, 17, 23\rangle\\
\end{array}
\end{equation*}

\item $t=6$, $n=38$, $m=7$, $s=3$, $M=1$
\begin{equation*}
\begin{array}{ccc}
P: \langle 0, 1, 3, 30\rangle & \langle 0, 4, 9, 16\rangle\\
\end{array}
\end{equation*}

\item $t=7$, $n=44$, $m=3$, $s=2$, $M=1$
\begin{equation*}
\begin{array}{ccc}
P: \langle 0, 1, 5, 8\rangle & \langle 0, 2, 39, 36\rangle\\
R: \langle 0, 9, 27, 40\rangle & \langle 0, 16, 33, 26\rangle & \langle 0, 19, 32, 30\rangle\\
\end{array}
\end{equation*}

\item $t=8$, $n=50$, $m=3$, $s=5$, $M=1$
\begin{equation*}
\begin{array}{ccc}
P: \langle 0, 2, 19, 33\rangle\\
R: \langle 0, 14, 42, 22\rangle & \langle 0, 20, 5, 15\rangle & \langle 0, 24, 34, 40\rangle\\
\end{array}
\end{equation*}

\item $t=9$, $n=56$, $m=3$, $s=2$, $M=1$
\begin{equation*}
\begin{array}{ccc}
P: \langle 0, 1, 5, 8\rangle & \langle 0, 2, 51, 44\rangle & \langle 0, 25, 36, 48\rangle\\
R: \langle 0, 18, 45, 47\rangle & \langle 0, 22, 9, 39\rangle & \langle 0, 26, 10,16\rangle\\
\end{array}
\end{equation*}

\item $t=11$, $n=68$, $m=11$, $s=3$, $M=1$
\begin{equation*}
\begin{array}{ccc}
P: \langle 0, 1, 43, 63\rangle & \langle 0, 20, 46, 39\rangle\\
R: \langle 0, 17, 10, 9\rangle & \langle 0, 44, 47, 25\rangle & \langle 0, 36, 59, 6\rangle\\
\langle 0, 31, 29, 4\rangle & \langle 0, 12, 45, 8\rangle.\\
\end{array}
\end{equation*}
\end{enumerate}
\end{proof}
\vskip 10pt

\begin{proposition}
There exists a $[2,1,1]$-GDC$(6)$ of type $3^{2t+1}$
with size $6t^2+3t$ for each $t \in \{2,3\}$.
\end{proposition}
\begin{proof}
For each $t \in \{2,3\}$, let $X_t=\bbZ_{3(2t+1)}$, $\G_t=\{\{i,2t+1+i,4t+2+i\}:i \in \bbZ_{2t+1}\}$ and  $\C_t$ be the set of cyclic shifts of the vectors generated by the following vectors respectively. Then $(X_t,\G_t,\C_t)$
is a $[2,1,1]$-GDC$(6)$ of type $3^{2t+1}$ with size $6t^2+3t$, where
\begin{enumerate}
\item $t=2$, $n=15$, $m=2$, $s=2$, $M=1$
\begin{equation*}
\begin{array}{ccc}
P: \langle 0, 1, 4, 12\rangle\\
\end{array}
\end{equation*}

\item $t=3$, $n=21$, $m=2$, $s=3$, $M=1$
\begin{equation*}
\begin{array}{ccc}
P: \langle 0, 1, 9, 13\rangle.\\
\end{array}
\end{equation*}
\end{enumerate}
\end{proof}
\vskip 10pt

\begin{proposition}
There exists a $[2,1,1]$-GDC$(6)$ of type $4^t$ with size $\frac{8t(t-1)}{3}$ for each $t \in \{4,7\}$.
\end{proposition}
\begin{proof}
For each $t \in \{4,7\}$, let $X_t=\bbZ_{4t}$, $\G_t=\{\{i,t+i,2t+i,3t+i\}:i \in \bbZ_t\}$ and  $\C_t$ be the set of cyclic shifts of the vectors generated by the following vectors respectively. Then $(X_t,\G_t,\C_t)$
is a $[2,1,1]$-GDC$(6)$ of type $4^t$ with size $\frac{8t(t-1)}{3}$,
where
\begin{enumerate}
\item $t=4$, $n=16$, $m=5$, $s=2$, $M=1$
\begin{equation*}
\begin{array}{ccc}
P: \langle 0, 1, 7, 10\rangle\\
\end{array}
\end{equation*}

\item $t=7$, $n=28$, $m=3$, $s=2$, $M=1$
\begin{equation*}
\begin{array}{ccc}
P: \langle 0, 1, 5, 17\rangle\\
R: \langle 0, 2, 10, 11\rangle & \langle 0, 6, 19, 24\rangle.\\
\end{array}
\end{equation*}
\end{enumerate}
\end{proof}
\vskip 10pt

\begin{proposition}
There exists a $[2,1,1]$-GDC$(6)$ of type $6^t$ with size $6t(t-1)$ for each $t \in \{4,5,6,7\}$.
\end{proposition}
\begin{proof}
For each $t \in \{4,5,6,7\}$, let $X_t=\bbZ_{6t}$,
$\G_t=\{\{i,t+i,2t+i,3t+i,4t+i,5t+i\}:i \in \bbZ_t\}$ and  $\C_t$ be the set of cyclic shifts of the vectors generated by the following vectors respectively. Then
$(X_t,\G_t,\C_t)$ is a $[2,1,1]$-GDC$(6)$ of type $6^t$ with size
$6t(t-1)$, where
\begin{enumerate}
\item $t=4$, $n=24$, $m=5$, $s=2$, $M=1$
\begin{equation*}
\begin{array}{ccc}
P: \langle 0, 1, 3, 22\rangle\\
R: \langle 0, 7, 13, 18\rangle\\
\end{array}
\end{equation*}

\item $t=5$, $n=30$, $m=17$, $s=2$, $M=1$
\begin{equation*}
\begin{array}{ccc}
P: \langle 0, 3, 1, 7\rangle & \langle 0, 6, 19, 22\rangle\\
\end{array}
\end{equation*}

\item $t=6$, $n=36$, $m=7$, $s=2$, $M=1$
\begin{equation*}
\begin{array}{ccc}
P: \langle 0, 1, 10, 14\rangle\\
R: \langle 0, 25, 4, 21\rangle & \langle 0, 16, 33, 2\rangle & \langle 0, 5, 28, 8\rangle\\
\end{array}
\end{equation*}

\item $t=7$, $n=42$, $m=11$, $s=3$, $M=1$
\begin{equation*}
\begin{array}{ccc}
P: \langle 0, 3, 4, 34\rangle & \langle 0, 12, 29, 32\rangle.\\
\end{array}
\end{equation*}
\end{enumerate}
\end{proof}
\vskip 10pt

\begin{proposition}
There exists a $[2,1,1]$-GDC$(6)$ of type $7^4$ with size $98$.
\end{proposition}
\begin{proof}
Let $X=\bbZ_{28}$, $\G=\{\{i,4+i,8+i,\ldots,24+i\}:i
\in \bbZ_4\}$ and  $\C$ be the set of quasi-cyclic shifts of the vectors generated by the following vectors. Then $(X,\G,\C)$ is a $[2,1,1]$-GDC$(6)$ of type
$7^4$ with size $98$, where $n=28$, $m=5$, $s=2$, $M=2$ and
\begin{equation*}
\begin{array}{ccc}
P: \langle 0, 13, 3, 18\rangle & \langle 1, 12, 3, 14\rangle\\
R: \langle 1, 4, 18, 27\rangle & \langle 0, 5, 26, 27\rangle & \langle 1, 8, 2, 15\rangle.\\
\end{array}
\end{equation*}

\end{proof}
\vskip 10pt

\begin{proposition}
There exists a $[2,1,1]$-GDC$(6)$ of type $10^4$
with size $200$.
\end{proposition}
\begin{proof}
Let $X=\bbZ_{40}$, $\G=\{\{i,4+i,8+i,\ldots,36+i\}:i
\in \bbZ_4\}$ and $\C$ be the set of quasi-cyclic shifts of the vectors generated by the following vectors. Then $(X,\G,\C)$ is a $[2,1,1]$-GDC$(6)$ of type
$10^4$ with size $200$, where $n=40$, $m=3$, $s=3$, $M=2$ and
\begin{equation*}
\begin{array}{ccc}
P: \langle 0, 13, 2, 11\rangle & \langle 0, 3, 21, 22\rangle\\
R: \langle 0, 35, 1, 34\rangle & \langle 0, 25, 10, 15\rangle & \langle 1, 34, 11, 24\rangle\\
\langle 0, 5, 14, 31\rangle.\\
\end{array}
\end{equation*}
\end{proof}
\vskip 10pt

\begin{proposition}
There exists a $[2,1,1]$-GDC$(6)$ of type
$12^t9^1$ with size $12t(2t+1)$ for each $t \in \{4,5,\ldots,15,17,18,$
$19,23\}$.
\end{proposition}
\begin{proof}
For each $t \in \{4,5,\ldots,15,17,18,19,23\}$, let $X_t=\bbZ_{12t}
\cup (\{a,b,c\}\times \bbZ_3)$, $\G_t=\{\{i,t+i,2t+i,\ldots,11t+i\}:i \in \bbZ_t\} \cup
\{\{a,b,c\}\times \bbZ_3\}$ and $\C_t$ be the set of cyclic shifts of the vectors generated by the following vectors respectively. Then $(X_t,\G_t,\C_t)$ is a
$[2,1,1]$-GDC$(6)$ of type $12^t9^1$ with size $12t(2t+1)$.
\begin{enumerate}

\item $t=4$, $n=48$, $m=5$, $s=2$, $M=1$
\begin{equation*}
\begin{array}{ccc}
P: \langle 0, 3, 21, 30\rangle\\
R: \langle a_0, 0, 13, 35\rangle & \langle 5, 43, a_0, 36\rangle & \langle b_0, 0, 23, 25\rangle\\
\langle 23, 37, b_0, 18\rangle & \langle c_0, 0, 7, 5\rangle & \langle 25, 47, c_0, 36\rangle\\
\langle 0, 2, 1, 19\rangle\\
\end{array}
\end{equation*}

\item $t=5$, $n=60$, $m=17$, $s=3$, $M=1$
\begin{equation*}
\begin{array}{ccc}
P: \langle 0, 2, 3, 6\rangle\\
R: \langle a_0, 0, 52, 53\rangle & \langle 17, 46, a_0, 0\rangle & \langle b_0, 0, 46, 59\rangle\\
\langle 4, 41, b_0, 0\rangle & \langle c_0, 0, 44, 7\rangle & \langle 5, 52, c_0, 33\rangle\\
\langle 0, 24, 12, 33\rangle & \langle 0, 21, 32, 18\rangle\\
\end{array}
\end{equation*}

\item $t=6$, $n=72$, $m=19$, $s=2$, $M=1$
\begin{equation*}
\begin{array}{ccc}
P: \langle 0, 9, 50, 31\rangle & \langle 0, 7, 33, 10\rangle & \langle 0, 67, 34, 69\rangle\\
R: \langle a_0, 0, 20, 28\rangle & \langle 8, 61, a_0, 45\rangle & \langle b_0, 0, 32, 52\rangle\\
\langle 55, 59, b_0, 12\rangle & \langle c_0, 0, 35, 40\rangle & \langle 29, 37, c_0, 12\rangle\\
\langle 0, 1, 17, 44\rangle\\
\end{array}
\end{equation*}

\item $t=7$, $n=84$, $m=5$, $s=3$, $M=1$
\begin{equation*}
\begin{array}{ccc}
P: \langle 0, 2, 20, 71\rangle & \langle 0, 8, 83, 66\rangle\\
R: \langle a_0, 0, 65, 73\rangle & \langle 16, 77, a_0, 78\rangle & \langle b_0, 0, 68, 25\rangle\\
\langle 7, 74, b_0, 54\rangle & \langle c_0, 0, 13, 53\rangle & \langle 25, 68, c_0, 30\rangle\\
\langle 0, 55, 31, 81\rangle & \langle 0, 12, 15, 48\rangle & \langle 0, 33, 37, 57\rangle\\
\end{array}
\end{equation*}

\item $t=8$, $n=96$, $m=77$, $s=4$, $M=1$
\begin{equation*}
\begin{array}{ccc}
P: \langle 0, 13, 87, 46\rangle & \langle 0, 23, 62, 81\rangle\\
R: \langle a_0, 0, 1, 68\rangle & \langle 26, 31, a_0, 51\rangle & \langle b_0, 0, 95, 28\rangle\\
\langle 34, 41, b_0, 78\rangle & \langle c_0, 0, 59, 76\rangle & \langle 28, 47, c_0, 3\rangle\\
\langle 0, 66, 78, 84\rangle & \langle 0, 54, 60, 90\rangle & \langle 0, 35, 4, 31\rangle\\
\end{array}
\end{equation*}

\item $t=9$, $n=108$, $m=7$, $s=3$, $M=1$
\begin{equation*}
\begin{array}{ccc}
P: \langle 0, 25, 39, 76\rangle & \langle 0, 24, 55, 77\rangle & \langle 0, 95, 105 , 80\rangle\\
\langle 0, 30, 4, 8\rangle\\
R: \langle a_0, 0, 106, 59\rangle & \langle 82, 98, a_0, 9\rangle & \langle b_0, 0,  79, 74\rangle\\
\langle 23, 46, b_0, 72\rangle & \langle c_0, 0, 50, 73\rangle & \langle 22, 86, c_0, 18\rangle\\
\langle 0, 43, 89, 29\rangle\\
\end{array}
\end{equation*}

\item $t=10$, $n=120$, $m=113$, $s=3$, $M=1$
\begin{equation*}
\begin{array}{ccc}
P: \langle 0, 82, 33, 86\rangle & \langle 0, 29, 117, 78\rangle & \langle 0, 77, 79 , 73\rangle\\
R: \langle a_0, 0, 46, 32\rangle & \langle 37, 44, a_0, 78\rangle & \langle b_0, 0,  109, 74\rangle\\
\langle 34, 86, b_0, 99\rangle & \langle c_0, 0, 8, 55\rangle & \langle 38, 43, c_0, 54\rangle\\
\langle 0, 15, 99, 51\rangle & \langle 0, 114, 63, 12\rangle & \langle 0, 96, 42, 87\rangle\\
\langle 0, 81, 25, 56\rangle & \langle 0, 75, 27, 3\rangle & \langle 0, 35, 23, 22\rangle\\
\end{array}
\end{equation*}

\item $t=11$, $n=132$, $m=5$, $s=5$, $M=1$
\begin{equation*}
\begin{array}{ccc}
P: \langle 0, 62, 90, 104\rangle & \langle 0, 16, 129, 81\rangle & \langle 0, 120, 19, 27\rangle\\
R: \langle a_0, 0, 125, 67\rangle & \langle 62, 112, a_0, 87\rangle & \langle b_0, 0, 7, 35\rangle\\
\langle 61, 71, b_0, 120\rangle & \langle c_0, 0, 97, 17\rangle & \langle 0, 14, 85, 43\rangle\\
\langle 0, 106, 83, 119\rangle & \langle 56, 58, c_0, 15\rangle\\
\end{array}
\end{equation*}

\item $t=12$, $n=144$, $m=11$, $s=3$, $M=1$
\begin{equation*}
\begin{array}{ccc}
P: \langle 0, 62, 77, 49\rangle & \langle 0, 44, 118, 5\rangle & \langle 0, 27, 37, 136\rangle\\
\langle 0, 16, 67, 35\rangle\\
R: \langle a_0, 0, 23, 142\rangle & \langle 10, 86, a_0, 132\rangle & \langle b_0, 0, 88, 50\rangle\\
\langle 5, 25, b_0, 108\rangle & \langle c_0, 0, 53, 91\rangle & \langle 16, 47, c_0, 24\rangle\\
\langle 0, 114, 95, 104\rangle & \langle 0, 90, 79, 61\rangle & \langle 0, 63, 141, 138\rangle\\
\langle 0, 137, 66, 86\rangle & \langle 0, 102, 69, 143\rangle & \langle 0, 28, 34, 98\rangle\\
\langle 0, 18, 57, 89\rangle\\
\end{array}
\end{equation*}

\item $t=13$, $n=156$, $m=37$, $s=5$, $M=1$
\begin{equation*}
\begin{array}{ccc}
P: \langle 0, 50, 29, 125\rangle & \langle 0, 154, 100, 54\rangle & \langle 0, 148, 49, 155\rangle\\
R: \langle 0, 48, 69, 25\rangle & \langle 0, 90, 145, 17\rangle & \langle 0, 14, 73, 45\rangle \\
\langle 55, 116, b_0, 135\rangle & \langle c_0, 0, 115, 107\rangle & \langle 0, 12, 99, 9\rangle\\
\langle 68, 139, c_0, 144\rangle & \langle 0, 24, 96, 84\rangle & \langle 0, 120, 105, 43\rangle\\
\langle 136, 140, a_0, 111\rangle & \langle a_0, 0, 118, 11\rangle & \langle b_0, 0, 38, 109\rangle\\
\end{array}
\end{equation*}

\item $t=14$, $n=168$, $m=103$, $s=5$, $M=1$
\begin{equation*}
\begin{array}{ccc}
P: \langle 0, 25, 64, 19\rangle & \langle 0, 46, 29, 90\rangle & \langle 0, 38, 61, 125\rangle\\
R: \langle b_0, 0, 76, 2\rangle & \langle 0, 60, 48, 72\rangle & \langle 52, 59, a_0, 111\rangle\\
\langle 22, 113, b_0, 30\rangle & \langle c_0, 0, 41, 4\rangle & \langle 35, 154, c_0, 138\rangle\\
\langle 0, 93, 136, 102\rangle & \langle 0, 45, 150, 65\rangle & \langle 0, 24, 124, 104\rangle\\
\langle 0, 51, 133, 11\rangle & \langle 0, 69, 96, 132\rangle & \langle 0, 147, 120, 145\rangle\\
\langle a_0, 0, 148, 86\rangle & \langle 0, 3, 36, 35\rangle\\
\end{array}
\end{equation*}

\item $t=15$, $n=180$, $m=7$, $s=4$, $M=1$
\begin{equation*}
\begin{array}{ccc}
P: \langle 0, 153, 96, 12\rangle & \langle 0, 161, 3, 177\rangle & \langle 0, 109, 116, 14\rangle\\
\langle 0, 6, 82, 38\rangle & \langle 0, 41, 142, 5\rangle\\
R: \langle a_0, 0, 68, 13\rangle & \langle 4, 71, a_0, 162\rangle & \langle b_0, 0, 4, 170\rangle\\
\langle 5, 145, b_0, 165\rangle & \langle c_0, 0, 79, 119\rangle & \langle 4, 134, c_0, 135\rangle\\
\langle 0, 80, 136, 10\rangle & \langle 0, 54, 179, 26\rangle & \langle 0, 169, 17, 134\rangle\\
\langle 0, 18, 2, 115\rangle & \langle 0, 77, 52, 70\rangle\\
\end{array}
\end{equation*}

\item $t=17$, $n=204$, $m=41$, $s=8$, $M=1$
\begin{equation*}
\begin{array}{ccc}
P: \langle 0, 3, 50, 125\rangle & \langle 0, 126, 37, 55\rangle & \langle 0, 105, 14, 191\rangle\\
R: \langle a_0, 0, 20, 52\rangle & \langle 65, 73, a_0, 9\rangle & \langle b_0, 0, 164, 4\rangle\\
\langle 41, 169, b_0, 141\rangle & \langle c_0, 0, 116, 64\rangle & \langle 38, 82, c_0, 66\rangle\\
\langle 0, 108, 56, 104\rangle & \langle 0, 80, 112, 172\rangle & \langle 0, 168, 184, 192\rangle\\
\langle 0, 60, 72, 180\rangle & \langle 0, 48, 132, 88\rangle\\
\end{array}
\end{equation*}

\item $t=18$, $n=216$, $m=23$, $s=4$, $M=1$
\begin{equation*}
\begin{array}{ccc}
P: \langle 0, 89, 25, 38\rangle & \langle 0, 118, 61, 78\rangle & \langle 0, 172, 185, 84\rangle\\
\langle 0, 163, 88, 140\rangle & \langle 0, 29, 105, 203\rangle & \langle 0, 86, 48, 109\rangle\\
R: \langle a_0, 0, 65, 142\rangle & \langle 5, 151, a_0, 129\rangle & \langle 0, 1, 27, 190\rangle\\
\langle 2, 151, b_0, 201\rangle & \langle c_0, 0, 22, 41\rangle & \langle 7, 86, c_0, 39\rangle\\
\langle b_0, 0, 167, 37\rangle & \langle 0, 153, 147, 92\rangle & \langle 0, 45, 60, 195\rangle\\
\langle 0, 154, 12, 129\rangle & \langle 0, 99, 192, 96\rangle & \langle 0, 207, 177, 64\rangle\\
\langle 0, 81, 31, 183\rangle\\
\end{array}
\end{equation*}

\item $t=19$, $n=228$, $m=43$, $s=7$, $M=1$
\begin{equation*}
\begin{array}{ccc}
P: \langle 0, 74, 43, 139\rangle & \langle 0, 88, 226, 186\rangle & \langle 0, 21, 156, 93\rangle\\
\langle 0, 23, 196, 31\rangle\\
R: \langle 0, 40, 82, 126\rangle & \langle 14, 142, a_0, 93\rangle & \langle b_0, 0, 67, 149\rangle\\
\langle 29, 112, b_0, 201\rangle & \langle c_0, 0, 97, 32\rangle & \langle 28, 74, c_0, 189\rangle\\
\langle 0, 104, 210, 50\rangle & \langle 0, 63, 51, 12\rangle & \langle 0, 27, 168, 87\rangle\\
\langle a_0, 0, 203, 157\rangle & \langle 0, 70, 131, 71\rangle\\
\end{array}
\end{equation*}

\item $t=23$, $n=276$, $m=7$, $s=7$, $M=1$
\begin{equation*}
\begin{array}{ccc}
P: \langle 0, 8, 86, 114\rangle & \langle 0, 200, 257, 43\rangle & \langle 0, 58, 261, 125\rangle\\
\langle 0, 228, 211, 105\rangle\\
R: \langle 0, 51, 88, 15\rangle & \langle 0, 2, 170, 165\rangle & \langle b_0, 0, 109, 83\rangle\\
\langle 13, 161, b_0, 240\rangle & \langle c_0, 0, 251, 181\rangle & \langle 11, 187, c_0, 216\rangle\\
\langle 16, 134, a_0, 204\rangle & \langle a_0, 0, 89, 13\rangle & \langle 0, 80, 87, 145\rangle\\
\langle 0, 55, 126, 210\rangle & \langle 0, 98, 179, 30\rangle & \langle 0, 249, 75, 162\rangle\\
\langle 0, 224, 49, 167\rangle & \langle 0, 72, 108, 90\rangle & \langle 0, 214, 147, 1\rangle\\
\langle 0, 252, 110, 142\rangle & \langle 0, 14, 201, 68\rangle & \langle 0, 64, 233, 149\rangle\\
\langle 0, 104, 195, 215\rangle.\\
\end{array}
\end{equation*}

\end{enumerate}
\end{proof}
\vskip 10pt

\begin{proposition}
\label{GDC(6)12t151}There exists a $[2,1,1]$-GDC$(6)$ of type
$12^t15^1$ with size $12t(2t+3)$ for each $t \in \{7,8,\ldots,15\}$.
\end{proposition}
\begin{proof}
For each $t \in \{7,8,\ldots,15\}$, let $X_t=\bbZ_{12t} \cup
(\{a,b,c,d,e\}\times \bbZ_3)$,
$\G_t=\{\{i,t+i,2t+i,\ldots,11t+i\}:i \in \bbZ_t\} \cup
\{\{a,b,c,d,e\}\times \bbZ_3\}$  and  $\C_t$ be the set of cyclic shifts of the vectors generated by the following vectors respectively. Then $(X_t,\G_t,\C_t)$ is a
$[2,1,1]$-GDC$(6)$ of type $12^t15^1$ with size $12t(2t+3)$.
\begin{enumerate}

\item $t=7$, $n=84$, $m=11$, $s=3$, $M=1$
\begin{equation*}
\begin{array}{ccc}
P: \langle 0, 39, 36, 3\rangle & \langle 0, 18, 52, 83\rangle\\
R: \langle b_0, 0, 25, 50\rangle & \langle 17, 22, a_0, 0\rangle & \langle a_0, 0, 16, 8\rangle\\
\langle e_0, 0, 55, 23\rangle & \langle c_0, 0, 37, 32\rangle & \langle 49, 59, c_0, 78\rangle\\
\langle 79, 83, d_0, 12\rangle & \langle 4, 68, b_0, 45\rangle & \langle 2, 28, e_0, 3\rangle\\
\langle d_0, 0, 11, 22\rangle & \langle 0, 40, 71, 2\rangle\\
\end{array}
\end{equation*}

\item $t=8$, $n=96$, $m=5$, $s=3$, $M=1$
\begin{equation*}
\begin{array}{ccc}
P: \langle 0, 6, 59, 93\rangle & \langle 0, 1, 11, 38\rangle\\
R: \langle a_0, 0, 20, 19\rangle & \langle 20, 88, a_0, 9\rangle & \langle b_0, 0, 34, 41\rangle\\
\langle 13, 95, b_0, 78\rangle & \langle c_0, 0, 67, 62\rangle & \langle 25, 47, c_0, 51\rangle\\
\langle e_0, 0, 70, 29\rangle & \langle 0, 60, 3, 69\rangle & \langle 17, 64, e_0, 60\rangle\\
\langle 19, 71, d_0, 21\rangle & \langle d_0, 0, 23, 76\rangle & \langle 0, 84, 15, 33\rangle\\
\langle 0, 18, 31, 75\rangle\\
\end{array}
\end{equation*}

\item $t=9$, $n=108$, $m=5$, $s=3$, $M=1$
\begin{equation*}
\begin{array}{ccc}
P: \langle 0, 58, 39, 42\rangle & \langle 0, 83, 41, 105\rangle & \langle 0, 79, 94, 57\rangle\\
R: \langle a_0, 0, 26, 43\rangle & \langle 58, 62, a_0, 78\rangle & \langle b_0, 0, 88, 14\rangle\\
\langle 56, 61, b_0, 96\rangle & \langle c_0, 0, 59, 7\rangle & \langle 34, 35, c_0, 102\rangle\\
\langle 31, 92, d_0, 3\rangle & \langle e_0, 0, 101, 76\rangle & \langle 22, 86, e_0, 33\rangle\\
\langle 0, 100, 48, 52\rangle & \langle d_0, 0, 70, 95\rangle & \langle 0, 24, 12, 73\rangle\\
\end{array}
\end{equation*}

\item $t=10$, $n=120$, $m=43$, $s=3$, $M=1$
\begin{equation*}
\begin{array}{ccc}
P: \langle 0, 68, 61, 39\rangle & \langle 0, 104, 58, 102\rangle & \langle 0, 107, 14, 23\rangle\\
R: \langle a_0, 0, 55, 11\rangle & \langle 67, 110, a_0, 102\rangle & \langle b_0, 0, 8, 1\rangle\\
\langle 55, 80, b_0, 6\rangle & \langle c_0, 0, 98, 85\rangle & \langle 0, 63, 72, 105\rangle\\
\langle 64, 95, d_0, 102\rangle & \langle 0, 33, 45, 51\rangle & \langle 50, 55, e_0, 36\rangle\\
\langle 58, 62, c_0, 111\rangle & \langle d_0, 0, 65, 97\rangle & \langle e_0, 0, 67, 119\rangle\\
\langle 0, 114, 48, 15\rangle & \langle 0, 24, 99, 93\rangle\\
\end{array}
\end{equation*}

\item $t=11$, $n=132$, $m=5$, $s=4$, $M=1$
\begin{equation*}
\begin{array}{ccc}
P: \langle 0, 102, 82, 61\rangle & \langle 0, 7, 16, 124\rangle & \langle 0, 12, 3, 50\rangle\\
R: \langle d_0, 0, 74, 109\rangle & \langle a_0, 0, 25, 131\rangle & \langle b_0, 0, 1, 128\rangle\\
\langle 34, 44, b_0, 6\rangle & \langle c_0, 0, 106, 2\rangle & \langle 14, 79, c_0, 54\rangle\\
\langle 19, 98, d_0, 3\rangle & \langle 25, 128, a_0, 96\rangle & \langle 17, 34, e_0, 102\rangle\\
\langle e_0, 0, 127, 5\rangle & \langle 0, 108, 27, 52\rangle & \langle 0, 34, 81, 129\rangle\\
\langle 0, 6, 119, 19\rangle\\
\end{array}
\end{equation*}

\item $t=12$, $n=144$, $m=101$, $s=3$, $M=1$
\begin{equation*}
\begin{array}{ccc}
P: \langle 0, 8, 17, 3\rangle & \langle 0, 14, 57, 100\rangle & \langle 0, 16, 113, 127\rangle\\
R: \langle a_0, 0, 79, 77\rangle & \langle 14, 103, a_0, 84\rangle & \langle e_0, 0, 47, 52\rangle\\
\langle 20, 73, b_0, 21\rangle & \langle c_0, 0, 98, 10\rangle & \langle 11, 61, c_0, 60\rangle\\
\langle 11, 70, d_0, 132\rangle & \langle 0, 42, 124, 63\rangle & \langle 11, 13, e_0, 120\rangle\\
\langle 0, 93, 66, 27\rangle & \langle d_0, 0, 35, 22\rangle & \langle 0, 76, 105, 101\rangle\\
\langle 0, 30, 58, 99\rangle & \langle 0, 13, 7, 44\rangle & \langle 0, 126, 122, 116\rangle\\
\langle 0, 54, 87, 119\rangle & \langle b_0, 0, 106, 95\rangle & \langle 0, 61, 135, 67\rangle\\
\end{array}
\end{equation*}

\item $t=13$, $n=156$, $m=7$, $s=6$, $M=1$
\begin{equation*}
\begin{array}{ccc}
P: \langle 0, 56, 21, 135\rangle & \langle 0, 12, 11, 131\rangle & \langle 0, 114, 61, 57\rangle\\
R: \langle a_0, 0, 94, 122\rangle & \langle 47, 55, a_0, 69\rangle & \langle b_0, 0, 154, 62\rangle\\
\langle 59, 91, b_0, 69\rangle & \langle d_0, 0, 110, 34\rangle & \langle 49, 53, c_0, 99\rangle\\
\langle 64, 104, d_0, 6\rangle & \langle e_0, 0, 106, 86\rangle & \langle 23, 151, e_0, 141\rangle\\
\langle c_0, 0, 38, 82\rangle & \langle 0, 68, 142, 70\rangle\\
\end{array}
\end{equation*}

\item $t=14$, $n=168$, $m=101$, $s=5$, $M=1$
\begin{equation*}
\begin{array}{ccc}
P: \langle 0, 39, 37, 34\rangle & \langle 0, 19, 120, 123\rangle & \langle 0, 54, 13, 58\rangle\\
R: \langle a_0, 0, 77, 22\rangle & \langle 20, 28, a_0, 60\rangle & \langle b_0, 0, 62, 31\rangle\\
\langle 29, 94, b_0, 45\rangle & \langle d_0, 0, 7, 128\rangle & \langle c_0, 0, 44, 49\rangle\\
\langle 22, 98, d_0, 153\rangle & \langle 20, 55, c_0, 162\rangle & \langle 23, 139, e_0, 159\rangle\\
\langle 0, 147, 81, 111\rangle & \langle 0, 60, 69, 2\rangle & \langle 0, 73, 1, 66\rangle\\
\langle 0, 27, 63, 17\rangle & \langle 0, 122, 45, 11\rangle & \langle e_0, 0, 38, 124\rangle\\
\langle 0, 12, 117, 155\rangle\\
\end{array}
\end{equation*}

\item $t=15$, $n=180$, $m=37$, $s=4$, $M=1$
\begin{equation*}
\begin{array}{ccc}
P: \langle 0, 52, 58, 49\rangle & \langle 0, 61, 37, 54\rangle & \langle 0, 68, 134, 31\rangle\\
\langle 0, 74, 167, 118\rangle & \langle 0, 81, 24, 28\rangle\\
R: \langle a_0, 0, 53, 43\rangle & \langle 41, 61, a_0, 171\rangle & \langle b_0, 0, 17, 85\rangle\\
\langle 59, 64, b_0, 24\rangle & \langle c_0, 0, 113, 79\rangle & \langle 85, 140, c_0, 0\rangle\\
\langle d_0, 0, 151, 122\rangle & \langle 58, 173, d_0, 93\rangle & \langle e_0, 0, 149, 7\rangle\\
\langle 43, 53, e_0, 123\rangle & \langle 0, 36, 77, 50\rangle & \langle 0, 25,3, 111\rangle\\
\langle 0, 108, 89, 147\rangle.\\
\end{array}
\end{equation*}
\end{enumerate}
\end{proof}
\vskip 10pt

\begin{proposition}
There exists a $[2,1,1]$-GDC$(6)$ of type $13^4$
with size $338$.
\end{proposition}
\begin{proof}
Let $X=\bbZ_{52}$, $\G=\{\{i,4+i,8+i,\ldots,48+i\}:i
\in \bbZ_4\}$ and $\C$ be the set of quasi-cyclic shifts of the vectors generated by the following vectors. Then $(X,\G,\C)$ is a $[2,1,1]$-GDC$(6)$ of type
$13^4$ with size $338$, where
$n=52$, $m=5$, $s=3$, $M=2$, and
\begin{equation*}
\begin{array}{cccc}
P: \langle 0, 5, 3, 46\rangle & \langle 1, 20, 22, 47\rangle & \langle 0, 7, 18, 41\rangle\\
R: \langle 0, 1, 11, 26\rangle & \langle 0, 13, 39, 42\rangle & \langle 0, 17, 30, 47\rangle\\
\langle 0, 43, 14, 29\rangle.\\
\end{array}
\end{equation*}
\end{proof}
\vskip 10pt

\begin{proposition}
There exists a $[2,1,1]$-GDC$(6)$ of type $22^4$
with size $968$.
\end{proposition}
\begin{proof}
Let $X=\bbZ_{88}$, $\G=\{\{i,4+i,8+i,\ldots,84+i\}:i
\in \bbZ_4\}$ and $\C$ be the set of quasi-cyclic shifts of the vectors generated by the following vectors. Then $(X,\G,\C)$ is a $[2,1,1]$-GDC$(6)$ of type
$22^4$ with size $968$, where
$n=88$, $m=5$, $s=3$, $M=2$, and
\begin{equation*}
\begin{array}{ccc}
P: \langle 0, 9, 23, 26\rangle & \langle 0, 1, 3, 62\rangle\\
R: \langle 0, 77, 67, 22\rangle & \langle 0, 69, 55, 38\rangle & \langle 0, 51, 82, 21\rangle\\
\langle 0, 7, 66, 13\rangle & \langle 0, 35, 18, 65\rangle.\\
\end{array}
\end{equation*}
\end{proof}
\vskip 10pt

\begin{proposition}
There exists a $[2,1,1]$-GDC$(6)$ of type
$1^{12}2^1$ with size $28$.
\end{proposition}
\begin{proof}Let $X=\bbZ_{14}$, and $\G_t=\{\{0,1\}\} \cup \{\{i\}:i
\in \bbZ_{14} \setminus \{0,1\}\}$. Then $(X,\G,\C)$ is a
$[2,1,1]$-GDC$(6)$ of type $1^{12}2^1$ with size $28$, where $\C$ is
the set of
\begin{equation*}
\begin{array}{cccc}
\langle 4, 6, 7, 11\rangle & \langle 1, 5, 11, 7\rangle & \langle 2, 10, 8, 7\rangle & \langle 5, 6, 0, 8\rangle\\
\langle 8, 11, 5, 10\rangle & \langle 8, 12, 6, 2\rangle & \langle 10, 12, 9, 5\rangle & \langle 0, 9, 11, 12\rangle\\
\langle 3, 7, 1, 10\rangle & \langle 1, 13, 2, 3\rangle & \langle 1, 6, 10, 12\rangle & \langle 2, 5, 12, 13\rangle\\
\langle 3, 6, 13, 9\rangle & \langle 4, 10, 13, 1\rangle & \langle 0, 8, 7, 13\rangle & \langle 11, 12, 3, 1\rangle\\
\langle 2, 9, 1, 6\rangle & \langle 0, 4, 2, 5\rangle & \langle 9, 13, 10, 8\rangle & \langle 1, 8, 4, 9\rangle\\
\langle 7, 13, 5, 6\rangle & \langle 5, 9, 3, 4\rangle & \langle 7, 11, 2, 9\rangle & \langle 0, 10, 3, 6\rangle\\
\langle 3, 4, 8, 12\rangle & \langle 2, 3, 0, 11\rangle & \langle 11, 13, 4, 0\rangle & \langle 7, 12, 0, 4\rangle.\\
\end{array}
\end{equation*}
\end{proof}
\vskip 10pt

\begin{proposition}
There exists a $[2,1,1]$-GDC$(6)$ of type
$1^t11^1$ with size $6u^2+20u$, where $6u=t$ for each $t \in \{30,36,54,$
$66,78\}$.
\end{proposition}
\begin{proof}
For each $t \in \{30,36,54,66,78\}$, let $X_t=\bbZ_{t} \cup
(\{a,b,c,d,e\}\times \bbZ_2) \cup \{f\}$, $\G_t=\{\{i\}:i \in
\bbZ_t\} \cup \{\{a,b,c,d,e\}\times \bbZ_2 \cup \{f\}\}$ and  $\C_t$ be the set of quasi-cyclic shifts of the vectors generated by the following vectors respectively. Then
$(X_t,\G_t,\C_t)$ is a $[2,1,1]$-GDC$(6)$ of type $1^t11^1$ with size
$6u^2+20u$,  where $6u=t$ and
\begin{enumerate}

\item $t=30$, $n=30$, $m=1$, $s=1$, $M=3$
\begin{equation*}
\begin{array}{ccc}
P: \langle a_0, 24, 7, 2 \rangle & \langle 1, 21, a_0, 0 \rangle & \langle a_0, 28, 29, 3 \rangle\\
\langle 2, 11, a_0, 28 \rangle & \langle b_0, 15, 13, 10 \rangle & \langle 1, 20,b_0, 15 \rangle\\
\langle b_0, 0, 23, 2 \rangle & \langle 4, 29, b_0, 12 \rangle & \langle c_0, 6, 5, 27 \rangle\\
\langle 3, 29, c_0, 14 \rangle & \langle c_0, 26, 19, 4 \rangle & \langle 1, 22, c_0, 24 \rangle\\
\langle d_0, 2, 29, 4 \rangle & \langle 1, 11, d_0, 18 \rangle & \langle d_0, 12, 19, 9 \rangle\\
\langle 3, 22, d_0, 8 \rangle & \langle e_0, 14, 17, 28 \rangle & \langle 3, 17, e_0, 18 \rangle\\
\langle e_0, 6, 9, 7 \rangle & \langle 1, 4, e_0, 8 \rangle & \langle f, 11, 0, 10 \rangle\\
\langle 0, 4, f, 17 \rangle & \langle 0, 24, 12, 16 \rangle & \langle 1, 7, 19, 29 \rangle\\
\langle 2, 8, 0, 20 \rangle\\
\end{array}
\end{equation*}

\item $t=36$, $n=36$, $m=1$, $s=1$, $M=3$
\begin{equation*}
\begin{array}{ccc}
P: \langle a_0, 24, 13, 14 \rangle & \langle 1, 3, a_0, 6 \rangle & \langle a_0, 28, 15, 35 \rangle\\
\langle 2, 35, a_0, 4 \rangle & \langle b_0, 12, 19, 16 \rangle & \langle 1, 2, b_0, 27 \rangle\\
\langle b_0, 9, 8, 11 \rangle & \langle 4, 23, b_0, 18 \rangle & \langle c_0, 30, 5, 27 \rangle\\
\langle 1, 4, c_0, 0 \rangle & \langle c_0, 26, 22, 1 \rangle & \langle 3, 35, c_0, 26 \rangle\\
\langle d_0, 0, 15, 10 \rangle & \langle 1, 17, d_0, 30 \rangle & \langle d_0, 32,17, 25 \rangle\\
\langle 3, 22, d_0, 20 \rangle & \langle e_0, 23, 2, 6 \rangle & \langle 1, 14, e_0, 23 \rangle\\
\langle e_0, 10, 21, 1 \rangle & \langle 0, 9, e_0, 22 \rangle & \langle f, 0, 31, 5 \rangle\\
\langle 1, 29, f, 9 \rangle & \langle 0, 24, 8, 6 \rangle & \langle 1, 31, 10,26 \rangle\\
\langle 2, 9, 30, 3 \rangle & \langle 2, 28, 32, 16 \rangle & \langle 0, 16, 28, 1 \rangle\\
\langle 2, 14, 24, 20 \rangle\\
\end{array}
\end{equation*}

\item $t=54$, $n=54$, $m=5$, $s=6$, $M=3$
\begin{equation*}
\begin{array}{ccc}
P: \langle 0, 1, 2, 3\rangle\\
R: \langle a_0, 30, 19, 4\rangle & \langle 1, 23, a_0, 38\rangle & \langle a_0, 8, 39,  41\rangle\\
\langle 3, 16, a_0, 48\rangle & \langle b_0, 24, 33, 47\rangle & \langle 1, 52, b_0, 12\rangle\\
\langle b_0, 38, 40, 19\rangle & \langle 3, 29, b_0, 2\rangle & \langle c_0, 18, 22, 25\rangle\\
\langle 1, 47, c_0, 50\rangle & \langle c_0, 2, 27, 23\rangle & \langle 3, 40, c_0, 30\rangle\\
\langle d_0, 42, 29, 37\rangle & \langle 1, 5, d_0, 21\rangle & \langle d_0, 45, 10, 20\rangle\\
\langle 2, 16, d_0, 0\rangle & \langle e_0, 0, 52, 38\rangle & \langle 1, 46, e_0,27\rangle\\
\langle e_0, 15, 29, 31\rangle & \langle 2, 47, e_0, 6\rangle & \langle f, 31, 24, 5\rangle\\
\langle 0, 32, f, 22\rangle & \langle 0, 42, 36, 6\rangle & \langle 1, 31, 20, 13\rangle\\
\langle 1, 22, 40, 49\rangle & \langle 1, 14, 6, 43\rangle & \langle 2, 26, 8, 45\rangle\\
\langle 1, 17, 11, 16\rangle & \langle 0, 24, 35, 44\rangle & \langle 1, 35, 7, 53\rangle\\
\langle 2, 14, 53, 25\rangle\\
\end{array}
\end{equation*}

\item $t=66$, $n=66$, $m=17$, $s=4$, $M=3$
\begin{equation*}
\begin{array}{ccc}
P: \langle 0, 1, 2, 3\rangle & \langle 0, 5, 6, 10\rangle\\
R: \langle a_0, 42, 22, 13\rangle & \langle 1, 11, a_0, 44\rangle & \langle a_0, 26, 17, 39\rangle\\
\langle 3, 64, a_0, 36\rangle & \langle b_0, 24, 39, 2\rangle & \langle 1, 64, b_0, 47\rangle\\
\langle b_0, 59, 58, 1\rangle & \langle 2, 21, b_0, 0\rangle & \langle c_0, 42, 32, 17\rangle\\
\langle 1, 46, c_0, 12\rangle & \langle c_0, 51, 49, 34\rangle & \langle 2, 5, c_0, 57\rangle\\
\langle d_0, 18, 26, 46\rangle & \langle 1, 16, d_0, 9\rangle & \langle d_0, 63, 29, 13\rangle\\
\langle 2, 53, d_0, 18\rangle & \langle e_0, 14, 23, 18\rangle & \langle 1, 17, e_0, 56\rangle\\
\langle e_0, 63, 1, 4\rangle & \langle 0, 22, e_0, 57\rangle & \langle f, 36, 10, 32\rangle\\
\langle 1, 29, f, 57\rangle & \langle 0, 54, 43, 48\rangle & \langle 1, 19, 31, 59\rangle\\
\langle 1, 10, 24, 37\rangle & \langle 1, 7, 40, 27\rangle & \langle 1, 35, 55, 65\rangle\\
\langle 1, 53, 23, 63\rangle & \langle 2, 20, 42, 64\rangle & \langle 2, 26, 49, 28\rangle\\
\langle 1, 25, 62, 32\rangle & \langle 0, 23, 11, 35\rangle & \langle 2, 62, 37, 23\rangle\\
\langle 0, 39, 63, 53\rangle & \langle 0, 20, 31, 65\rangle\\
\end{array}
\end{equation*}

\item $t=78$, $n=78$, $m=11$, $s=4$, $M=3$
\begin{equation*}
\begin{array}{ccc}
P: \langle 0, 1, 2, 9\rangle & \langle 0, 44, 62, 56\rangle & \langle 0, 61, 66, 19\rangle\\
R: \langle a_0, 18, 41, 7\rangle & \langle 1, 11, a_0, 56\rangle & \langle a_0, 44, 28, 3\rangle\\
\langle 3, 10, a_0, 6\rangle & \langle b_0, 6, 32, 46\rangle & \langle 1, 52, b_0, 42\rangle\\
\langle b_0, 75, 31, 29\rangle & \langle 2, 23, b_0, 21\rangle & \langle c_0, 72, 67, 68\rangle\\
\langle 1, 35, c_0, 21\rangle & \langle c_0, 3, 16, 41\rangle & \langle 2, 16, c_0, 72\rangle\\
\langle d_0, 72, 27, 29\rangle & \langle 1, 20, d_0, 46\rangle & \langle d_0, 28, 13, 20\rangle\\
\langle 3, 53, d_0, 66\rangle & \langle e_0, 48, 27, 52\rangle & \langle 1, 8, e_0, 47\rangle\\
\langle e_0, 29, 2, 49\rangle & \langle 3, 58, e_0, 18\rangle & \langle f, 15, 46, 56\rangle\\
\langle 1, 5, f, 27\rangle & \langle 0, 6, 60, 48\rangle & \langle 0, 29, 65, 12\rangle\\
\langle 0, 77, 51, 59\rangle & \langle 1, 77, 34, 14\rangle & \langle 1, 32, 17, 48\rangle\\
\langle 1, 41, 74, 29\rangle & \langle 1, 50, 22, 73\rangle & \langle 1, 49, 23, 40\rangle\\
\langle 2, 77, 45, 29\rangle & \langle 2, 11, 67, 10\rangle & \langle 0, 25, 69, 27\rangle\\
\langle 1, 62, 38, 66\rangle & \langle 1, 26, 68, 58\rangle & \langle 1, 4, 10,19\rangle\\
\langle 0, 28, 39, 52\rangle.\\
\end{array}
\end{equation*}
\end{enumerate}
\end{proof}

\begin{proposition}
$A_4(n,6,[2,1,1])$  =  $U(n,6,[2,1,1])$ for each $n
\in \{6, 8, 10, 11, 13, 14, 16, 17, 19, 22, 23, 25, 28, 31, 34,$ $
35, 37, 43, 55, 67, 79, 103\}$.
\end{proposition}
\begin{proof}
For each $n \in \{6, 8, 10, 11, 13, 14, 16, 17, 19, 22,$ $ 23, 25,
28, 31, 34, 35, 37, 43, 55, 67, 79, 103\}$, $\C_n$ is an optimal
$(n,6,[2,1,1])_4$-code with size $U(n,6,[2,1,1])$, where $\C_n$ is the set of cyclic (or quasi-cyclic) shifts of the vectors generated by the following vectors respectively.

\begin{enumerate}

\item $n=6$, $m=1$, $s=1$, $M=2$
\begin{equation*}
\begin{array}{ccc}
P: \langle0, 1, 2, 3\rangle\\
\end{array}
\end{equation*}

\item $n=8$, $m=1$, $s=1$,  $M=1$
\begin{equation*}
\begin{array}{ccc}
P: \langle0, 1, 3, 5\rangle\\
\end{array}
\end{equation*}

\item $n=10$, $m=1$, $s=1$,  $M=2$
\begin{equation*}
\begin{array}{ccc}
P: \langle0, 1, 2, 8\rangle & \langle0, 4, 7, 9\rangle & \langle1, 3, 6, 7\rangle\\
\end{array}
\end{equation*}

\item $n=11$, $m=3$, $s=3$, $M=11$
\begin{equation*}
\begin{array}{ccc}
P: \langle 1, 2, 3, 4\rangle & \langle 1, 5, 6, 7\rangle & \langle 2, 6, 5, 10\rangle\\
R: \langle 2, 8, 9, 7\rangle & \langle 0, 10, 4, 5\rangle & \langle 4, 8, 1, 0\rangle\\
\langle 3, 5, 0, 8\rangle & \langle 8, 10, 6, 3\rangle & \langle 0, 9, 6, 2\rangle\\
\langle 4, 5, 2, 9\rangle\\
\end{array}
\end{equation*}

\item $n=13$, $m=7$, $s=2$, $M=1$
\begin{equation*}
\begin{array}{ccc}
P: \langle 0, 1, 4, 10\rangle\\
\end{array}
\end{equation*}

\item $n=14$, $m=1$, $s=1$,  $M=1$
\begin{equation*}
\begin{array}{ccc}
P: \langle0, 2, 3, 9\rangle & \langle0, 6, 5, 10\rangle\\
\end{array}
\end{equation*}

\item $n=16$, $m=1$, $s=1$,  $M=2$
\begin{equation*}
\begin{array}{ccc}
P: \langle0, 1, 2, 3\rangle & \langle0, 4, 9, 15\rangle & \langle0, 6, 13, 14\rangle\\
\langle1, 5, 10, 12\rangle & \langle1, 11, 9, 14\rangle\\
\end{array}
\end{equation*}

\item $n=17$, $M=17$, the code is given in Table \ref{17,6,[2,1,1]}.

\item $n=19$, $m=5$, $s=3$, $M=1$
\begin{equation*}
\begin{array}{ccc}
P: \langle 0, 1, 8, 12\rangle\\
\end{array}
\end{equation*}

\item $n=22$, $m=3$, $s=3$, $M=2$
\begin{equation*}
\begin{array}{ccc}
P: \langle 0, 13, 1, 10\rangle\\
R: \langle 1, 17, 13, 21\rangle & \langle 1, 4, 2, 15\rangle & \langle 0, 16, 12, 21\rangle\\
\langle 1, 8, 0, 12\rangle\\
\end{array}
\end{equation*}

\item $n=23$, $M=23$, the code  is given in Table \ref{23,6,[2,1,1]}.

\item $n=25$, $m=7$, $s=2$, $M=1$
\begin{equation*}
\begin{array}{ccc}
P: \langle 0, 1, 3, 23\rangle & \langle 0, 8, 20, 13\rangle\\
\end{array}
\end{equation*}

\item $n=28$, $m=5$, $s=5$, $M=2$
\begin{equation*}
\begin{array}{ccc}
P: \langle 0, 1, 3, 6\rangle\\
R: \langle 1, 5, 12, 17\rangle & \langle 0, 24, 14, 17\rangle & \langle 0, 16, 8, 23\rangle\\
\langle 1, 9, 2, 15\rangle\\
\end{array}
\end{equation*}

\item $n=31$, $m=2$, $s=5$, $M=1$
\begin{equation*}
\begin{array}{ccc}
P: \langle 0, 1, 6, 14\rangle\\
\end{array}
\end{equation*}

\item $n=34$, $m=7$, $s=5$, $M=2$
\begin{equation*}
\begin{array}{ccc}
P: \langle 0, 3, 14, 23\rangle\\
R: \langle 0, 32, 10, 24\rangle & \langle 1, 3, 11, 20\rangle & \langle 0, 16, 13, 20\rangle\\
\langle 0, 27, 17, 28\rangle & \langle 1, 19, 7, 22\rangle & \langle1, 16, 15, 31\rangle\\
\end{array}
\end{equation*}

\item $n=35$, $M=35$, the code  is given in Table \ref{35,6,[2,1,1]}.

\item $n=37$, $m=5$, $s=6$, $M=1$
\begin{equation*}
\begin{array}{ccc}
P: \langle 0, 1, 11, 27\rangle\\
\end{array}
\end{equation*}

\item $n=43$, $m=4$, $s=7$, $M=1$
\begin{equation*}
\begin{array}{ccc}
P: \langle 0, 1, 7, 13\rangle\\
\end{array}
\end{equation*}

\item $n=55$, $m=7$, $s=2$, $M=1$
\begin{equation*}
\begin{array}{ccc}
P: \langle 0, 9, 53, 6\rangle & \langle 0, 54, 11, 15\rangle & \langle 0, 4, 40, 25\rangle\\
R: \langle 0, 31, 49, 45\rangle & \langle 0, 20, 39, 23\rangle & \langle 0, 17, 30, 43\rangle\\
\end{array}
\end{equation*}

\item $n=67$, $m=3$, $s=11$, $M=1$
\begin{equation*}
\begin{array}{ccc}
P: \langle 0, 1, 13, 55\rangle\\
\end{array}
\end{equation*}

\item $n=79$, $m=3$, $s=13$, $M=1$
\begin{equation*}
\begin{array}{ccc}
P: \langle 0, 1, 24, 56\rangle\\
\end{array}
\end{equation*}

\item $n=103$, $m=3$, $s=17$, $M=1$
\begin{equation*}
\begin{array}{ccc}
P: \langle 0, 1, 7, 97\rangle.\\
\end{array}
\end{equation*}
\end{enumerate}
\end{proof}
\vskip 10pt

\begin{proposition}
\label{CCC(6)7} $A_4(7,6,[2,1,1])=4$.
\end{proposition}
\begin{proof}
The $4$ required codewords are:
\begin{equation*}
\begin{array}{ccc}
\langle 0, 1, 2, 3\rangle & \langle 0, 4, 5, 6\rangle & \langle 2, 3, 4, 5\rangle\\
\langle 5, 6, 1, 2\rangle.\\
\end{array}
\end{equation*}
\end{proof}
\vskip 10pt
%

\end{document}